\documentclass[
aps,
,pra
,amsmath
,amssymb
,amsfonts
,twocolumn
,tightenlines
,superscriptaddress
,nofootinbib
,longbibliography
]{revtex4-2}
\nonstopmode
\usepackage[dvipsnames]{xcolor}
\usepackage[normalem]{ulem}
\usepackage{graphicx}
\usepackage{dcolumn}
\usepackage{bm}
\usepackage{hyperref}
\hypersetup{colorlinks = true, linkcolor = [rgb]{0.19411,0.51882,0.667058}, urlcolor = [rgb]{0.125490,0.29542,0.1647058}, citecolor = [rgb]{0.75882,0.37411,0.14117}}

\usepackage{setspace}

\usepackage{amssymb,amsthm}
\newtheorem{theorem}{Theorem}
\newtheorem{lemma}[theorem]{Lemma}
\usepackage{braket}

\usepackage[normalem]{ulem}

\graphicspath{{../images/}}
\usepackage[toc,page]{appendix}

\usepackage{braket}

\def\>{\rangle}
\def\<{\langle}

\usepackage{color}

\usepackage{subfigure}
\usepackage{mathtools}
\usepackage{graphicx}
\usepackage{bm}
\usepackage{amsmath}
\usepackage{amssymb}

\newcommand\floor[1]{\left\lfloor#1\right\rfloor}

\allowdisplaybreaks 

\begin{document}

\title{Exact quantum sensing limits for bosonic dephasing channels}

\author{Zixin Huang}
\email{zixin.huang@mq.edu.au}
\affiliation{School of Mathematical and Physical Sciences, Macquarie University, NSW 2109, Australia}
\affiliation{Centre for Quantum Software and Information, Faculty of Engineering and Information Technology, University of Technology Sydney, Sydney, NSW, Australia}

\author{Ludovico Lami}
\email{ludovico.lami@gmail.com}
\affiliation{QuSoft, Science Park 123, 1098 XG Amsterdam, the Netherlands}
\affiliation{Korteweg--de Vries Institute for Mathematics, University of Amsterdam, Science Park 105-107, 1098 XG Amsterdam, the Netherlands}
\affiliation{Institute for Theoretical Physics, University of Amsterdam, Science Park 904, 1098 XH Amsterdam, the Netherlands}

\author{Mark M.~Wilde}
\email{wilde@cornell.edu}
\affiliation{School of Electrical and Computer Engineering, Cornell University, Ithaca, New York 14850, USA}
\date{\today}

\begin{abstract}  
Dephasing is a prominent noise mechanism that afflicts quantum information carriers, and it is one of the main challenges towards realizing useful quantum computation,  communication, and sensing. 
Here we consider discrimination and estimation of bosonic dephasing channels,
when using the most general adaptive strategies allowed by quantum mechanics. 
We reduce these difficult quantum problems to simple classical ones based on the probability densities defining the bosonic dephasing channels.
By doing so, we rigorously establish the optimal performance of various distinguishability and estimation tasks and construct explicit strategies to achieve this performance.
To the best of our knowledge, this is the first example of a non-Gaussian bosonic channel for which there are exact solutions for these tasks.
\end{abstract}

\maketitle

\tableofcontents

\section{Introduction}

An important channel to consider in the context of quantum technologies is the bosonic dephasing channel (BDC). 
{A} single-mode BDC~$\mathcal{D}_p$ is characterized by a probability density function~$p(\phi)${,} where $\phi \in [-\pi,\pi]$ represents the random angle of phase space rotation induced by the channel~\cite{lami2023exact}. Accordingly, the action of $\mathcal{D}_p$ on an input density operator $\rho$ is given by
\begin{align}
\label{eq:bdc}
\mathcal{D}_p(\rho)\coloneqq \int_{-\pi}^{\pi}d\phi\ p(\phi)\, e^{- i\hat{n}\phi}\rho
e^{ i\hat{n}\phi},
\end{align}
where $\hat{n}$ is the photon number operator~\cite{gerry2004introductory}.
Dephasing is a major noise mechanism that afflicts quantum information carriers~\cite{RevModPhys.88.041001}, and it is one of the main challenges towards realizing useful quantum computation, communication, and sensing. In a dephasing noise process, the relative phase information between different photon-number components of a superposed state is lost; for quantum communication, such a process can be understood as arising from, e.g., temperature fluctuations of the environment that stretch or contract the length of a fiber~\cite{W92}. As such, 
{the problem of} understanding the ultimate quantum limits for quantum information tasks using such channels has received {considerable} attention recently~\cite{JC10,PhysRevA.102.042413,Z21,Fanizza2021squeezingenhanced,rexiti2022discrimination,AMM23,lami2023exact,RevModPhys.87.307,google2021exponential}.

Two important tasks for characterizing the capabilities of BDCs are channel discrimination (quantum hypothesis testing) and parameter estimation (quantum metrology).
For hypothesis testing, the task is to distinguish between models describing different physical processes. The most basic setting involves a binary decision, for which the goal is to distinguish between two hypotheses, commonly called the null hypothesis and the alternative hypothesis. Quantum state discrimination is crucial in several applications (e.g., quantum communication~\cite{sidhu2023linear}, astronomical sensing~\cite{PhysRevLett.127.130502,PhysRevA.107.022409}, and spectroscopy~\cite{PhysRevLett.125.180502}), and it has been extensively studied~\cite{bae2015quantum,chefles200412}. 
Quantum channel discrimination, a generalization of state discrimination, has been studied less; however, there is an increasing body of literature on this topic~\cite{CDP08,Duan09,PW09,Hayashi09,Harrow10,MPW10,Cooney2016,PhysRevLett.118.100502,takeoka2016optimal,Puzzuoli2017,wilde2020amortized,WW19,FFRS20,KW21,FF21,bergh2022parallelization,SHW22,bergh2023infinite}. In channel discrimination, the unknown channel is called $n$ times, and the goal is to perform a measurement on the final state to determine which channel was called. This setting is more complex than state discrimination because one can optimize over various strategies in order to make the error probabilities as low as possible. One can either employ an adaptive strategy or a non-adaptive parallel strategy. It is known that adaptive strategies possess advantages over non-adaptive ones in the non-asymptotic regime~\cite{Harrow10,KW21}; however, these advantages come with limits~\cite{bergh2022parallelization,bergh2023infinite}. In the setting of symmetric error, adaptive strategies can also outperform parallel ones in the asymptotic regime~\cite{SHW22}, but they do not in the asymptotic setting of asymmetric error~\cite{WW19,wilde2020amortized,FFRS20,bergh2023infinite}.

Quantum metrology
deals with the optimal estimation of parameters encoded in quantum states and quantum channels, and the typical goal is to minimize the variance of the parameter of interest. Quantum strategies for estimation involve nonclassical effects like  entanglement to achieve precision limits beyond those that are allowed by classical physics. 
The ultimate quantum limit for quantum metrology using adaptive strategies is in general exceedingly difficult to characterize, as the nested optimizations over each step of an adaptive protocol often lead to a mathematically intractable problem.
When estimating parameters encoded into quantum channels, one can again consider parallel and adaptive strategies \cite{ZJ21}, with there being differences in their performance  \cite{LHYY23}.
Previous works have considered bounds for  parameter estimation of unitary channels~\cite{escher2011general,demkowicz2012elusive,PhysRevLett.113.250801}, for teleportation-covariant and Gaussian channels~\cite{PhysRevLett.118.100502,takeoka2016optimal}, and for general channels~\cite{katariya2021geometric}.

\begin{figure}[t]
\includegraphics[width=\linewidth]{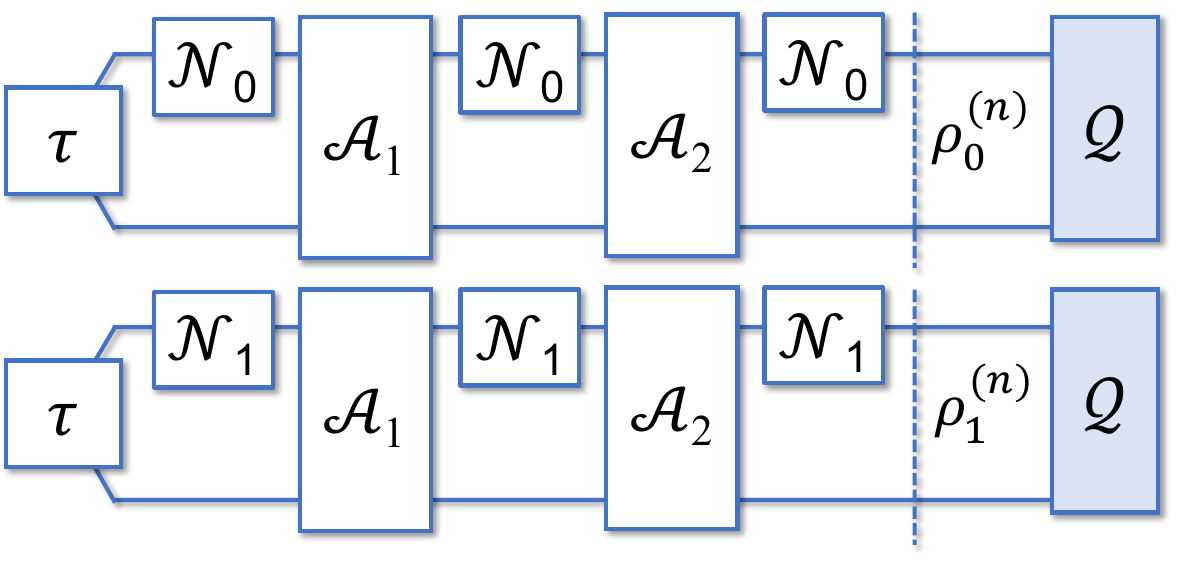}
\caption{\label{fig:scheme} A general, adaptive protocol for channel discrimination and parameter estimation, when either $\mathcal{N}_0$ or $\mathcal{N}_1$ is called three times. The initial input state is $\tau$, the adaptive operations are $\mathcal{A}_1$ and $\mathcal{A}_2$, and the final measurement is~$\mathcal{Q}$. The final states are denoted by $\rho_0^{(n)}$ and $\rho_1^{(n)}$, and $n=3$ in this case.
}
\end{figure}

In this paper, we consider several tasks associated with  BDCs:
\begin{enumerate}
\item channel discrimination of two BDCs in the symmetric error setting of hypothesis testing,
\item channel discrimination of two BDCs in the asymmetric error setting of hypothesis testing,
\item noise parameter estimation of BDCs.
\end{enumerate} 
We consider several other discrimination tasks  and discuss the multimode generalization of our results in some of the appendices. By making a connection to the physics (``environmental state'') that gives rise to the channel processes, for 1., 2., and 3., we quantify the largest distinguishability or estimability that can be realized between BDCs using the most general strategy achievable by adaptive protocols.
We do so by proving optimality bounds and showing their attainability. In the latter case, we provide a fully rigorous proof of the convergence of the guessed probability density to the original one.
To the best of our knowledge, our results here constitute the first example of a class of non-Gaussian bosonic channels for which we have exact solutions for these tasks.

The structure of our paper is as follows. In Sections~\ref{sec:QCD}--\ref{sec:QPE}, we introduce the theoretical frameworks  of quantum channel discrimination and estimation. We present our main results in Section~\ref{sec:results};
the optimality parts of our proofs are given in Section~\ref{sec:optimality}, and the attainability parts 
in Section~\ref{sec:attainability}. We show that in the energy-unconstrained limit 
the corresponding figures of merit match, thus leading to exact solutions for channel discrimination and estimation tasks involving BDCs. In Section~\ref{sec:deph-plus-loss}, we generalize our findings when photon loss is present in addition to dephasing, showing that the same fundamental limits apply. We finally conclude in Section~\ref{sec:conclusion} with a summary and some directions for future research. In Appendix~\ref{app:other-scens}, we discuss a variety of other scenarios to which our results apply, including the strong converse exponent, the error exponent, multiple channel discrimination, and antidistinguishability. In Appendix~\ref{app:multimode}, we briefly 
touch upon how our results generalize to multimode BDCs, and in Appendix~\ref{app:multiparam}, how our findings generalize to multiparameter estimation.

\section{Quantum channel discrimination and estimation}

\subsection{Quantum channel discrimination}

\label{sec:QCD}

The goal of quantum channel discrimination is to distinguish one quantum channel $\mathcal{N}_0$ from another channel~$\mathcal{N}_1$ by calling the unknown channel $n$ times. A strategy for doing so is abbreviated by~$\mathcal{A}$, which denotes the initial state $\tau$ prepared, the $n-1$ adaptive operations between every call to the unknown channel, and the final measurement, denoted by $\mathcal{Q} \coloneqq (Q_0, Q_1)$. The conditions $ Q_0, Q_1 \geq 0$ and $Q_0 + Q_1 = I$ hold, so that $\mathcal{Q}$ is a positive operator-valued measure (POVM). The measurement outcome $Q_0$ corresponds to deciding $\mathcal{N}_0$, and the outcome $Q_1$ corresponds to deciding $\mathcal{N}_1$. An example of such a quantum channel discrimination protocol with $n=3$ is depicted in Figure~\ref{fig:scheme}. The type~I error probability is the probability of deciding $\mathcal{N}_1$ when the actual channel is $\mathcal{N}_0$, and the type~II error probability is the probability of deciding $\mathcal{N}_0$ when the actual channel is $\mathcal{N}_1$. We denote these error probabilities, respectively, as follows:
\begin{equation}
    \alpha_n(\mathcal{A})  \coloneqq \operatorname{Tr}\!\left[Q_1\rho_0^{(n)}\right]  , \qquad
     \beta_n(\mathcal{A})  \coloneqq \operatorname{Tr}\!\left[Q_0\rho_1^{(n)}\right] \label{eq:ch-disc-err-probs},
\end{equation}
where $\rho_0^{(n)}$ denotes the final state of the protocol in case the unknown channel is $\mathcal{N}_0$ and $\rho_1^{(n)}$ denotes the final state of the protocol in case the unknown channel is $\mathcal{N}_1$. Additionally, in the notations $\alpha_n(\mathcal{A})$ and $\beta_n(\mathcal{A})$, we have left the dependence on the channels $\mathcal{N}_0$ and $\mathcal{N}_1$ implicit, but we explicitly denote the dependence on the number~$n$ of calls to the unknown channel and the strategy~$\mathcal{A}$.   
In radar applications, $\alpha_n$ is known as the false alarm probability, and $\beta_n$ is called the missed detection probability~\cite{van2004detection}. For more details of quantum channel discrimination, please refer to~\cite[Section~III-A]{wilde2020amortized}.

The difficulty of quantum channel discrimination arises because most operational quantities of interest involve optimizations of these error probabilities over every possible adaptive strategy $\mathcal{A}$. For the finite-dimensional case, these optimizations can be phrased as semi-definite programs~\cite{KW21}, but their computational complexity grows exponentially in $n$. For the infinite-dimensional case, the resulting optimization problem is an infinite-dimensional semi-definite program and, in general, is even more difficult to solve.

Classical hypothesis testing, on the other hand, is much simpler to handle computationally.
In this scenario, a sample~$\phi$ is selected from a probability density function $p$ or $q$, and the goal is to identify whether $\phi$ came from $p$ or $q$. This setting can be generalized to the multiple sample setting, in which $n$ samples, abbreviated as $\phi^n\equiv (\phi_1, \ldots, \phi_n)$, are selected from $p^{\otimes n}$ or $q^{\otimes n}$. In order to decide whether  $p^{\otimes n}$ or $q^{\otimes n}$ is the underlying density, the most general procedure one can perform on $\phi^n$ is a randomized test $t$~\cite[Section~2.1]{tan2014}, which is the classical version of a POVM (i.e., $t(\phi^n) \in [0,1]$ is the probability of deciding $p^{\otimes n}$ for all $\phi^n$). The error probabilities of classical hypothesis testing are then given by
\begin{align}
    \alpha_n(t) & \coloneqq  1 - \int d\phi^n\,  t(\phi^n) \, p^{\otimes n}(\phi^n), 
    \label{eq:cl-disc-err-prob-a} \\
    \beta_n(t) & \coloneqq   \int d\phi^n\,  t(\phi^n) \, q^{\otimes n}(\phi^n),
    \label{eq:cl-disc-err-prob-b} 
\end{align}
where, in the notation for $\alpha_n(t)$ and $\beta_n(t)$, we have left the dependence on the probability densities $p$ and $q$ implicit.

\subsection{Quantum channel estimation}

\label{sec:QPE}

The goal of quantum channel estimation is to estimate a parameter $\theta$ encoded in a  family $(\mathcal{N}_\theta)_{\theta \in \Theta} $ of quantum channels. 
The main difference with quantum channel discrimination is that $\theta$ is selected from a continuous set $\Theta$, and thus we can only obtain an estimate of it, rather than identify it exactly. However, one can think of channel discrimination as being a special case of channel estimation in which the set $\Theta$ is a finite set consisting of just two elements (i.e., $\Theta = \{0,1\}$ for channel discrimination).

The most general protocol for channel estimation is adaptive, similar to that discussed for channel discrimination. As such, we use the same notation $\mathcal{A}$ to refer to an adaptive strategy for channel estimation. However, the main difference is that the final POVM of a channel estimation protocol outputs an estimate $\hat{\theta}$ of the unknown parameter $\theta$, and thus, the final POVM is of the form $(Q_{\hat{\theta}})_{\hat{\theta}\in \Theta}$, where $Q_{\hat{\theta}} \geq 0$ for all $\hat{\theta}$ and $\int d\hat{\theta} \, Q_{\hat{\theta}}=I$. Denoting the final state of an $n$-shot protocol by~$\rho^{(n)}_\theta$ whenever the underlying channel is $\mathcal{N}_\theta$, the conditional probability density of observing $\hat{\theta}$ is
$ \operatorname{Tr}\!\left[Q_{\hat{\theta}} \rho^{(n)}_\theta\right]$.

To quantify the performance of a channel estimation protocol, we employ a cost function $c(\hat{\theta},\theta)$ that measures the deviation of the estimate $\hat{\theta}$ from the true value~$\theta$. 
The basic properties for such a cost function are as follows \cite[Section~2.1]{korostelev2011mathematical}:
\begin{enumerate}
\item $c(\hat{\theta},\theta)=0$ if $\hat{\theta} = \theta$,
\item $c(\hat{\theta},\theta)=c(\theta,\hat{\theta})$,
\item $c(\hat{\theta}',\theta) \geq c(\hat{\theta},\theta)$ for $\hat{\theta}' \geq \hat{\theta} \geq \theta$,
\item $c(\hat{\theta},\theta)$ is not identically equal to zero,
\item for some constants $k,a>0$, the following bound holds: $c(\hat{\theta},\theta) \leq k(1+|\hat{\theta} - \theta|^a)$ for all $\hat{\theta},\theta$.
\end{enumerate}
Beyond basic properties expected for such a function, we require it to be continuous. 
Common choices when $\Theta = \mathbb{R}$ include the absolute deviation $c(\hat{\theta},\theta) = |\hat{\theta} -\theta|$ and quadratic cost $c(\hat{\theta},\theta) = (\hat{\theta} -\theta)^2$. 
We then define the risk of an $n$-round adaptive strategy~$\mathcal{A}$ to be the expected cost:
\begin{equation}
    r_n(\theta, \mathcal{A})\coloneqq \int d\hat{\theta}\, \operatorname{Tr}\!\left[Q_{\hat{\theta}} \rho^{(n)}_\theta\right] \, c(\hat{\theta},\theta).
\end{equation}
As with channel discrimination, in the notation $r_n(\theta, \mathcal{A})$, we leave the dependence on $\mathcal{N}_\theta$ implicit.
If we choose $c(\hat{\theta},\theta) = (\hat{\theta} -\theta)^2$, the expected cost is the mean squared error or variance, which is the standard error metric for parameter estimation.
The difficulty in channel estimation also lies in the complexity of adaptive strategies that can be considered.

Classical parameter estimation is again simpler. In this scenario, there is a parameterized family $(p_\theta)_{\theta \in \Theta}$ of probability densities. Given access to $n$ independent samples selected from $p_\theta^{\otimes n}$, labeled by $\phi^n$, an  estimate~$\hat{\theta}$ is output according to the conditional probability density~$t(\hat{\theta}|\phi^n)$. Then the conditional probability density for outputting~$\hat{\theta}$  is
\begin{equation}
    s(\hat{\theta}|\theta) \coloneqq \int d \phi^n\,  t(\hat{\theta}|\phi^n)\,  p_\theta^{\otimes n}(\phi^n),
\end{equation}
and thus the risk of a classical strategy $t$ for estimating~$\theta$ is
\begin{equation}
    r_n (\theta, t) \coloneqq \int d\hat{\theta} \, s(\hat{\theta}|\theta) \, c (\hat{\theta}, \theta).
\end{equation}
Here, in the notation $r_n(\theta, t)$, we again leave the dependence  on the underlying probability density $p_\theta$ implicit.

\section{Main results}
\label{sec:results}

One of the main insights of our paper is that, when restricting the underlying channels to BDCs $\mathcal{D}_p$ and $\mathcal{D}_q$ of the form in~\eqref{eq:bdc} and with respective probability densities~$p$ and $q$, various optimized functions of the error probabilities $\alpha_n(\mathcal{A})$ and $\beta_n(\mathcal{A})$, as defined in~\eqref{eq:ch-disc-err-probs}, are equal to the same optimized functions of the classical error probabilities $\alpha_n(t)$ and $\beta_n(t)$, as defined in~\eqref{eq:cl-disc-err-prob-a}--\eqref{eq:cl-disc-err-prob-b}. Our first main result is that the following equality holds for all~$\lambda \in (0,1)$:
\begin{equation}
     \inf_{\mathcal{A}} \big\{ \lambda \alpha_n(\mathcal{A})+ (1-\lambda)\beta_n(\mathcal{A}) \big\}
     = \inf_{t} \big\{\lambda\alpha_n(t) + (1-\lambda) \beta_n(t) \big\}
     \label{eq:symm-err-cl} ,
\end{equation}
where the optimization on the left is over all adaptive strategies and that on the right is over all randomized tests (we adopt this same abbreviated notation  in all related statements that follow). Our second main result is that the following equality holds for all $\varepsilon \in (0,1)$:
\begin{equation}
   \inf_{\mathcal{A}} \left\{ \beta_n(\mathcal{A}): \alpha_n (\mathcal{A}) \leq \varepsilon \right \} =  \inf_{t } \left\{ \beta_n(t): \alpha_n (t) \leq \varepsilon \right \}.
    \label{eq:asymm-err-prob-cl}
\end{equation}
In Appendix~\ref{app:hypo-test-reg-equiv}, we prove that these claims are actually a consequence of the fact that the sets of achievable pairs, $\left\{\left(\alpha_n (\mathcal{A}), \beta_n (\mathcal{A})\right)\right\}_{\mathcal{A}}$ and $\left\{\left(\alpha_n (t), \beta_n (t)\right)\right\}_{t}$, coincide.  

A related insight of our paper concerns parameterized families $(\mathcal{D}_{p_\theta})_{\theta \in \Theta}$ of BDCs, for which the probability density $p_\theta$ underlying $\mathcal{D}_{p_\theta}$ is parameterized. In this case, our third main result is that the following equality holds for an arbitrary risk function for which the underlying cost function is continuous:
\begin{equation}
   \inf_{\mathcal{A}} r_n(\theta, \mathcal{A}) = \inf_t r_n(\theta, t),
   \label{eq:ch-est-main-result}
\end{equation}
where, similar to~\eqref{eq:symm-err-cl},  the optimization on the left is over all adaptive strategies for channel estimation of $\mathcal{D}_{p_\theta}$ and the optimization on the right is over all classical estimation strategies for $p_\theta$.

The equalities in~\eqref{eq:symm-err-cl},~\eqref{eq:asymm-err-prob-cl}, and~\eqref{eq:ch-est-main-result} are some of the main results of our paper, and 
as we will see, they give a complete understanding of the fundamental limits of channel discrimination and parameter estimation for BDCs. 
The equalities in~\eqref{eq:symm-err-cl}--\eqref{eq:asymm-err-prob-cl} thus represent a significant reduction in the difficulty of channel discrimination for BDCs, i.e., reducing it to a classical hypothesis testing problem, for which there is a wealth of knowledge that we can apply. A similar statement applies for~\eqref{eq:ch-est-main-result} and channel estimation of BDCs. In the subsections that follow, we discuss various scenarios of interest in more detail.

\subsection{Symmetric hypothesis testing}

In the symmetric setting of channel discrimination, the goal is to find an optimal strategy that attains the minimum average error probability, defined as
\begin{equation}
\inf_{\mathcal{A}} \big\{ \lambda \alpha_n(\mathcal{A})+ (1-\lambda)\beta_n(\mathcal{A}) \big\},
\label{eq:symm-err}
\end{equation}
where $\lambda \in (0,1)$ is the prior probability that channel $\mathcal{N}_0$ is selected and $\alpha_n(\mathcal{A})$ and $\beta_n(\mathcal{A})$ are defined in~\eqref{eq:ch-disc-err-probs}. Our result here establishes, for  BDCs $\mathcal{D}_p$ and $\mathcal{D}_q$, that the equality in~\eqref{eq:symm-err-cl} holds. As a consequence of a well known result from statistics (see, e.g.,~\cite[Lemma~1.4]{canonne2022topics}), we have the following explicit form for the error probability:
\begin{equation} \begin{aligned}
&\inf_{t } \big\{ \lambda\alpha_n(t) + (1-\lambda) \beta_n(t)\big\} \\
&\qquad = \frac{1}{2}\left(1 - \left\Vert \lambda p^{\otimes n} - (1\!-\!\lambda)q^{\otimes n}\right \Vert_1 \right),
\end{aligned} \end{equation}
where $\left\Vert f \right \Vert_1 \coloneqq \int d\phi^n \, |f(\phi^n)| $ is the 
$\ell_1$-norm of a function~$f$.

The non-asymptotic error exponent for channel discrimination in the symmetric setting is defined as
\begin{equation}
C_n(\lambda,\mathcal{N}_0,\mathcal{N}_1)\coloneqq  
\sup_{\mathcal{A}} \left\{-\frac{1}{n}\ln \left( \lambda \alpha_n(\mathcal{A})+ (1-\lambda)\beta_n(\mathcal{A})\right) \right \}.
\label{eq:non-asym-sym-exp}
\end{equation}
By appealing to the Chernoff theorem from probability theory~\cite{chernoff1952measure} and defining the Chernoff divergence of two probability densities $p$ and $q$ as
\begin{equation}
    C(p\Vert q)  \coloneqq -\ln \inf_{s \in [0,1]} \int_{-\pi}^{\pi} d\phi \, p(\phi)^s q(\phi)^{1-s},
    \label{eq:Chernoff-div}
\end{equation}
our result implies the following simple expression for the asymptotic error exponent of BDCs $\mathcal{D}_p$ and $\mathcal{D}_q$:
\begin{equation}
    \lim_{n \to \infty} C_n(\lambda,\mathcal{D}_p,\mathcal{D}_q)  = C(p\Vert q) .
\end{equation}
This is because
\begin{align}
    & \lim_{n \to \infty} C_n(\lambda,\mathcal{D}_p,\mathcal{D}_q)\notag \\
    &  = \lim_{n \to \infty} \sup_{t } \left\{-\frac{1}{n} \ln \left( \lambda\alpha_n(t) + (1-\lambda) \beta_n(t)\right) \right\}\\
    & = C(p\Vert q),
\end{align}
where the first equality follows  from \eqref{eq:symm-err-cl} and~\eqref{eq:non-asym-sym-exp} and the second from Chernoff's theorem.

\subsection{Asymmetric hypothesis testing}

\label{ssec:asym}

In the asymmetric setting of channel discrimination, the goal is to minimize the type~II error probability subject to a constraint on the type~I error probability. One example is in radar applications or quantum illumination~\cite{PhysRevA.71.062340,lloyd2008enhanced,PhysRevLett.101.253601}, where one is willing to tolerate a certain rate of false alarms but then desires to minimize the chances of missed detections~\cite{WTLB17}. Indeed,  the asymmetric error setting is the right setting to focus on for such applications and leads to a curve known as the receiver operating characteristic \cite{van2004detection}.

More formally, for $\varepsilon \in (0,1)$,
the goal is to optimize over every adaptive strategy $\mathcal{A}$ in order to minimize the type~II error probability $\beta_n$:
\begin{equation}
    \inf_{\mathcal{A}} \left\{ \beta_n(\mathcal{A}): \alpha_n (\mathcal{A}) \leq \varepsilon \right \},
    \label{eq:asymm-err-prob}
\end{equation}
where $\alpha_n(\mathcal{A})$ and $\beta_n(\mathcal{A})$ are defined in~\eqref{eq:ch-disc-err-probs}.
Our result here is that, for BDCs $\mathcal{D}_p$ and $\mathcal{D}_q$, the equality in~\eqref{eq:asymm-err-prob-cl} holds,
where $\alpha_n(t)$ and $\beta_n(t)$ are defined in~\eqref{eq:cl-disc-err-prob-a}--\eqref{eq:cl-disc-err-prob-b}.

The non-asymptotic error exponent for channel discrimination in the asymmetric setting is defined as
\begin{equation}
Z_n(\varepsilon,\mathcal{N}_0,\mathcal{N}_1) \coloneqq  
\sup_{\mathcal{A}} \left\{-\frac{1}{n}\ln \beta_n(\mathcal{A}): \alpha_n (\mathcal{A}) \leq \varepsilon \right\}.
\end{equation}
By appealing to Strassen's theorem~\cite[Theorem~3.1]{strassen1962asymptotische}, itself a refinement of the classical Stein's lemma~\cite{stein_unpublished,chernoff_1956}, our result implies the following expansion of~${Z_n}(\varepsilon,\mathcal{D}_p,\mathcal{D}_q)$:
\begin{multline}
{Z_n}(\varepsilon,\mathcal{D}_p,\mathcal{D}_q) = \\ D(p\Vert q) + \sqrt{\frac{V(p\Vert q)}{n}}\,\Phi^{-1}(\varepsilon) 
    + \frac{\ln n}{2n} + O\!\left(\frac{1}{n}\right). \label{eq:2nd-order-expansion}
\end{multline}
(See also~\cite[Proposition~2.3]{tan2014}.)
In the above, the relative entropy, the relative entropy variance, and the inverse cumulative distribution function of a standard normal random variable are respectively defined as
\begin{align}
    D(p\Vert q) & \coloneqq \int_{-\pi}^{\pi} d\phi\,  p(\phi) \ln \!\left(\frac{p(\phi)}{q(\phi)}\right) , \\
     V(p\Vert q) & \coloneqq \int_{-\pi}^{\pi} d\phi\,  p(\phi)
     \left[ \ln \!\left(\frac{p(\phi)}{q(\phi)}\right) - D(p\Vert q)\right]^2 ,\\
     \Phi^{-1}(\varepsilon) & \coloneqq \sup\!\left\{
a\in\mathbb{R}:\Phi(a)\leq\varepsilon\right\},
\end{align}
where
\begin{equation}
    \Phi(a)\coloneqq \frac{1}{\sqrt{2\pi}}\int_{-\infty}^{a}dx\,\exp\!\left(
-x^{2}/2\right).
\end{equation}
The equality in~\eqref{eq:2nd-order-expansion} follows from the equality in~\eqref{eq:asymm-err-prob-cl} and from the following expansion, a direct application of Strassen's theorem:
\begin{multline}
    \sup_{t } \left\{-\frac{1}{n}\ln \beta_n(t): \alpha_n (t) \leq \varepsilon \right\} = \\
    D(p\Vert q) + \sqrt{\frac{V(p\Vert q)}{n}}\, \Phi^{-1}(\varepsilon) 
    + \frac{\ln n}{2n} + O\!\left(\frac{1}{n}\right).
\end{multline}
Note that $\Phi^{-1}(\varepsilon) < 0 $ for $\varepsilon < 1/2$ and $\Phi^{-1}(\varepsilon) > 0 $ for $\varepsilon > 1/2$. As such, by inspecting~\eqref{eq:2nd-order-expansion}, we see that, as a function of the number~$n$ of channel uses, 
the error exponent ${Z_n}(\varepsilon,\mathcal{D}_p,\mathcal{D}_q)$ approaches the optimal asymptotic value $D(p\Vert q)$ from below for $\varepsilon < 1/2$ and from above for $\varepsilon > 1/2$, at a speed determined by the relative entropy variance $V(p\Vert q)$.
The regime of practical interest occurs when $\varepsilon <1/2$, for which the error exponent ${Z_n}(\varepsilon,\mathcal{D}_p,\mathcal{D}_q)$ thus approaches $D(p\Vert q)$ from below.

\subsection{Quantum metrology}

In the setting of quantum metrology (an umbrella term containing quantum channel estimation), the goal is to minimize the risk over all possible adaptive strategies. Our main result for channel estimation, which applies to a continuous family $(\mathcal{D}_{p_\theta})_\theta$ of BDCs, states that the equality in~\eqref{eq:ch-est-main-result} holds. In the case that the underlying cost function is the quadratic cost function and restricting the optimization to be over adaptive strategies that lead to an unbiased estimator, a consequence of our finding and the classical Cram\'{e}r--Rao bound is the following inequality:
\begin{equation}
    \inf_{\mathcal{A}} r_n(\theta, \mathcal{A}) \geq \frac{1}{n F(\theta)},
    \label{eq:fisher-info-ineq}
\end{equation}
where the Fisher information of the parameterized family~$(p_\theta)_\theta$ is defined as follows:
\begin{equation}
\label{eq:FI}
    F(\theta) \coloneqq \int_{-\pi}^{\pi} d\phi ~p_\theta(\phi)  \left(
    \frac{d}{d\theta}\, \ln p_\theta(\phi)\right)^2.
\end{equation}
The inequality in~\eqref{eq:fisher-info-ineq} follows directly from~\eqref{eq:ch-est-main-result} and the well known Cram\'er--Rao bound $ \inf_{t} r_n(\theta, t) \geq [n F(\theta)]^{-1}$~\cite[Corollary~1.9]{korostelev2011mathematical}.

\section{Examples of bosonic dephasing channels}

To gain some intuition about our findings, let us consider 
the specific examples of BDCs studied in~\cite{lami2023exact}. We plot the Chernoff divergence, the relative entropy, and the Fisher information of certain instances of these channels, which are the main quantities of interest in the asymptotic settings of symmetric channel discrimination, asymmetric channel discrimination, and channel estimation, respectively.

As stressed in~\cite{lami2023exact} and previous works~\cite{JC10,PhysRevA.102.042413}, perhaps the most important class of BDCs are those resulting from setting the probability density $p(\phi)$ in~\eqref{eq:bdc} to be the wrapped normal distribution:
\begin{equation}
    p_{\gamma}(\phi) \coloneqq \frac{1}{\sqrt{2\pi\gamma}}\sum_{k=-\infty}^{+\infty} e^{-\frac{1}{2\gamma} (\phi+2\pi k)^2} ,
\end{equation}
where $\gamma > 0 $ is the variance. This probability density  results from picking $\phi$ according to a mean-zero normal distribution of variance $\gamma$, but then outputting a value in $[-\pi, \pi]$ modulo $2\pi$. Physically, as considered in~\cite{JC10,PhysRevA.102.042413}, it corresponds to interacting the channel input mode with an environmental mode prepared in the vacuum state, according to the Hamiltonian $\hat{n} \otimes (\hat{e} + \hat{e}^\dag)$, where $\hat{e}$ is the annihilation operator for the environmental mode. It can alternatively be realized in terms of Lindbladian evolution for a  time $\gamma$ according to the single Lindblad operator $\hat{n}$.

Another probability density of interest for the BDC is based on the von Mises distribution:
\begin{align}
p_\lambda(\phi) \coloneqq  \frac{e^{\cos(\phi)/\lambda}}{2\pi I_0(1/\lambda)},
\end{align}
where $I_n$ denotes a modified Bessel function of the first kind. The parameter $\lambda$ determines the spread of the distribution, analogous to $\gamma$ for the wrapped normal. For $\lambda  \to
\infty$, it converges to the uniform density, while it becomes highly peaked at zero in the limit $\lambda \to 0$. 

The final circular distribution that we consider is the wrapped Cauchy distribution, given by
\begin{align}
p_{\kappa}(\phi) \coloneqq \frac{1}{2\pi} \frac{\sinh \sqrt{\kappa}}{\cosh\sqrt \kappa - \cos (\phi)}.
\end{align}
The parameter $\kappa > 0$ again determines the spread of the distribution. 

Figure~\ref{fig:chernoff} plots the Chernoff divergence of a pair of BDCs for each kind of distribution, with one spread parameter fixed at a value of $\gamma = \lambda = \kappa = 1$ and the other spread parameter varied. Figure~\ref{fig:relative_entropies} does the same for the relative entropy. Figure~\ref{fig:FI} plots the Fisher information of the underlying channel parameter $\gamma$, $\lambda$, or $\kappa$. We find similar qualitative behavior for all three kinds of probability densities. 

\begin{figure}
\includegraphics[width=\linewidth]{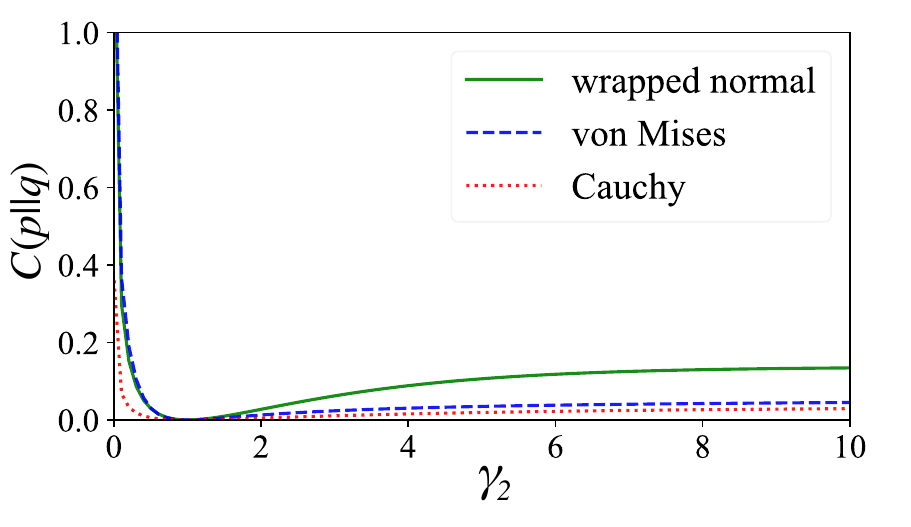}
\caption{\label{fig:chernoff}
The Chernoff exponent of different BDCs as a function of $\gamma_2 = \lambda_2 = \kappa_2$, for $\gamma_1 = \lambda_1 = \kappa_1 = 1$.}
\end{figure}

\begin{figure}
\includegraphics[width=\linewidth]{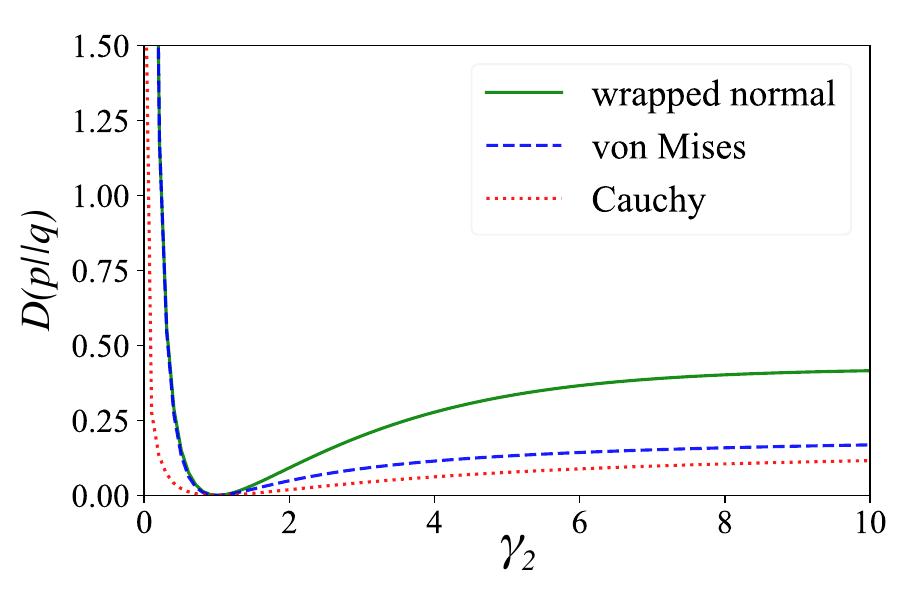}
\caption{\label{fig:relative_entropies}
The relative entropy of different BDCs as a function of $\gamma_2 = \lambda_2 = \kappa_2$, for $\gamma_1 = \lambda_1 = \kappa_1 = 1$.}
\end{figure}

\begin{figure}
\includegraphics[width=\linewidth]{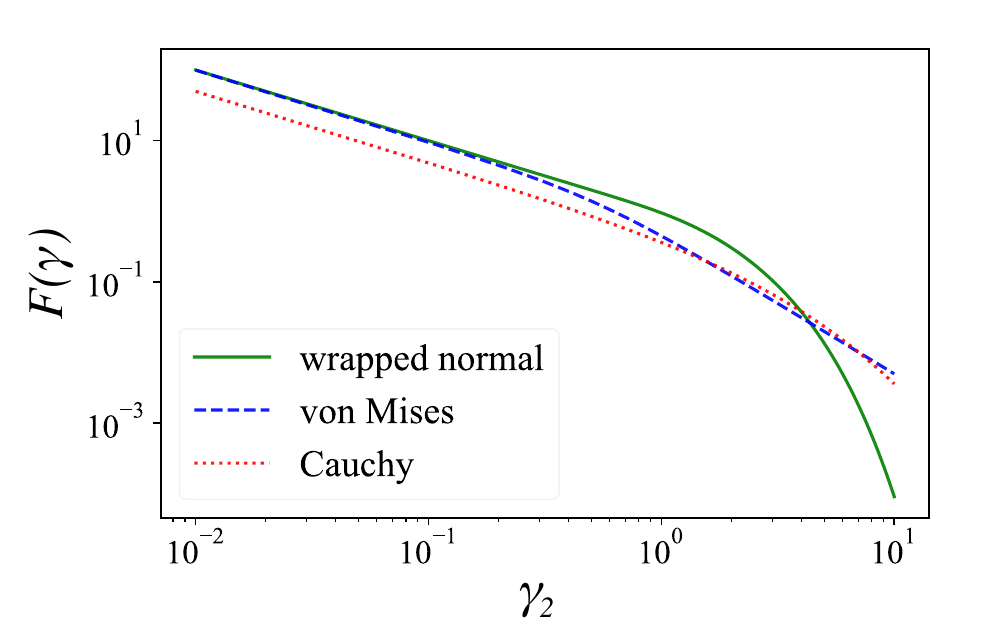}
\caption{\label{fig:FI}
The Fisher information  of  different BDCs as a function of $\gamma = \kappa = \lambda$ shown as a log-log plot.  The Fisher information quickly drops to zero as $\gamma$ increases.}
\end{figure}

\section{Optimality}

\label{sec:optimality}

In this section, we prove one side of the equalities in~\eqref{eq:symm-err-cl},  \eqref{eq:asymm-err-prob-cl}, and~\eqref{eq:ch-est-main-result}  (called the ``optimality'' part), based on a simple observation about all BDCs of the form in~\eqref{eq:bdc}. Namely, they can be simulated by the method discussed in~\cite[Section~3.3]{matsumoto2010metric}, that is, by means of adjoining a parameterized environment state followed by the action of an unparameterized channel.  
After~\cite{matsumoto2010metric} appeared, similar observations were made for other channels in several subsequent works, including~\cite{PhysRevLett.113.250801} and~\cite{PhysRevLett.118.100502,takeoka2016optimal,wilde2020amortized}, and here our contribution is to make a similar observation for BDCs. Namely, all BDCs can be simulated by composing the following two processes:
\begin{enumerate}
    \item A classical background phase~$\phi$ is chosen randomly according to the probability density $p(\phi)$ in~\eqref{eq:bdc}.
    
    \item The input system has the phase operator $e^{- i\hat{n}\phi}$ applied to it, based on the value of $\phi$ chosen, and the value~$\phi$ is subsequently discarded.
\end{enumerate}

\begin{figure}[t]
\includegraphics[width=\linewidth]{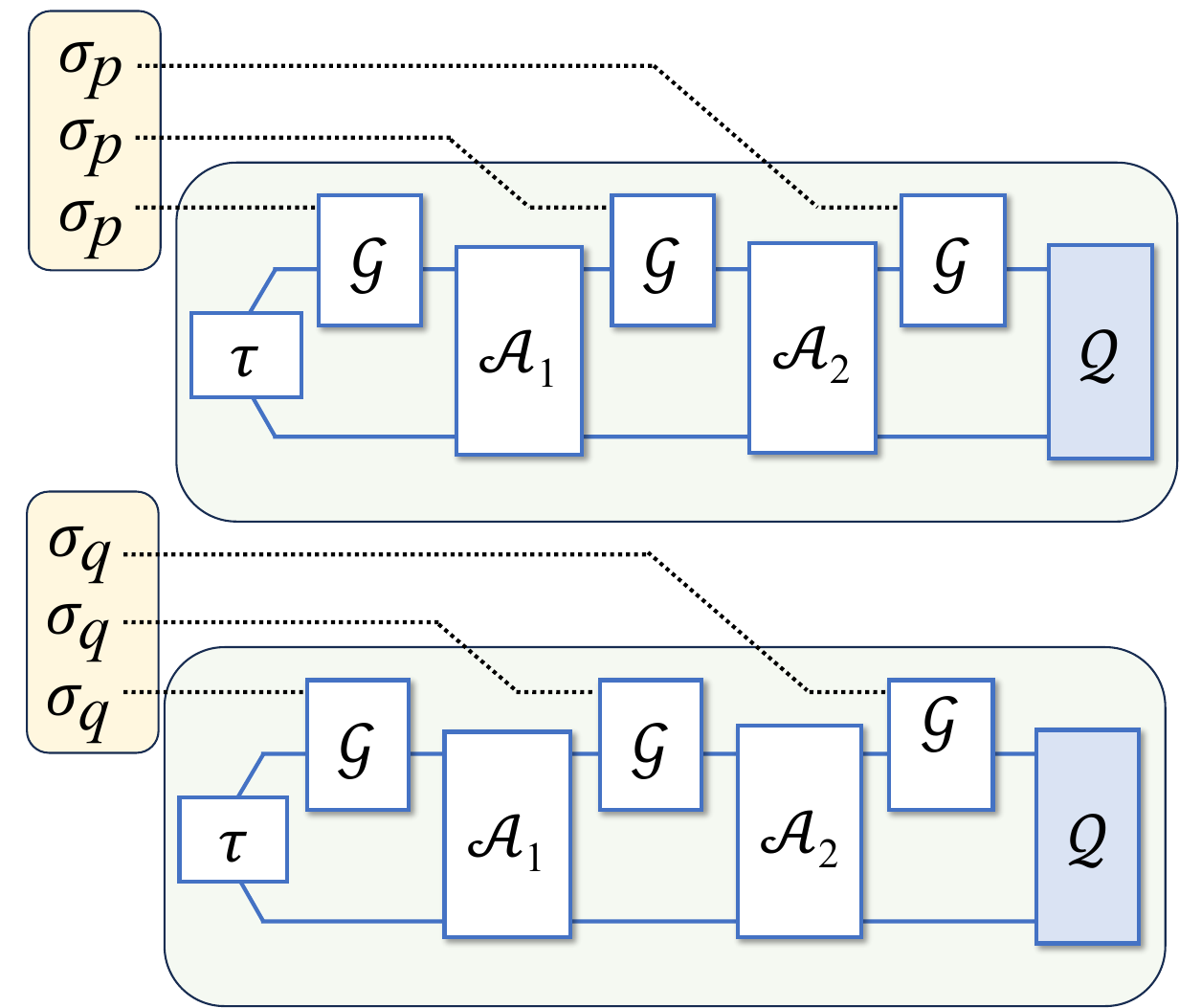}
\caption{\label{fig:env}Channel discrimination and parameter estimation for environment-parameterized channels $\mathcal{D}_p$ and $\mathcal{D}_q$, where the underlying environment states are
$\sigma_p$ and $\sigma_q$, respectively.
The yellow-shaded boxes denote the underlying environmental states to which we do not have access.}
\end{figure}

More formally, first we define an environment state $\sigma_p$ that encodes the probability density $p(\phi)$ in~\eqref{eq:bdc} as follows:
\begin{equation}
\sigma_p \coloneqq  \int_{-\pi}^{\pi} d\phi ~p(\phi)|\phi\rangle\!\langle\phi|,
\label{eq:seed_state}
\end{equation}
where, in the physics literature, $\{|\phi\rangle\}_{\phi}$ is usually interpreted as a set of `eigenkets' obeying the `orthogonality relation' $\langle \phi' | \phi \rangle = \delta(\phi - \phi')$ (i.e., $\{|\phi\rangle\}_{\phi}$ can be 
seen as 
an orthogonal basis for the phase $\phi$). 
We may also interpret~\eqref{eq:seed_state} as a representation of a random variable on $[-\pi,\pi]$ with probability density $p$. Then $\mathcal{D}_p$ decomposes as
\begin{equation}
    \mathcal{D}_p = \mathcal{G} \circ \mathcal{F}_p ,
    \label{eq:bdc-composition}
\end{equation}
where
\begin{align}
\mathcal{F}_p(\rho) & \coloneqq \rho \otimes  \sigma_p , \\ 
\mathcal{G}(\rho' \otimes \rho'') & \coloneqq \int_{-\pi}^{\pi}
d\phi' \, e^{- i \hat{n} \phi'} \rho' e^{ i \hat{n} \phi'} \operatorname{Tr}[| \phi'\rangle\!\langle \phi' | \rho'' ].
\end{align}
The first channel $\mathcal{F}_p$ appends the environment state $\sigma_p$  to the input state $\rho$, while $ \mathcal{G}$ measures $\sigma_p$ and, based on the measured phase $\phi$, applies the unitary phase operator $e^{- i \hat n \phi}$  to $\rho$.
The action of $\mathcal{D}_p$ on an arbitrary input state~$\rho$ is thus as follows:
\begin{equation}
    \mathcal{D}_p(\rho) = \mathcal{G}(\rho \otimes  \sigma_p).
    \label{eq:deph-ch-decompose}
\end{equation}

The implications of the composition in~\eqref{eq:bdc-composition} are far reaching, indeed leading to our optimality bounds.
The idea is that when we decompose the channel this way, we can ``pull back'' the environmental state from our analysis of an adaptive strategy $\mathcal{A}$, as depicted in Figure~\ref{fig:env}. Then, a quantum channel discrimination or estimation task can be recast as a classical state discrimination or estimation task, respectively. Given that the operations in the adaptive strategy $\mathcal{A}$ composed with $n$ instances of the channel~$\mathcal{G}$ are independent of $p$ and $q$,  the optimality of the distinguishability task is then limited by the distinguishability of the environmental states. More formally, every $n$-round adaptive strategy $\mathcal{A}$ for channel discrimination, when applied to the BDC~$\mathcal{D}_p$, can be composed with $n$ calls of the $p$-independent channel~$\mathcal{G}$ to view the resulting strategy as a particular classical test~$t$ performed on the probability density~$p^{\otimes n}$. 
This reasoning then implies the following inequality, which holds for every adaptive strategy~$\mathcal{A}$ applied to BDCs~$\mathcal{D}_p$ and~$\mathcal{D}_q$:
\begin{align}
 \lambda \alpha_n(\mathcal{A})+ (1-\lambda)\beta_n(\mathcal{A}) 
     & \geq  \inf_{t } \big\{ \lambda\alpha_n(t) + (1-\lambda) \beta_n(t) \big\}.
     \label{eq:classical-sym-err-lower}
\end{align}
Since this inequality holds for every adaptive strategy $\mathcal{A}$, we conclude the inequality ``$\geq$'' in~\eqref{eq:symm-err-cl} for BDCs $\mathcal{D}_p$ and $\mathcal{D}_q$. 
Furthermore, for every adaptive strategy $\mathcal{A}$ satisfying $\alpha_n (\mathcal{A}) \leq \varepsilon$, we conclude, by the same reasoning, for BDCs $\mathcal{D}_p$ and $\mathcal{D}_q$ that
\begin{equation}
     \beta_n(\mathcal{A}) \geq   \inf_{t } \left\{ \beta_n(t): \alpha_n (t) \leq \varepsilon \right \}.
     \label{eq:classical-asym-err-lower}
\end{equation}
Since this inequality holds for every adaptive strategy $\mathcal{A}$ satisfying $\alpha_n (\mathcal{A}) \leq \varepsilon$, we conclude the inequality ``$\geq$'' in~\eqref{eq:asymm-err-prob-cl} for BDCs $\mathcal{D}_p$ and $\mathcal{D}_q$.

The same reasoning applies for channel estimation, with respect to a continuous family $(\mathcal{D}_{p_\theta})_\theta$ of parameterized BDCs. Indeed, every $n$-round adaptive strategy~$\mathcal{A}$ for channel estimation can be composed with $n$ calls of the $\theta$-independent channel $\mathcal{L}$ to view the resulting strategy as a particular randomized estimator $t$ performed on the probability density $p_\theta^{\otimes n}$. We then conclude the following inequality for every continuous family $(\mathcal{D}_{p_\theta})_\theta$ of parameterized BDCs:
\begin{equation}
     r_n(\theta, \mathcal{A}) \geq  \inf_t r_n(\theta, t).
\end{equation}
Since this inequality holds for every adaptive strategy $\mathcal{A}$, we conclude the inequality ``$\geq$'' in~\eqref{eq:ch-est-main-result} for every continuous family $(\mathcal{D}_{p_\theta})_\theta$ of parameterized BDCs.

\section{Attainability}
\label{sec:attainability}

Now we prove the other side (``attainability'') of the equalities in~\eqref{eq:symm-err-cl}, \eqref{eq:asymm-err-prob-cl}, and~\eqref{eq:ch-est-main-result}. Again here, the basic principle behind our reasoning is simple. As we will show, for a BDC~$\mathcal{D}_p$, it is possible to input a sequence $(\rho_\nu)_{\nu \in \mathbb{N}}$ of states to it and perform a POVM $(M_\phi)_{\phi}$ such that, for all $\phi\in[-\pi,\pi]$,
\begin{equation}
    p(\phi) = \lim_{\nu \to \infty} p_\nu(\phi) ,
    \label{eq:att-limit-bdc}
\end{equation}
where
\begin{equation}
    p_\nu(\phi) \coloneqq \operatorname{Tr}[M_\phi \mathcal{D}_p(\rho_\nu)].
    \label{eq:p_nu-dens}
\end{equation}
In the formulation above, we have used $\nu$ as an abstract index for a sequence of states. In Sections~\ref{sec:phot-num-sup-scheme}--\ref{sec:coh-state-scheme}, we provide concrete examples in which $\nu$ is replaced by photon number or used as an index for a sequence of coherent states with increasing energy.
A channel satisfying~\eqref{eq:att-limit-bdc}--\eqref{eq:p_nu-dens} is said to be environment seizable~\cite[Definition~36]{wilde2020amortized} because it is possible to perform pre- and post-processing of the channel in order to ``seize" the background environment state.
In this case, we can recover the probability density $p(\phi)$, characterizing a BDC $\mathcal{D}_p$, exactly in the $\nu \to \infty$ limit and process it directly. It is similarly possible to do this for all  $n\in \mathbb{N}$ because, for all $\phi^n \in [-\pi,\pi]^n$,
\begin{equation}
    p^{\otimes n}(\phi^n) = \lim_{\nu \to \infty} p^{\otimes n}_\nu(\phi^n),
\end{equation}
as a direct consequence of~\eqref{eq:att-limit-bdc}. 
Thus, a particular sequence of strategies for channel discrimination is to input the state~$\rho_\nu$ to every channel use, followed by the measurement~$M_\phi$, leading to the density~$p^{\otimes n}_\nu(\phi)$. We then process the resulting densities with a classical test $t$. As we will see shortly, such a sequence of strategies is optimal in the limit $\nu \to \infty$.

In the case that~\eqref{eq:att-limit-bdc} holds, it directly follows that the type~I and type~II error probabilities under an arbitrary test $t$ obey the following equalities:
\begin{equation}
    \alpha_n(t)  = \lim_{\nu \to \infty}  \alpha_{n,\nu}(t), \label{eq:type-i-ii-limit} \qquad 
    \beta_n(t)  = \lim_{\nu \to \infty}  \beta_{n,\nu}(t),
\end{equation}
where
\begin{align}
    \alpha_{n,\nu}(t) & \coloneqq  1 - \int d\phi^n\,  t(\phi^n)\, p_\nu^{\otimes n}(\phi^n), \\
    \beta_{n,\nu}(t) & \coloneqq   \int d\phi^n\,  t(\phi^n) \, q_\nu^{\otimes n}(\phi^n),
\end{align}
and with $q_\nu(\phi)$ defined as in~\eqref{eq:p_nu-dens}, but with $\mathcal{D}_p$ replaced by $\mathcal{D}_q$.

As a consequence of~\eqref{eq:type-i-ii-limit}, if we can show that there exists a sequence $(\rho_\nu)_{\nu \in \mathbb{N}}$ of states and a measurement $M_\phi$ such that~\eqref{eq:att-limit-bdc} holds, then the desired attainability claims hold because the aforementioned strategy is a particular kind of adaptive strategy~$\mathcal{A}$; that is, for every test~$t$ and for every test $t'$ such that $\alpha_n (t') \leq \varepsilon$,
\begin{equation}
\begin{split}
\inf_{\mathcal{A}} \big\{\lambda \alpha_n(\mathcal{A})+ (1-\lambda)\beta_n(\mathcal{A}) \big\}
     & \leq \lambda\alpha_n(t) + (1-\lambda) \beta_n(t) , \\
     \inf_{\mathcal{A}} \left\{ \beta_n(\mathcal{A}): \alpha_n (\mathcal{A}) \leq \varepsilon \right \}  & \leq  \beta_n(t').
     \end{split}
     \label{eq:attainability-final-step-pf}
\end{equation}
Since the inequalities hold for every test $t$ and for every test $t'$ such that $\alpha_n (t') \leq \varepsilon$, we conclude that the same inequalities hold with infima taken on the right-hand side. Combining this claim with the optimality results from the previous section 
concludes the proof of the desired equalities in~\eqref{eq:symm-err-cl} and~\eqref{eq:asymm-err-prob-cl}.

We can make similar conclusions for channel estimation for a continuous family $(\mathcal{D}_{p_\theta})_\theta$ of parameterized BDCs. Indeed, in the case that~\eqref{eq:att-limit-bdc} holds and the estimator $t$ satisfies 
\begin{equation}
\int d\hat{\theta} \, \int d \phi^n\,  t(\hat{\theta}|\phi^n)\, c(\hat{\theta}, \theta) < \infty, 
\end{equation}
a simple application of Lebesgue's dominated convergence theorem shows that
\begin{equation}
r_n(\theta, t) = \lim_{\nu \to \infty} r_{n,\nu}(\theta, t),
\label{eq:risk-convergence}
\end{equation}
where
\begin{align}
    r_{n,\nu}(\theta, t) & \coloneqq \int d\hat{\theta} \, \int d \phi^n\,  t(\hat{\theta}|\phi^n)\,  p_{\theta,\nu}^{\otimes n}(\phi^n) \, c(\hat{\theta}, \theta), \label{eq:risk-limit} \\
    p_{\theta,\nu}(\phi) & \coloneqq \operatorname{Tr}[M_\phi \mathcal{D}_{p_\theta}(\rho_\nu)].
\end{align}
As a consequence of~\eqref{eq:risk-limit}, and similar reasoning used above in~\eqref{eq:attainability-final-step-pf}, we conclude the following attainability inequality:
\begin{equation}
\inf_{\mathcal{A}} r_n(\theta, \mathcal{A})
      \leq \inf_{t} r_n(\theta, t) ,
\end{equation}
which thus finishes the proof of our main channel estimation result in~\eqref{eq:ch-est-main-result}.

In the two subsections that follow, we exhibit two specific schemes for which the needed equality in~\eqref{eq:att-limit-bdc} holds. Moreover, the methods are simple to describe in physical terms, involving either 1) preparation of a uniform superposition of photon-number states at the input and a quantum Fourier transform followed by photon-number measurement and classical post-processing at the output or 2) preparation of a coherent state at the input and heterodyne detection followed by classical post-processing at the output. See Figure~\ref{fig:attainability} for a visual depiction of the two methods. The latter method  is robust to loss in the channel in addition to dephasing, due to the fact that the purity of coherent states is retained under a pure-loss channel (see Section~\ref{sec:deph-plus-loss} for more discussions of this point). The first scheme is similar to that introduced in~\cite{HB93}, and the measurement used in the first scheme can be considered an approximation of the canonical phase measurement, also discussed in \cite{Wiseman_1998}. The second scheme has been considered in \cite{Wiseman_1998}. 

\begin{figure}
    \centering
   \includegraphics[width=\linewidth]{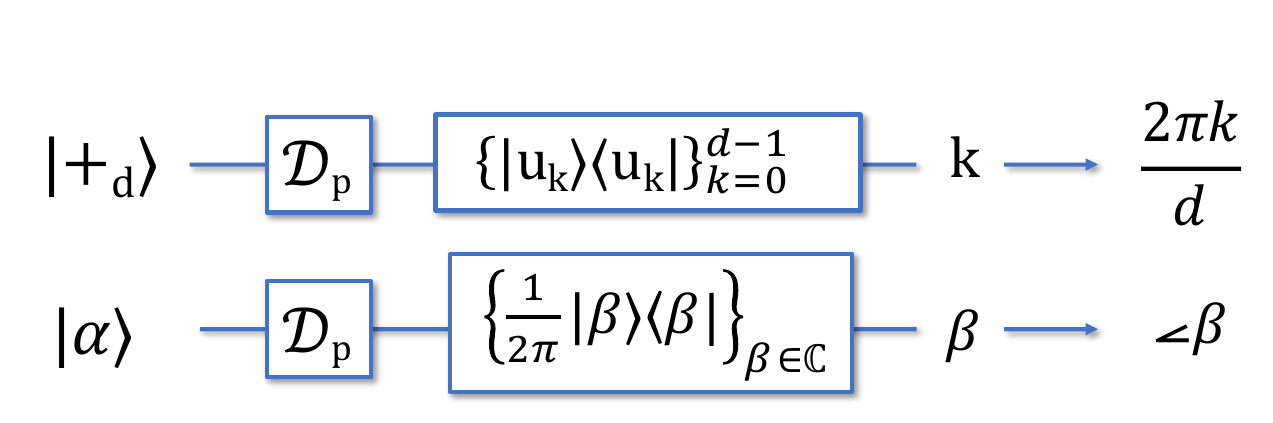}
    \caption{Two measurement methods that achieve the optimality bounds. The method on the top (see Section~\ref{sec:phot-num-sup-scheme}) involves preparing a uniform superposition of photon-number states, transmitting through the BDC $\mathcal{D}_{p}$, performing a Fourier transform, followed by photodetection and classical post-processing. The method on the bottom (see Section~\ref{sec:coh-state-scheme}) involves preparing a coherent state $|\alpha\rangle$, where $\alpha \in \mathbb{R}_+$, transmitting through the BDC $\mathcal{D}_{p}$, performing heterodyne detection, followed by classical post-processing. }
    \label{fig:attainability}
\end{figure}

\subsection{Photon-number-superposition method}

\label{sec:phot-num-sup-scheme}

As stated above, this method involves preparing a uniform superposition of photon-number states at the input and performing a quantum Fourier transform, followed by photon detection and classical post-processing at the output.
Photon-number superposition states have been 
well investigated in the context of optical phase estimation (see  \cite{demkowicz2015quantum,PhysRevA.90.023856,rodriguez2022determination,cooper2013experimental,dowling2008quantum}). The scheme we consider below is quite similar to that proposed in  \cite{HB93}.

Let us begin by defining a $d$-level, uniform  superposition of photon-number states:
\begin{align}
\label{eq:+_d}
|+_{d}\rangle\coloneqq \frac{1}{\sqrt{d}}\sum_{n=0}^{d-1}|n\rangle,
\end{align}
where $|n\rangle$ is a photon-number state~\cite{gerry2004introductory}. A property of $\ket{+_d}$ is that phases become encoded into it as follows:
\begin{align}
e^{- i\hat{n}\phi}|+_{d}\rangle=\frac{1}{\sqrt{d}}\sum_{n=0}^{d-1}e^{-
in\phi}|n\rangle.
\end{align}
Such encoded phases can be recovered approximately by performing a measurement in the Fourier basis, which is defined for all $k\in\left\{  0,\ldots,d-1\right\}  $
as
\begin{align}
\label{eq:fourier-basis-states}
|u_{k}\rangle\coloneqq \frac{1}{\sqrt{d}}\sum_{n=0}^{d-1}e^{-2\pi ikn/d}|n\rangle.
\end{align}
Indeed, as shown in Appendix~\ref{app:phot-num-fourier-dist}, we find that the probability of measuring $k\in\left\{  0,\ldots,d-1\right\}  $ is as follows:
\begin{equation}
\begin{split}
p_d(k|\phi) \coloneqq&\ \left\vert \langle u_{k}|e^{- i\hat{n}\phi}|+_{d}\rangle\right\vert ^{2}  \\
=&\ \frac{1}{d^{2}}\frac{\sin^{2}\!\left(  \pi k-\frac{d\phi}{2}  \right)}{\sin^{2}\!\left(  \frac{\pi k}{d}-\frac{\phi}{2} \right)  }\, .
\end{split}
\label{eq:phot-num-fourier-dist}
\end{equation}
The function on the last line above is 
proportional to the Fej\'er kernel, a well-known object in Fourier analysis~\cite{hoffman2007banach}; as a function of $\phi$, it is peaked at $\phi = 2\pi k/d$.
By classical post-processing applied to the value~$2\pi k/ d$
(i.e., performing the shift  $2\pi k/d - 2\pi$     if $ 2\pi k/d \in (\pi,2\pi)$, adding uniform noise chosen from the interval $[0,2\pi/d]$ of size $2\pi/d$), the discrete probability distribution $p_d(k|\phi)$ can be smoothed into a continuous probability density. Mathematically, this corresponds to the probability mass function $p_d(k|\phi)$ being convolved with a rectangle function $\Pi_d(x)$ associated with a uniform density of width $2\pi/d$, leading to the probability density
\begin{equation}
p_d(\hat{\phi}|\phi) \coloneqq  \sum_{k=0}^{d-1}p_{d}(k|\phi) \,\Pi_d\!\left( \hat{\phi} - \frac{2\pi k}{d} \,\,\mathrm{mod}\,\, 2\pi \right) 
\label{eq:smoothed-density}
\end{equation}
where $x\,\,\mathrm{mod}\,\, 2\pi \coloneqq \min\left\{ x+2\pi k:\, x+2\pi k\geq 0,\, k\in \mathbb{Z} \right\}$, and $\Pi_{d}(x)$ is defined as
\begin{equation}
\begin{aligned}
\Pi_{d}(x)\coloneqq \left\{
\begin{array}
[c]{ll}%
0 & : x<0\text{ or } x \geq \frac{2\pi}{d},\\[1ex]
\frac{d}{2\pi } & : x\geq 0 \text{ and } x < \frac{2\pi}{d}.
\end{array}
\right.
\end{aligned}
\label{eq:Pi_d}
\end{equation}
It then follows for every probability density $p(\phi)$ that
\begin{equation}
\lim_{d\rightarrow\infty}\int_{-\pi}^{\pi } d\hat{\phi}\ \left\vert p(\hat{\phi})- \int_{-\pi}^{\pi}d\phi\ p(\phi)p_{d}(\hat{\phi}|\phi)\right\vert =0,
\label{eq:convergence-photon-numb-fourier-meas}
\end{equation}
which is essentially equivalent to the less formal statement that $p_d(\hat{\phi}|\phi)$ converges to the Dirac delta function $\delta(\hat{\phi}-\phi)$ in the $d \to \infty$ limit. The above convergence statement in \eqref{eq:convergence-photon-numb-fourier-meas} is proved in a rigorous way in Lemma~\ref{pns_convergence_lemma} in Appendix~\ref{app:phot-num-fourier-dist}.

Now applying this reasoning to the BDC $\mathcal{D}_p$, we find from a direct application of~\eqref{eq:phot-num-fourier-dist} that
\begin{align}
    \operatorname{Tr}[|u_{k}\rangle\!\langle u_{k}|\mathcal{D}_{p}(|+_{d}%
\rangle\!\langle+_{d}|)]&=\int_{-\pi}^{\pi }d\phi\ p(\phi)p_{d}(k|\phi) ,
\label{eq:photon-num-fourier-meas-BDC-other}
\end{align}
Let us then denote by $(M_\phi)_\phi$ the measurement that 1) performs a Fourier basis measurement $\{|u_{k}\rangle\!\langle u_{k}|\}_k$ with outcome~$k$, 2) calculates the value $\phi = 2\pi k/ d$  and shifts by $-2\pi$ if $\phi \in (\pi,2\pi)$, and 3) finally adds to this value uniform noise selected from an interval of size $2\pi /d$ to produce an outcome $\phi$. It then follows as a consequence of~\eqref{eq:convergence-photon-numb-fourier-meas} and~\eqref{eq:photon-num-fourier-meas-BDC-other} that
\begin{equation}
    p(\phi) = \lim_{d\to \infty} \operatorname{Tr}[M_\phi \mathcal{D}_p(|+_{d}\rangle\!\langle+_{d}|)],
\end{equation}
concluding our proof of~\eqref{eq:att-limit-bdc} for this scheme.

\subsection{Coherent-state method}

\label{sec:coh-state-scheme}

This method involves preparing a coherent state at the input and performing heterodyne detection and classical post-processing at the output, which is a routine method  for optical phase estimation (see, e.g.,~\cite{PhysRevA.87.050303,Wiseman_1998,martin2020implementation,PhysRevLett.106.153603,PhysRevA.92.012317}).
The coherent-state method is easier to implement in practice than the photon-number-superposition method.   In this approach, the initial state is given
by the following coherent state:
\begin{align}
|\alpha\rangle\coloneqq e^{-\frac{1}{2}\left\vert \alpha\right\vert ^{2}}\sum
_{n=0}^{\infty}\frac{\alpha^{n}}{\sqrt{n!}}|n\rangle,
\end{align}
where $\alpha \in \mathbb{C}$. For the scheme we use here, we fix $\alpha \in \mathbb{R}_+$.
After the phase rotation $e^{- i\hat{n}\phi}$ acts, the state becomes
\begin{align}
\label{eq:coh_rot}
e^{- i\hat{n}\phi}|\alpha\rangle &   =|\alpha e^{- i\phi}\rangle,
\end{align}
as reviewed in Appendix~\ref{app:coh-state-method}.
Performing a heterodyne measurement (with POVM elements $\{\frac{1}{\pi}|\beta\rangle\!\langle\beta|\}_{\beta \in\mathbb{C}}$) 
on the state in~\eqref{eq:coh_rot} then leads to the measurement outcome~$\beta$.
The final step (classical post-processing) is to compute the argument of $\beta$, i.e.,  $\hat \phi \coloneqq \arg(\beta)$, as an estimate of the phase~$\phi$; the probability density for $\hat \phi$ is known as the Rician phase distribution and is given by~\cite[Eqs.~(10) \& (20)]{LZJ20}
\begin{multline}
p_{\alpha}(\hat{\phi}|\phi)\coloneqq \frac{e^{-\left\vert \alpha\right\vert ^{2}%
}}{2\pi}
+\frac{1}{2}\frac{\left\vert \alpha\right\vert }{\sqrt{\pi}}\cos
(\hat{\phi}-\phi )e^{-\left\vert \alpha\right\vert ^{2}\sin
^{2}(\hat{\phi}-\phi)}\times \\
\left[  1+\operatorname{erf}(\left\vert
\alpha\right\vert \cos(\hat{\phi}-\phi ))\right]  .
\label{eq:Rician}
\end{multline}
See 
Appendix~\ref{app:rician-phase-deriv} for a derivation of the Rician phase probability density, provided for convenience. Notably, this probability density is highly peaked at $\hat{\phi}=\phi$ and converges to a Dirac delta function in the following sense:
\begin{align}
\lim_{\alpha \rightarrow\infty}\int_{-\pi}^{\pi} d\hat{\phi}
\left |p(\hat{\phi}) - \int_{-\pi}^{\pi} d\phi
\ p_{\alpha }(\hat{\phi}
|\phi)\, p(\phi)\right |,
\label{eq:coh-state-sch-converge}
\end{align}
where $p(\phi)$ is an arbitrary probability density defined on the interval $[-\pi,\pi]$. We provide a rigorous statement of the above convergence in Lemma~\ref{coherent_state_method_convergence_lemma} in Appendix~\ref{app:coh-state-method}.

Finally, denoting by $(M_\phi)_\phi$ the measurement that 1) performs heterodyne detection with outcome $\beta$ and 2) outputs the value $\phi = \operatorname{arg}(\beta)$, it follows as a consequence of~\eqref{eq:coh-state-sch-converge} that
\begin{equation}
    p(\phi) = \lim_{\alpha \to \infty} \operatorname{Tr}[M_\phi \mathcal{D}_p(|\alpha\rangle\!\langle \alpha |)],
    \label{eq:coh-state-sch-lim}
\end{equation}
concluding our proof of~\eqref{eq:att-limit-bdc} for this scheme. Let us remark that an explicit form for the POVM $(M_\phi)_\phi$ was obtained in \cite[Eq.~(3.10)]{Wiseman_1998} and is as follows:
\begin{equation}
    M_\phi = \frac{1}{2\pi} \sum_{m,n=0}^{\infty} |m\rangle\!\langle n| e^{i \phi (m-n)}
    \frac{\Gamma\!\left(\frac{n+m}{2}+1\right)}{\sqrt{n!m!}}.
\end{equation}

\subsection{Comparison of methods for finite energy}

Let us compare the performance of the photon-number-superposition and coherent-state methods to the fundamental limit when there is an energy constraint in place. In particular, let us consider channel discrimination (asymmetric error) of two bosonic dephasing channels $\mathcal{D}_{p_{\gamma_1}}$ and $\mathcal{D}_{p_{\gamma_2}}$, for which the underlying probability densities $p_{\gamma_1}$ and $p_{\gamma_2}$ are wrapped normal distributions with respective variances $\gamma_1$ and $\gamma_2$. Figure~\ref{fig:RE_ratio} illustrates how quickly the relative entropy of these schemes converges to the optimal relative entropy.
For the photon-number-superposition scheme, the probability density as a function of $d$ is given by 
\begin{equation}
    p_{d,\gamma}(\hat \phi) \coloneqq 
\label{eq:photon-num-fourier-meas-BDC}
\int_{-\pi}^{\pi }d\phi\ p_\gamma (\phi)\, p_{d}(\hat{\phi}|\phi) ,
\end{equation}
where $p_{d}(\hat{\phi}|\phi)$ is defined in \eqref{eq:smoothed-density},
from which we can calculate the relative entropy $D(p_{d,\gamma_1}\Vert p_{d,\gamma_2})$ of this scheme as a function of~$d$.
For the coherent-state scheme,
the probability density as a function of $\alpha$ is given by
\begin{align}
p_{\alpha,\gamma}(\hat \phi) &\coloneqq  \int_{-\pi}^\pi d\phi ~p_\alpha(\hat\phi|\phi)\,  p_\gamma(\phi),
\end{align}
from which we can calculate the relative entropy $D(p_{\alpha,\gamma_1} \Vert\, p_{\alpha,\gamma_2} )$ for this scheme as a function of $\alpha$.
Interestingly, Figure~\ref{fig:RE_ratio} indicates that these schemes in practice come close to achieving the fundamental limit, and we also see that the coherent-state scheme has an advantage over the photon-number scheme for the same fixed energy, given that the mean photon number is 9.5 for the state $|+_d\rangle$ in \eqref{eq:+_d} when~$d=20$.

\begin{figure}
\includegraphics[width=\linewidth]{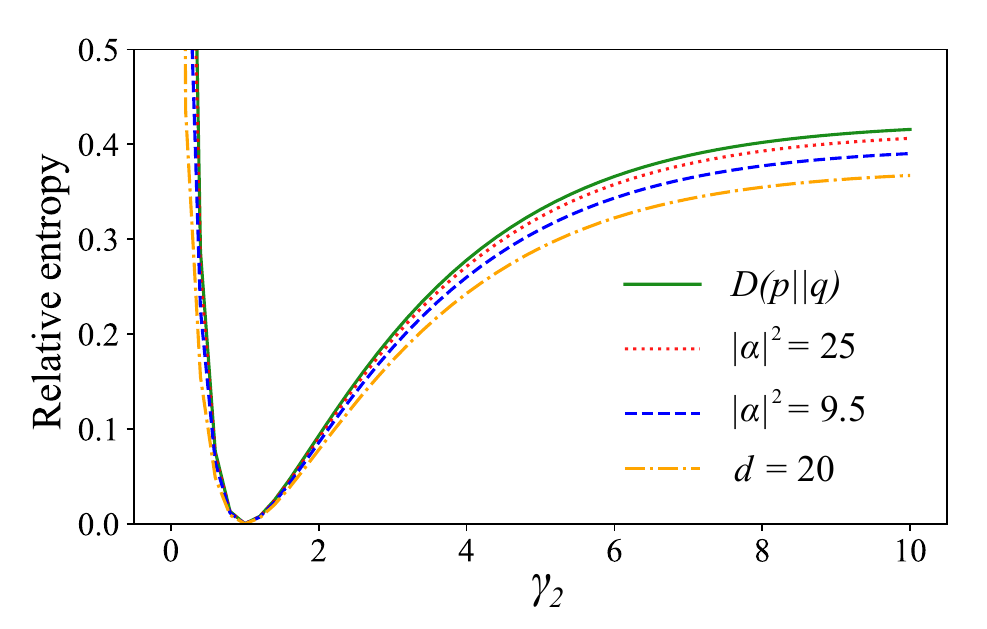}
\caption{\label{fig:RE_ratio} 
Comparison of photon-number-superposition and coherent-state schemes to the fundamental limit, when considering channel discrimination in the setting of asymmetric error. The figure plots the various relative entropies when $\gamma_1 = 1$; here $d=20$ for the photon-number-superposition scheme and $|\alpha|^2 \in \{9.5,25\}$ for the coherent-state scheme. Note that the average energy of the input state is the same for $|\alpha|^2 = 9.5$ and $d=20$.}
\end{figure}

\color{black}

\section{Bosonic dephasing and loss}

\label{sec:deph-plus-loss}

Our results apply more generally to a scenario that involves photon loss in addition to dephasing. This indicates a certain robustness of our results, since we expect to encounter photon loss in any realistic scenario. Namely, suppose that the two channels to distinguish are $\mathcal{L}_\eta \circ \mathcal{D}_p$ and $\mathcal{L}_\eta \circ \mathcal{D}_q$, where $\mathcal{L}_\eta$ is a pure-loss bosonic channel of transmissivity $\eta \in (0,1] $. This composite channel has been studied in the context of quantum communication, under the name  bosonic loss-dephasing channel~\cite{Leviant2022, loss-dephasing}. Our main observation here is that the distinguishability of these channels in all scenarios considered is no different from the distinguishability of $\mathcal{D}_p$ and $\mathcal{D}_q$. Thus, all results stated above for $\mathcal{D}_p$ and $\mathcal{D}_q$ hold also for~$\mathcal{L}_\eta \circ \mathcal{D}_p$ and $\mathcal{L}_\eta \circ \mathcal{D}_q$.

The optimality part of this claim follows by similar reasoning used in Section~\ref{sec:optimality}. That is, since the pure-loss channel $\mathcal{L}_\eta$ is common to both  $\mathcal{L}_\eta \circ \mathcal{D}_p$ and $\mathcal{L}_\eta \circ \mathcal{D}_q$, it can be considered as part of an adaptive strategy used to discriminate these channels, and so their distinguishability is still limited by the classical environmental states $\sigma_p$ and $\sigma_q$. The attainability part follows by using the scheme from Section~\ref{sec:attainability} and the fact that coherent states retain their purity after the action of a pure-loss channel. That is, $\mathcal{L}_\eta(|\alpha \rangle \! \langle \alpha |) = |\sqrt{\eta}\alpha \rangle \! \langle \sqrt{\eta} \alpha |$. Then the following limit holds by applying the same reasoning used to justify~\eqref{eq:coh-state-sch-lim}:
\begin{equation}
\begin{split}
    p(\phi) & = \lim_{\alpha \to \infty} \operatorname{Tr}[M_\phi (\mathcal{L}_\eta \circ \mathcal{D}_p)(|\alpha\rangle\!\langle \alpha |)], \\
    & = \lim_{\alpha \to \infty} \operatorname{Tr}[M_\phi  \mathcal{D}_p(|\sqrt{\eta} \alpha \rangle \! \langle \sqrt{\eta} \alpha  |)].
\end{split}
    \label{eq:limit-with-loss}
\end{equation}
Here we also used the fact that dephasing channels and pure-loss channels commute; i.e., $\mathcal{L}_\eta \circ \mathcal{D}_p =   \mathcal{D}_p \circ \mathcal{L}_\eta$.

Similarly, all estimation results stated above for the parameterized family $(\mathcal{D}_{p_\theta})_{\theta \in \Theta}$ hold also for the parameterized family $(\mathcal{L}_\eta \circ \mathcal{D}_{p_\theta})_{\theta \in \Theta}$. This follows from the same reasoning given for the discrimination setting. Namely, the optimality part follows because all estimation strategies for the family $(\mathcal{L}_\eta \circ \mathcal{D}_{p_\theta})_{\theta \in \Theta}$ are limited by those of the family $(p_\theta)_{\theta \in \Theta}$ of probability densities. Then for the attainability part, the equality in~\eqref{eq:limit-with-loss} applies, allowing us to apply the reasoning in Section~\ref{sec:attainability} again.

Finally, numerical estimates using the probability distribution derived in Appendix~\ref{app:lossy-fejer} indicate that the photon-number-superposition method from Section~\ref{sec:phot-num-sup-scheme} might be optimal also in the presence of loss, provided that one considers the limit of infinite energy. That is, although the uniform superposition state in~\eqref{eq:+_d} is affected detrimentally by loss, it seems to retain sufficiently high coherence to effectively detect a phase-space rotation. The coherent-state scheme from Section~\ref{sec:coh-state-scheme}, however, might still have an advantage over the photon-number-superposition method also in the presence of loss if one considers the finite-energy setting; Figure~\ref{fig:RE_ratio} illustrates that this is indeed the case for channel discrimination in the setting of asymmetric error.

\section{Conclusion}

\label{sec:conclusion}

In conclusion, we have determined the fundamental limits of discrimination and estimation for BDCs, complementing the recent results of~\cite{lami2023exact} on communication. Not only have we accomplished this for asymptotic quantities like the symmetric and asymmetric error exponents for channel discrimination, but we have also done so for the underlying fundamental, operational quantities like the symmetric and asymmetric error probabilities of an arbitrary $n$-round adaptive strategy (see~\eqref{eq:symm-err-cl} and~\eqref{eq:asymm-err-prob-cl}, respectively). We have done the same for the main operational quantity in channel estimation, the risk of an $n$-round adaptive strategy (see~\eqref{eq:ch-est-main-result}). The main ideas for these results relied on the method of simulation from~\cite{matsumoto2010metric}, for the optimality part, and to exhibit a sequence of strategies that pre- and post-process a BDC to recover its underlying probability density, for the attainability part. This is similar in spirit to previous results of~\cite{PhysRevLett.118.100502,takeoka2016optimal}.

Going forward from here, the main pressing open question is to determine the limits for these tasks whenever there is a realistic energy constraint in place. More specifically, we think it is interesting to determine which scheme, either the photon-number-superposition scheme from Section~\ref{sec:phot-num-sup-scheme} or the coherent-state scheme from Section~\ref{sec:coh-state-scheme}, performs better in the finite-energy regime, as well as in the case that there is photon loss in addition to dephasing. There are certainly other schemes besides these two to consider as well. Furthermore, given that our findings in Section~\ref{sec:deph-plus-loss} only apply when the transmissivity parameter $\eta$ is fixed, it is open to determine the limits of discrimination and estimation when the transmissivity parameter varies in addition to  the dephasing channel.

Another natural generalization of our results is to the case of an arbitrary random unitary channel of the form 
\begin{equation}
\mathcal{N}_{p,H}(\rho) \coloneqq \int_{-\infty}^{+\infty} dt\ p(t)\, e^{-iHt}\rho\, e^{iHt} ,
\end{equation}
where $p$ is a probability density on the real line and $H$ is a general Hamiltonian. The same simulation arguments from Section~\ref{sec:optimality} allow for concluding optimality bounds, that all adaptive strategies for discriminating or estimating channels from this class are limited by the underlying classical probability densities. Based on the insights from~\cite[Proposition~2]{VV-diamond}, we expect that seizing the underlying probability density $p$ might be possible for a large class of Hamiltonians. If that is the case, then our results could be extended far beyond the setting we considered here.

\bigskip

\textit{Data availability statement}---All codes used to generate the figures in this paper are available with the arXiv posting of this paper as arXiv ancillary files.

\begin{acknowledgements}

We thank Francisco Elohim Becerra, Maison Clouâtré, Bunyamin Kartal, Ufuk Keskin, Marco A. Rodr\'iguez-Garc\'ia, and  Moe Win for helpful comments.
ZH is supported by a Sydney Quantum Academy Postdoctoral Fellowship
and an ARC DECRA Fellowship (DE230100144) ``Quantum-enabled super-resolution imaging''. 
MMW acknowledges support from the National Science Foundation under Grant No.~2304816 and is grateful to CWI Amsterdam for hospitality during a research visit when this work was finalized.

\end{acknowledgements}

\bibliography{BDC}

\begin{thebibliography}{84}%
\makeatletter
\providecommand \@ifxundefined [1]{%
 \@ifx{#1\undefined}
}%
\providecommand \@ifnum [1]{%
 \ifnum #1\expandafter \@firstoftwo
 \else \expandafter \@secondoftwo
 \fi
}%
\providecommand \@ifx [1]{%
 \ifx #1\expandafter \@firstoftwo
 \else \expandafter \@secondoftwo
 \fi
}%
\providecommand \natexlab [1]{#1}%
\providecommand \enquote  [1]{``#1''}%
\providecommand \bibnamefont  [1]{#1}%
\providecommand \bibfnamefont [1]{#1}%
\providecommand \citenamefont [1]{#1}%
\providecommand \href@noop [0]{\@secondoftwo}%
\providecommand \href [0]{\begingroup \@sanitize@url \@href}%
\providecommand \@href[1]{\@@startlink{#1}\@@href}%
\providecommand \@@href[1]{\endgroup#1\@@endlink}%
\providecommand \@sanitize@url [0]{\catcode `\\12\catcode `\$12\catcode `\&12\catcode `\#12\catcode `\^12\catcode `\_12\catcode `\%12\relax}%
\providecommand \@@startlink[1]{}%
\providecommand \@@endlink[0]{}%
\providecommand \url  [0]{\begingroup\@sanitize@url \@url }%
\providecommand \@url [1]{\endgroup\@href {#1}{\urlprefix }}%
\providecommand \urlprefix  [0]{URL }%
\providecommand \Eprint [0]{\href }%
\providecommand \doibase [0]{https://doi.org/}%
\providecommand \selectlanguage [0]{\@gobble}%
\providecommand \bibinfo  [0]{\@secondoftwo}%
\providecommand \bibfield  [0]{\@secondoftwo}%
\providecommand \translation [1]{[#1]}%
\providecommand \BibitemOpen [0]{}%
\providecommand \bibitemStop [0]{}%
\providecommand \bibitemNoStop [0]{.\EOS\space}%
\providecommand \EOS [0]{\spacefactor3000\relax}%
\providecommand \BibitemShut  [1]{\csname bibitem#1\endcsname}%
\let\auto@bib@innerbib\@empty
\bibitem [{\citenamefont {Lami}\ and\ \citenamefont {Wilde}(2023)}]{lami2023exact}%
  \BibitemOpen
  \bibfield  {author} {\bibinfo {author} {\bibfnamefont {L.}~\bibnamefont {Lami}}\ and\ \bibinfo {author} {\bibfnamefont {M.~M.}\ \bibnamefont {Wilde}},\ }\bibfield  {title} {\bibinfo {title} {Exact solution for the quantum and private capacities of bosonic dephasing channels},\ }\href {https://doi.org/10.1038/s41566-023-01190-4} {\bibfield  {journal} {\bibinfo  {journal} {Nature Photonics}\ }\textbf {\bibinfo {volume} {17}},\ \bibinfo {pages} {525} (\bibinfo {year} {2023})}\BibitemShut {NoStop}%
\bibitem [{\citenamefont {Gerry}\ and\ \citenamefont {Knight}(2004)}]{gerry2004introductory}%
  \BibitemOpen
  \bibfield  {author} {\bibinfo {author} {\bibfnamefont {C.~C.}\ \bibnamefont {Gerry}}\ and\ \bibinfo {author} {\bibfnamefont {P.~L.}\ \bibnamefont {Knight}},\ }\href {https://doi.org/10.1017/CBO9780511791239} {\emph {\bibinfo {title} {Introductory Quantum Optics}}}\ (\bibinfo  {publisher} {Cambridge University Press},\ \bibinfo {year} {2004})\BibitemShut {NoStop}%
\bibitem [{\citenamefont {Suter}\ and\ \citenamefont {\'Alvarez}(2016)}]{RevModPhys.88.041001}%
  \BibitemOpen
  \bibfield  {author} {\bibinfo {author} {\bibfnamefont {D.}~\bibnamefont {Suter}}\ and\ \bibinfo {author} {\bibfnamefont {G.~A.}\ \bibnamefont {\'Alvarez}},\ }\bibfield  {title} {\bibinfo {title} {Colloquium: Protecting quantum information against environmental noise},\ }\href {https://doi.org/10.1103/RevModPhys.88.041001} {\bibfield  {journal} {\bibinfo  {journal} {Reviews of Modern Physics}\ }\textbf {\bibinfo {volume} {88}},\ \bibinfo {pages} {041001} (\bibinfo {year} {2016})}\BibitemShut {NoStop}%
\bibitem [{\citenamefont {Wanser}(1992)}]{W92}%
  \BibitemOpen
  \bibfield  {author} {\bibinfo {author} {\bibfnamefont {K.~H.}\ \bibnamefont {Wanser}},\ }\bibfield  {title} {\bibinfo {title} {Fundamental phase noise limit in optical fibres due to temperature fluctuations},\ }\href@noop {} {\bibfield  {journal} {\bibinfo  {journal} {Electronics Letters}\ }\textbf {\bibinfo {volume} {28}},\ \bibinfo {pages} {53} (\bibinfo {year} {1992})}\BibitemShut {NoStop}%
\bibitem [{\citenamefont {Jiang}\ and\ \citenamefont {Chen}(2010)}]{JC10}%
  \BibitemOpen
  \bibfield  {author} {\bibinfo {author} {\bibfnamefont {L.-Z.}\ \bibnamefont {Jiang}}\ and\ \bibinfo {author} {\bibfnamefont {X.-Y.}\ \bibnamefont {Chen}},\ }\bibfield  {title} {\bibinfo {title} {Evaluating the quantum capacity of bosonic dephasing channel},\ }in\ \href {https://doi.org/10.1117/12.870179} {\emph {\bibinfo {booktitle} {Quantum and Nonlinear Optics}}},\ Vol.\ \bibinfo {volume} {7846},\ \bibinfo {editor} {edited by\ \bibinfo {editor} {\bibfnamefont {Q.}~\bibnamefont {Gong}}, \bibinfo {editor} {\bibfnamefont {G.-C.}\ \bibnamefont {Guo}},\ and\ \bibinfo {editor} {\bibfnamefont {Y.-R.}\ \bibnamefont {Shen}}},\ \bibinfo {organization} {International Society for Optics and Photonics}\ (\bibinfo  {publisher} {SPIE},\ \bibinfo {year} {2010})\ pp.\ \bibinfo {pages} {244--249}\BibitemShut {NoStop}%
\bibitem [{\citenamefont {Arqand}\ \emph {et~al.}(2020)\citenamefont {Arqand}, \citenamefont {Memarzadeh},\ and\ \citenamefont {Mancini}}]{PhysRevA.102.042413}%
  \BibitemOpen
  \bibfield  {author} {\bibinfo {author} {\bibfnamefont {A.}~\bibnamefont {Arqand}}, \bibinfo {author} {\bibfnamefont {L.}~\bibnamefont {Memarzadeh}},\ and\ \bibinfo {author} {\bibfnamefont {S.}~\bibnamefont {Mancini}},\ }\bibfield  {title} {\bibinfo {title} {Quantum capacity of a bosonic dephasing channel},\ }\href {https://doi.org/10.1103/PhysRevA.102.042413} {\bibfield  {journal} {\bibinfo  {journal} {Physical Review A}\ }\textbf {\bibinfo {volume} {102}},\ \bibinfo {pages} {042413} (\bibinfo {year} {2020})}\BibitemShut {NoStop}%
\bibitem [{\citenamefont {Zhuang}(2021)}]{Z21}%
  \BibitemOpen
  \bibfield  {author} {\bibinfo {author} {\bibfnamefont {Q.}~\bibnamefont {Zhuang}},\ }\bibfield  {title} {\bibinfo {title} {Quantum-enabled communication without a phase reference},\ }\href {https://doi.org/10.1103/PhysRevLett.126.060502} {\bibfield  {journal} {\bibinfo  {journal} {Physical Review Letters}\ }\textbf {\bibinfo {volume} {126}},\ \bibinfo {pages} {060502} (\bibinfo {year} {2021})}\BibitemShut {NoStop}%
\bibitem [{\citenamefont {Fanizza}\ \emph {et~al.}(2021)\citenamefont {Fanizza}, \citenamefont {Rosati}, \citenamefont {Skotiniotis}, \citenamefont {Calsamiglia},\ and\ \citenamefont {Giovannetti}}]{Fanizza2021squeezingenhanced}%
  \BibitemOpen
  \bibfield  {author} {\bibinfo {author} {\bibfnamefont {M.}~\bibnamefont {Fanizza}}, \bibinfo {author} {\bibfnamefont {M.}~\bibnamefont {Rosati}}, \bibinfo {author} {\bibfnamefont {M.}~\bibnamefont {Skotiniotis}}, \bibinfo {author} {\bibfnamefont {J.}~\bibnamefont {Calsamiglia}},\ and\ \bibinfo {author} {\bibfnamefont {V.}~\bibnamefont {Giovannetti}},\ }\bibfield  {title} {\bibinfo {title} {Squeezing-enhanced communication without a phase reference},\ }\href {https://doi.org/10.22331/q-2021-12-23-608} {\bibfield  {journal} {\bibinfo  {journal} {Quantum}\ }\textbf {\bibinfo {volume} {5}},\ \bibinfo {pages} {608} (\bibinfo {year} {2021})}\BibitemShut {NoStop}%
\bibitem [{\citenamefont {Rexiti}\ \emph {et~al.}(2022)\citenamefont {Rexiti}, \citenamefont {Memarzadeh},\ and\ \citenamefont {Mancini}}]{rexiti2022discrimination}%
  \BibitemOpen
  \bibfield  {author} {\bibinfo {author} {\bibfnamefont {M.}~\bibnamefont {Rexiti}}, \bibinfo {author} {\bibfnamefont {L.}~\bibnamefont {Memarzadeh}},\ and\ \bibinfo {author} {\bibfnamefont {S.}~\bibnamefont {Mancini}},\ }\bibfield  {title} {\bibinfo {title} {Discrimination of dephasing channels},\ }\href {https://doi.org/10.1088/1751-8121/ac6d2c} {\bibfield  {journal} {\bibinfo  {journal} {Journal of Physics A: Mathematical and Theoretical}\ }\textbf {\bibinfo {volume} {55}},\ \bibinfo {pages} {245301} (\bibinfo {year} {2022})}\BibitemShut {NoStop}%
\bibitem [{\citenamefont {Arqand}\ \emph {et~al.}(2023)\citenamefont {Arqand}, \citenamefont {Memarzadeh},\ and\ \citenamefont {Mancini}}]{AMM23}%
  \BibitemOpen
  \bibfield  {author} {\bibinfo {author} {\bibfnamefont {A.}~\bibnamefont {Arqand}}, \bibinfo {author} {\bibfnamefont {L.}~\bibnamefont {Memarzadeh}},\ and\ \bibinfo {author} {\bibfnamefont {S.}~\bibnamefont {Mancini}},\ }\bibfield  {title} {\bibinfo {title} {Energy-constrained {LOCC}-assisted quantum capacity of the bosonic dephasing channel},\ }\href {https://doi.org/10.3390/e25071001} {\bibfield  {journal} {\bibinfo  {journal} {Entropy}\ }\textbf {\bibinfo {volume} {25}},\ \bibinfo {pages} {1001} (\bibinfo {year} {2023})}\BibitemShut {NoStop}%
\bibitem [{\citenamefont {Terhal}(2015)}]{RevModPhys.87.307}%
  \BibitemOpen
  \bibfield  {author} {\bibinfo {author} {\bibfnamefont {B.~M.}\ \bibnamefont {Terhal}},\ }\bibfield  {title} {\bibinfo {title} {Quantum error correction for quantum memories},\ }\href {https://doi.org/10.1103/RevModPhys.87.307} {\bibfield  {journal} {\bibinfo  {journal} {Reviews of Modern Physics}\ }\textbf {\bibinfo {volume} {87}},\ \bibinfo {pages} {307} (\bibinfo {year} {2015})}\BibitemShut {NoStop}%
\bibitem [{\citenamefont {AI}(2021)}]{google2021exponential}%
  \BibitemOpen
  \bibfield  {author} {\bibinfo {author} {\bibfnamefont {G.~Q.}\ \bibnamefont {AI}},\ }\bibfield  {title} {\bibinfo {title} {Exponential suppression of bit or phase errors with cyclic error correction},\ }\href {https://doi.org/10.1038/s41586-021-03588-y} {\bibfield  {journal} {\bibinfo  {journal} {Nature}\ }\textbf {\bibinfo {volume} {595}},\ \bibinfo {pages} {383} (\bibinfo {year} {2021})}\BibitemShut {NoStop}%
\bibitem [{\citenamefont {Sidhu}\ \emph {et~al.}(2023)\citenamefont {Sidhu}, \citenamefont {Bullock}, \citenamefont {Guha},\ and\ \citenamefont {Lupo}}]{sidhu2023linear}%
  \BibitemOpen
  \bibfield  {author} {\bibinfo {author} {\bibfnamefont {J.~S.}\ \bibnamefont {Sidhu}}, \bibinfo {author} {\bibfnamefont {M.~S.}\ \bibnamefont {Bullock}}, \bibinfo {author} {\bibfnamefont {S.}~\bibnamefont {Guha}},\ and\ \bibinfo {author} {\bibfnamefont {C.}~\bibnamefont {Lupo}},\ }\bibfield  {title} {\bibinfo {title} {Linear optics and photodetection achieve near-optimal unambiguous coherent state discrimination},\ }\href {https://doi.org/10.22331/q-2023-05-31-1025} {\bibfield  {journal} {\bibinfo  {journal} {Quantum}\ }\textbf {\bibinfo {volume} {7}},\ \bibinfo {pages} {1025} (\bibinfo {year} {2023})}\BibitemShut {NoStop}%
\bibitem [{\citenamefont {Huang}\ and\ \citenamefont {Lupo}(2021)}]{PhysRevLett.127.130502}%
  \BibitemOpen
  \bibfield  {author} {\bibinfo {author} {\bibfnamefont {Z.}~\bibnamefont {Huang}}\ and\ \bibinfo {author} {\bibfnamefont {C.}~\bibnamefont {Lupo}},\ }\bibfield  {title} {\bibinfo {title} {Quantum hypothesis testing for exoplanet detection},\ }\href {https://doi.org/10.1103/PhysRevLett.127.130502} {\bibfield  {journal} {\bibinfo  {journal} {Physical Review Letters}\ }\textbf {\bibinfo {volume} {127}},\ \bibinfo {pages} {130502} (\bibinfo {year} {2021})}\BibitemShut {NoStop}%
\bibitem [{\citenamefont {Huang}\ \emph {et~al.}(2023)\citenamefont {Huang}, \citenamefont {Schwab},\ and\ \citenamefont {Lupo}}]{PhysRevA.107.022409}%
  \BibitemOpen
  \bibfield  {author} {\bibinfo {author} {\bibfnamefont {Z.}~\bibnamefont {Huang}}, \bibinfo {author} {\bibfnamefont {C.}~\bibnamefont {Schwab}},\ and\ \bibinfo {author} {\bibfnamefont {C.}~\bibnamefont {Lupo}},\ }\bibfield  {title} {\bibinfo {title} {Ultimate limits of exoplanet spectroscopy: A quantum approach},\ }\href {https://doi.org/10.1103/PhysRevA.107.022409} {\bibfield  {journal} {\bibinfo  {journal} {Physical Review A}\ }\textbf {\bibinfo {volume} {107}},\ \bibinfo {pages} {022409} (\bibinfo {year} {2023})}\BibitemShut {NoStop}%
\bibitem [{\citenamefont {Shi}\ \emph {et~al.}(2020)\citenamefont {Shi}, \citenamefont {Zhang}, \citenamefont {Pirandola},\ and\ \citenamefont {Zhuang}}]{PhysRevLett.125.180502}%
  \BibitemOpen
  \bibfield  {author} {\bibinfo {author} {\bibfnamefont {H.}~\bibnamefont {Shi}}, \bibinfo {author} {\bibfnamefont {Z.}~\bibnamefont {Zhang}}, \bibinfo {author} {\bibfnamefont {S.}~\bibnamefont {Pirandola}},\ and\ \bibinfo {author} {\bibfnamefont {Q.}~\bibnamefont {Zhuang}},\ }\bibfield  {title} {\bibinfo {title} {Entanglement-assisted absorption spectroscopy},\ }\href {https://doi.org/10.1103/PhysRevLett.125.180502} {\bibfield  {journal} {\bibinfo  {journal} {Physical Review Letters}\ }\textbf {\bibinfo {volume} {125}},\ \bibinfo {pages} {180502} (\bibinfo {year} {2020})}\BibitemShut {NoStop}%
\bibitem [{\citenamefont {Bae}\ and\ \citenamefont {Kwek}(2015)}]{bae2015quantum}%
  \BibitemOpen
  \bibfield  {author} {\bibinfo {author} {\bibfnamefont {J.}~\bibnamefont {Bae}}\ and\ \bibinfo {author} {\bibfnamefont {L.-C.}\ \bibnamefont {Kwek}},\ }\bibfield  {title} {\bibinfo {title} {Quantum state discrimination and its applications},\ }\href {https://doi.org/10.1088/1751-8113/48/8/083001} {\bibfield  {journal} {\bibinfo  {journal} {Journal of Physics A: Mathematical and Theoretical}\ }\textbf {\bibinfo {volume} {48}},\ \bibinfo {pages} {083001} (\bibinfo {year} {2015})}\BibitemShut {NoStop}%
\bibitem [{\citenamefont {Chefles}(2004)}]{chefles200412}%
  \BibitemOpen
  \bibfield  {author} {\bibinfo {author} {\bibfnamefont {A.}~\bibnamefont {Chefles}},\ }\bibinfo {title} {Quantum states: Discrimination and classical information transmission. a review of experimental progress},\ in\ \href {https://doi.org/10.1007/978-3-540-44481-7_12} {\emph {\bibinfo {booktitle} {Quantum State Estimation}}},\ \bibinfo {editor} {edited by\ \bibinfo {editor} {\bibfnamefont {M.}~\bibnamefont {Paris}}\ and\ \bibinfo {editor} {\bibfnamefont {J.}~\bibnamefont {{\v{R}}eh{\'a}{\v{c}}ek}}}\ (\bibinfo  {publisher} {Springer Berlin Heidelberg},\ \bibinfo {address} {Berlin, Heidelberg},\ \bibinfo {year} {2004})\ pp.\ \bibinfo {pages} {467--511}\BibitemShut {NoStop}%
\bibitem [{\citenamefont {Chiribella}\ \emph {et~al.}(2008)\citenamefont {Chiribella}, \citenamefont {D'Ariano},\ and\ \citenamefont {Perinotti}}]{CDP08}%
  \BibitemOpen
  \bibfield  {author} {\bibinfo {author} {\bibfnamefont {G.}~\bibnamefont {Chiribella}}, \bibinfo {author} {\bibfnamefont {G.~M.}\ \bibnamefont {D'Ariano}},\ and\ \bibinfo {author} {\bibfnamefont {P.}~\bibnamefont {Perinotti}},\ }\bibfield  {title} {\bibinfo {title} {Memory effects in quantum channel discrimination},\ }\href {https://doi.org/10.1103/PhysRevLett.101.180501} {\bibfield  {journal} {\bibinfo  {journal} {Physical Review Letters}\ }\textbf {\bibinfo {volume} {101}},\ \bibinfo {pages} {180501} (\bibinfo {year} {2008})},\ \bibinfo {note} {arXiv:0803.3237}\BibitemShut {NoStop}%
\bibitem [{\citenamefont {Duan}\ \emph {et~al.}(2009)\citenamefont {Duan}, \citenamefont {Feng},\ and\ \citenamefont {Ying}}]{Duan09}%
  \BibitemOpen
  \bibfield  {author} {\bibinfo {author} {\bibfnamefont {R.}~\bibnamefont {Duan}}, \bibinfo {author} {\bibfnamefont {Y.}~\bibnamefont {Feng}},\ and\ \bibinfo {author} {\bibfnamefont {M.}~\bibnamefont {Ying}},\ }\bibfield  {title} {\bibinfo {title} {Perfect distinguishability of quantum operations},\ }\href {https://doi.org/10.1103/PhysRevLett.103.210501} {\bibfield  {journal} {\bibinfo  {journal} {Physical Review Letters}\ }\textbf {\bibinfo {volume} {103}},\ \bibinfo {pages} {210501} (\bibinfo {year} {2009})},\ \bibinfo {note} {arXiv:0908.0119}\BibitemShut {NoStop}%
\bibitem [{\citenamefont {Piani}\ and\ \citenamefont {Watrous}(2009)}]{PW09}%
  \BibitemOpen
  \bibfield  {author} {\bibinfo {author} {\bibfnamefont {M.}~\bibnamefont {Piani}}\ and\ \bibinfo {author} {\bibfnamefont {J.}~\bibnamefont {Watrous}},\ }\bibfield  {title} {\bibinfo {title} {All entangled states are useful for channel discrimination},\ }\href {https://doi.org/10.1103/PhysRevLett.102.250501} {\bibfield  {journal} {\bibinfo  {journal} {Physical Review Letters}\ }\textbf {\bibinfo {volume} {102}},\ \bibinfo {pages} {250501} (\bibinfo {year} {2009})}\BibitemShut {NoStop}%
\bibitem [{\citenamefont {Hayashi}(2009)}]{Hayashi09}%
  \BibitemOpen
  \bibfield  {author} {\bibinfo {author} {\bibfnamefont {M.}~\bibnamefont {Hayashi}},\ }\bibfield  {title} {\bibinfo {title} {Discrimination of two channels by adaptive methods and its application to quantum system},\ }\href {https://doi.org/10.1109/TIT.2009.2023726} {\bibfield  {journal} {\bibinfo  {journal} {IEEE Transactions on Information Theory}\ }\textbf {\bibinfo {volume} {55}},\ \bibinfo {pages} {3807} (\bibinfo {year} {2009})},\ \bibinfo {note} {arXiv:0804.0686}\BibitemShut {NoStop}%
\bibitem [{\citenamefont {Harrow}\ \emph {et~al.}(2010)\citenamefont {Harrow}, \citenamefont {Hassidim}, \citenamefont {Leung},\ and\ \citenamefont {Watrous}}]{Harrow10}%
  \BibitemOpen
  \bibfield  {author} {\bibinfo {author} {\bibfnamefont {A.~W.}\ \bibnamefont {Harrow}}, \bibinfo {author} {\bibfnamefont {A.}~\bibnamefont {Hassidim}}, \bibinfo {author} {\bibfnamefont {D.~W.}\ \bibnamefont {Leung}},\ and\ \bibinfo {author} {\bibfnamefont {J.}~\bibnamefont {Watrous}},\ }\bibfield  {title} {\bibinfo {title} {Adaptive versus nonadaptive strategies for quantum channel discrimination},\ }\href {https://doi.org/10.1103/PhysRevA.81.032339} {\bibfield  {journal} {\bibinfo  {journal} {Physical Review A}\ }\textbf {\bibinfo {volume} {81}},\ \bibinfo {pages} {032339} (\bibinfo {year} {2010})},\ \bibinfo {note} {arXiv:0909.0256}\BibitemShut {NoStop}%
\bibitem [{\citenamefont {Matthews}\ \emph {et~al.}(2010)\citenamefont {Matthews}, \citenamefont {Piani},\ and\ \citenamefont {Watrous}}]{MPW10}%
  \BibitemOpen
  \bibfield  {author} {\bibinfo {author} {\bibfnamefont {W.}~\bibnamefont {Matthews}}, \bibinfo {author} {\bibfnamefont {M.}~\bibnamefont {Piani}},\ and\ \bibinfo {author} {\bibfnamefont {J.}~\bibnamefont {Watrous}},\ }\bibfield  {title} {\bibinfo {title} {Entanglement in channel discrimination with restricted measurements},\ }\href {https://doi.org/10.1103/PhysRevA.82.032302} {\bibfield  {journal} {\bibinfo  {journal} {Physical Review A}\ }\textbf {\bibinfo {volume} {82}},\ \bibinfo {pages} {032302} (\bibinfo {year} {2010})}\BibitemShut {NoStop}%
\bibitem [{\citenamefont {Cooney}\ \emph {et~al.}(2016)\citenamefont {Cooney}, \citenamefont {Mosonyi},\ and\ \citenamefont {Wilde}}]{Cooney2016}%
  \BibitemOpen
  \bibfield  {author} {\bibinfo {author} {\bibfnamefont {T.}~\bibnamefont {Cooney}}, \bibinfo {author} {\bibfnamefont {M.}~\bibnamefont {Mosonyi}},\ and\ \bibinfo {author} {\bibfnamefont {M.~M.}\ \bibnamefont {Wilde}},\ }\bibfield  {title} {\bibinfo {title} {Strong converse exponents for a quantum channel discrimination problem and quantum-feedback-assisted communication},\ }\href {https://doi.org/10.1007/s00220-016-2645-4} {\bibfield  {journal} {\bibinfo  {journal} {Communications in Mathematical Physics}\ }\textbf {\bibinfo {volume} {344}},\ \bibinfo {pages} {797} (\bibinfo {year} {2016})},\ \bibinfo {note} {arXiv:1408.3373}\BibitemShut {NoStop}%
\bibitem [{\citenamefont {Pirandola}\ and\ \citenamefont {Lupo}(2017)}]{PhysRevLett.118.100502}%
  \BibitemOpen
  \bibfield  {author} {\bibinfo {author} {\bibfnamefont {S.}~\bibnamefont {Pirandola}}\ and\ \bibinfo {author} {\bibfnamefont {C.}~\bibnamefont {Lupo}},\ }\bibfield  {title} {\bibinfo {title} {Ultimate precision of adaptive noise estimation},\ }\href {https://doi.org/10.1103/PhysRevLett.118.100502} {\bibfield  {journal} {\bibinfo  {journal} {Physical Review Letters}\ }\textbf {\bibinfo {volume} {118}},\ \bibinfo {pages} {100502} (\bibinfo {year} {2017})}\BibitemShut {NoStop}%
\bibitem [{\citenamefont {Takeoka}\ and\ \citenamefont {Wilde}(2016)}]{takeoka2016optimal}%
  \BibitemOpen
  \bibfield  {author} {\bibinfo {author} {\bibfnamefont {M.}~\bibnamefont {Takeoka}}\ and\ \bibinfo {author} {\bibfnamefont {M.~M.}\ \bibnamefont {Wilde}},\ }\href@noop {} {\bibinfo {title} {Optimal estimation and discrimination of excess noise in thermal and amplifier channels}} (\bibinfo {year} {2016}),\ \Eprint {https://arxiv.org/abs/1611.09165} {arXiv:1611.09165 [quant-ph]} \BibitemShut {NoStop}%
\bibitem [{\citenamefont {Puzzuoli}\ and\ \citenamefont {Watrous}(2017)}]{Puzzuoli2017}%
  \BibitemOpen
  \bibfield  {author} {\bibinfo {author} {\bibfnamefont {D.}~\bibnamefont {Puzzuoli}}\ and\ \bibinfo {author} {\bibfnamefont {J.}~\bibnamefont {Watrous}},\ }\bibfield  {title} {\bibinfo {title} {Ancilla dimension in quantum channel discrimination},\ }\href {https://doi.org/10.1007/s00023-016-0537-y} {\bibfield  {journal} {\bibinfo  {journal} {Annales Henri Poincar{\'e}}\ }\textbf {\bibinfo {volume} {18}},\ \bibinfo {pages} {1153} (\bibinfo {year} {2017})},\ \bibinfo {note} {arXiv:1604.08197}\BibitemShut {NoStop}%
\bibitem [{\citenamefont {Wilde}\ \emph {et~al.}(2020)\citenamefont {Wilde}, \citenamefont {Berta}, \citenamefont {Hirche},\ and\ \citenamefont {Kaur}}]{wilde2020amortized}%
  \BibitemOpen
  \bibfield  {author} {\bibinfo {author} {\bibfnamefont {M.~M.}\ \bibnamefont {Wilde}}, \bibinfo {author} {\bibfnamefont {M.}~\bibnamefont {Berta}}, \bibinfo {author} {\bibfnamefont {C.}~\bibnamefont {Hirche}},\ and\ \bibinfo {author} {\bibfnamefont {E.}~\bibnamefont {Kaur}},\ }\bibfield  {title} {\bibinfo {title} {Amortized channel divergence for asymptotic quantum channel discrimination},\ }\href {https://doi.org/10.1007/s11005-020-01297-7} {\bibfield  {journal} {\bibinfo  {journal} {Letters in Mathematical Physics}\ }\textbf {\bibinfo {volume} {110}},\ \bibinfo {pages} {2277} (\bibinfo {year} {2020})}\BibitemShut {NoStop}%
\bibitem [{\citenamefont {Wang}\ and\ \citenamefont {Wilde}(2019)}]{WW19}%
  \BibitemOpen
  \bibfield  {author} {\bibinfo {author} {\bibfnamefont {X.}~\bibnamefont {Wang}}\ and\ \bibinfo {author} {\bibfnamefont {M.~M.}\ \bibnamefont {Wilde}},\ }\bibfield  {title} {\bibinfo {title} {Resource theory of asymmetric distinguishability for quantum channels},\ }\href {https://doi.org/10.1103/PhysRevResearch.1.033169} {\bibfield  {journal} {\bibinfo  {journal} {Physical Review Research}\ }\textbf {\bibinfo {volume} {1}},\ \bibinfo {pages} {033169} (\bibinfo {year} {2019})}\BibitemShut {NoStop}%
\bibitem [{\citenamefont {Fang}\ \emph {et~al.}(2020)\citenamefont {Fang}, \citenamefont {Fawzi}, \citenamefont {Renner},\ and\ \citenamefont {Sutter}}]{FFRS20}%
  \BibitemOpen
  \bibfield  {author} {\bibinfo {author} {\bibfnamefont {K.}~\bibnamefont {Fang}}, \bibinfo {author} {\bibfnamefont {O.}~\bibnamefont {Fawzi}}, \bibinfo {author} {\bibfnamefont {R.}~\bibnamefont {Renner}},\ and\ \bibinfo {author} {\bibfnamefont {D.}~\bibnamefont {Sutter}},\ }\bibfield  {title} {\bibinfo {title} {Chain rule for the quantum relative entropy},\ }\href {https://doi.org/10.1103/PhysRevLett.124.100501} {\bibfield  {journal} {\bibinfo  {journal} {Physical Review Letters}\ }\textbf {\bibinfo {volume} {124}},\ \bibinfo {pages} {100501} (\bibinfo {year} {2020})}\BibitemShut {NoStop}%
\bibitem [{\citenamefont {Katariya}\ and\ \citenamefont {Wilde}(2021{\natexlab{a}})}]{KW21}%
  \BibitemOpen
  \bibfield  {author} {\bibinfo {author} {\bibfnamefont {V.}~\bibnamefont {Katariya}}\ and\ \bibinfo {author} {\bibfnamefont {M.~M.}\ \bibnamefont {Wilde}},\ }\bibfield  {title} {\bibinfo {title} {Evaluating the advantage of adaptive strategies for quantum channel distinguishability},\ }\href {https://doi.org/10.1103/PhysRevA.104.052406} {\bibfield  {journal} {\bibinfo  {journal} {Physical Review A}\ }\textbf {\bibinfo {volume} {104}},\ \bibinfo {pages} {052406} (\bibinfo {year} {2021}{\natexlab{a}})}\BibitemShut {NoStop}%
\bibitem [{\citenamefont {Fang}\ and\ \citenamefont {Fawzi}(2021)}]{FF21}%
  \BibitemOpen
  \bibfield  {author} {\bibinfo {author} {\bibfnamefont {K.}~\bibnamefont {Fang}}\ and\ \bibinfo {author} {\bibfnamefont {H.}~\bibnamefont {Fawzi}},\ }\bibfield  {title} {\bibinfo {title} {Geometric {R}ényi divergence and its applications in quantum channel capacities},\ }\href {https://doi.org/10.1007/s00220-021-04064-4} {\bibfield  {journal} {\bibinfo  {journal} {Communications in Mathematical Physics}\ }\textbf {\bibinfo {volume} {384}},\ \bibinfo {pages} {1615} (\bibinfo {year} {2021})}\BibitemShut {NoStop}%
\bibitem [{\citenamefont {Bergh}\ \emph {et~al.}(2022)\citenamefont {Bergh}, \citenamefont {Datta}, \citenamefont {Salzmann},\ and\ \citenamefont {Wilde}}]{bergh2022parallelization}%
  \BibitemOpen
  \bibfield  {author} {\bibinfo {author} {\bibfnamefont {B.}~\bibnamefont {Bergh}}, \bibinfo {author} {\bibfnamefont {N.}~\bibnamefont {Datta}}, \bibinfo {author} {\bibfnamefont {R.}~\bibnamefont {Salzmann}},\ and\ \bibinfo {author} {\bibfnamefont {M.~M.}\ \bibnamefont {Wilde}},\ }\href@noop {} {\bibinfo {title} {Parallelization of sequential quantum channel discrimination in the non-asymptotic regime}} (\bibinfo {year} {2022}),\ \Eprint {https://arxiv.org/abs/2206.08350} {arXiv:2206.08350 [quant-ph]} \BibitemShut {NoStop}%
\bibitem [{\citenamefont {Salek}\ \emph {et~al.}(2022)\citenamefont {Salek}, \citenamefont {Hayashi},\ and\ \citenamefont {Winter}}]{SHW22}%
  \BibitemOpen
  \bibfield  {author} {\bibinfo {author} {\bibfnamefont {F.}~\bibnamefont {Salek}}, \bibinfo {author} {\bibfnamefont {M.}~\bibnamefont {Hayashi}},\ and\ \bibinfo {author} {\bibfnamefont {A.}~\bibnamefont {Winter}},\ }\bibfield  {title} {\bibinfo {title} {Usefulness of adaptive strategies in asymptotic quantum channel discrimination},\ }\href {https://doi.org/10.1103/PhysRevA.105.022419} {\bibfield  {journal} {\bibinfo  {journal} {Physical Review A}\ }\textbf {\bibinfo {volume} {105}},\ \bibinfo {pages} {022419} (\bibinfo {year} {2022})}\BibitemShut {NoStop}%
\bibitem [{\citenamefont {Bergh}\ \emph {et~al.}(2023)\citenamefont {Bergh}, \citenamefont {Kochanowski}, \citenamefont {Salzmann},\ and\ \citenamefont {Datta}}]{bergh2023infinite}%
  \BibitemOpen
  \bibfield  {author} {\bibinfo {author} {\bibfnamefont {B.}~\bibnamefont {Bergh}}, \bibinfo {author} {\bibfnamefont {J.}~\bibnamefont {Kochanowski}}, \bibinfo {author} {\bibfnamefont {R.}~\bibnamefont {Salzmann}},\ and\ \bibinfo {author} {\bibfnamefont {N.}~\bibnamefont {Datta}},\ }\href@noop {} {\bibinfo {title} {Infinite dimensional asymmetric quantum channel discrimination}} (\bibinfo {year} {2023}),\ \Eprint {https://arxiv.org/abs/2308.12959} {arXiv:2308.12959 [quant-ph]} \BibitemShut {NoStop}%
\bibitem [{\citenamefont {Zhou}\ and\ \citenamefont {Jiang}(2021)}]{ZJ21}%
  \BibitemOpen
  \bibfield  {author} {\bibinfo {author} {\bibfnamefont {S.}~\bibnamefont {Zhou}}\ and\ \bibinfo {author} {\bibfnamefont {L.}~\bibnamefont {Jiang}},\ }\bibfield  {title} {\bibinfo {title} {Asymptotic theory of quantum channel estimation},\ }\href {https://doi.org/10.1103/PRXQuantum.2.010343} {\bibfield  {journal} {\bibinfo  {journal} {PRX Quantum}\ }\textbf {\bibinfo {volume} {2}},\ \bibinfo {pages} {010343} (\bibinfo {year} {2021})}\BibitemShut {NoStop}%
\bibitem [{\citenamefont {Liu}\ \emph {et~al.}(2023)\citenamefont {Liu}, \citenamefont {Hu}, \citenamefont {Yuan},\ and\ \citenamefont {Yang}}]{LHYY23}%
  \BibitemOpen
  \bibfield  {author} {\bibinfo {author} {\bibfnamefont {Q.}~\bibnamefont {Liu}}, \bibinfo {author} {\bibfnamefont {Z.}~\bibnamefont {Hu}}, \bibinfo {author} {\bibfnamefont {H.}~\bibnamefont {Yuan}},\ and\ \bibinfo {author} {\bibfnamefont {Y.}~\bibnamefont {Yang}},\ }\bibfield  {title} {\bibinfo {title} {Optimal strategies of quantum metrology with a strict hierarchy},\ }\href {https://doi.org/10.1103/PhysRevLett.130.070803} {\bibfield  {journal} {\bibinfo  {journal} {Physical Review Letters}\ }\textbf {\bibinfo {volume} {130}},\ \bibinfo {pages} {070803} (\bibinfo {year} {2023})}\BibitemShut {NoStop}%
\bibitem [{\citenamefont {Escher}\ \emph {et~al.}(2011)\citenamefont {Escher}, \citenamefont {de~Matos~Filho},\ and\ \citenamefont {Davidovich}}]{escher2011general}%
  \BibitemOpen
  \bibfield  {author} {\bibinfo {author} {\bibfnamefont {B.~M.}\ \bibnamefont {Escher}}, \bibinfo {author} {\bibfnamefont {R.~L.}\ \bibnamefont {de~Matos~Filho}},\ and\ \bibinfo {author} {\bibfnamefont {L.}~\bibnamefont {Davidovich}},\ }\bibfield  {title} {\bibinfo {title} {General framework for estimating the ultimate precision limit in noisy quantum-enhanced metrology},\ }\href {https://doi.org/https://doi.org/10.1038/nphys1958} {\bibfield  {journal} {\bibinfo  {journal} {Nature Physics}\ }\textbf {\bibinfo {volume} {7}},\ \bibinfo {pages} {406} (\bibinfo {year} {2011})}\BibitemShut {NoStop}%
\bibitem [{\citenamefont {Demkowicz-Dobrza{\'n}ski}\ \emph {et~al.}(2012)\citenamefont {Demkowicz-Dobrza{\'n}ski}, \citenamefont {Kolody{\'n}ski},\ and\ \citenamefont {Gu{\c{t}}{\u{a}}}}]{demkowicz2012elusive}%
  \BibitemOpen
  \bibfield  {author} {\bibinfo {author} {\bibfnamefont {R.}~\bibnamefont {Demkowicz-Dobrza{\'n}ski}}, \bibinfo {author} {\bibfnamefont {J.}~\bibnamefont {Kolody{\'n}ski}},\ and\ \bibinfo {author} {\bibfnamefont {M.}~\bibnamefont {Gu{\c{t}}{\u{a}}}},\ }\bibfield  {title} {\bibinfo {title} {The elusive {H}eisenberg limit in quantum-enhanced metrology},\ }\href {https://doi.org/10.1038/ncomms2067} {\bibfield  {journal} {\bibinfo  {journal} {Nature Communications}\ }\textbf {\bibinfo {volume} {3}},\ \bibinfo {pages} {1063} (\bibinfo {year} {2012})}\BibitemShut {NoStop}%
\bibitem [{\citenamefont {Demkowicz-Dobrza\ifmmode~\acute{n}\else \'{n}\fi{}ski}\ and\ \citenamefont {Maccone}(2014)}]{PhysRevLett.113.250801}%
  \BibitemOpen
  \bibfield  {author} {\bibinfo {author} {\bibfnamefont {R.}~\bibnamefont {Demkowicz-Dobrza\ifmmode~\acute{n}\else \'{n}\fi{}ski}}\ and\ \bibinfo {author} {\bibfnamefont {L.}~\bibnamefont {Maccone}},\ }\bibfield  {title} {\bibinfo {title} {Using entanglement against noise in quantum metrology},\ }\href {https://doi.org/10.1103/PhysRevLett.113.250801} {\bibfield  {journal} {\bibinfo  {journal} {Physical Review Letters}\ }\textbf {\bibinfo {volume} {113}},\ \bibinfo {pages} {250801} (\bibinfo {year} {2014})}\BibitemShut {NoStop}%
\bibitem [{\citenamefont {Katariya}\ and\ \citenamefont {Wilde}(2021{\natexlab{b}})}]{katariya2021geometric}%
  \BibitemOpen
  \bibfield  {author} {\bibinfo {author} {\bibfnamefont {V.}~\bibnamefont {Katariya}}\ and\ \bibinfo {author} {\bibfnamefont {M.~M.}\ \bibnamefont {Wilde}},\ }\bibfield  {title} {\bibinfo {title} {Geometric distinguishability measures limit quantum channel estimation and discrimination},\ }\href {https://doi.org/10.1007/s11128-021-02992-7} {\bibfield  {journal} {\bibinfo  {journal} {Quantum Information Processing}\ }\textbf {\bibinfo {volume} {20}},\ \bibinfo {pages} {78} (\bibinfo {year} {2021}{\natexlab{b}})}\BibitemShut {NoStop}%
\bibitem [{\citenamefont {Van~Trees}(2004)}]{van2004detection}%
  \BibitemOpen
  \bibfield  {author} {\bibinfo {author} {\bibfnamefont {H.~L.}\ \bibnamefont {Van~Trees}},\ }\href@noop {} {\emph {\bibinfo {title} {Detection, Estimation, and Modulation Theory, Part I: Detection, Estimation, and Linear Modulation Theory}}}\ (\bibinfo  {publisher} {John Wiley \& Sons},\ \bibinfo {year} {2004})\BibitemShut {NoStop}%
\bibitem [{\citenamefont {Tan}(2014)}]{tan2014}%
  \BibitemOpen
  \bibfield  {author} {\bibinfo {author} {\bibfnamefont {V.~Y.~F.}\ \bibnamefont {Tan}},\ }\bibfield  {title} {\bibinfo {title} {Asymptotic estimates in information theory with non-vanishing error probabilities},\ }\href {https://doi.org/10.1561/0100000086} {\bibfield  {journal} {\bibinfo  {journal} {Foundations and Trends in Communications and Information Theory}\ }\textbf {\bibinfo {volume} {11}},\ \bibinfo {pages} {1} (\bibinfo {year} {2014})}\BibitemShut {NoStop}%
\bibitem [{\citenamefont {Korostelev}\ and\ \citenamefont {Korosteleva}(2011)}]{korostelev2011mathematical}%
  \BibitemOpen
  \bibfield  {author} {\bibinfo {author} {\bibfnamefont {A.~P.}\ \bibnamefont {Korostelev}}\ and\ \bibinfo {author} {\bibfnamefont {O.}~\bibnamefont {Korosteleva}},\ }\href {https://doi.org/10.1090/gsm/119} {\emph {\bibinfo {title} {Mathematical Statistics: Asymptotic Minimax Theory}}},\ Vol.\ \bibinfo {volume} {119}\ (\bibinfo  {publisher} {American Mathematical Society},\ \bibinfo {year} {2011})\BibitemShut {NoStop}%
\bibitem [{\citenamefont {Canonne}(2022)}]{canonne2022topics}%
  \BibitemOpen
  \bibfield  {author} {\bibinfo {author} {\bibfnamefont {C.~L.}\ \bibnamefont {Canonne}},\ }\bibfield  {title} {\bibinfo {title} {Topics and techniques in distribution testing: A biased but representative sample},\ }\href {https://doi.org/10.1561/0100000114} {\bibfield  {journal} {\bibinfo  {journal} {Foundations and Trends in Communications and Information Theory}\ }\textbf {\bibinfo {volume} {19}},\ \bibinfo {pages} {1032} (\bibinfo {year} {2022})}\BibitemShut {NoStop}%
\bibitem [{\citenamefont {Chernoff}(1952)}]{chernoff1952measure}%
  \BibitemOpen
  \bibfield  {author} {\bibinfo {author} {\bibfnamefont {H.}~\bibnamefont {Chernoff}},\ }\bibfield  {title} {\bibinfo {title} {A measure of asymptotic efficiency for tests of a hypothesis based on the sum of observations},\ }\href {https://doi.org/10.1214/aoms/1177729330} {\bibfield  {journal} {\bibinfo  {journal} {The Annals of Mathematical Statistics}\ }\textbf {\bibinfo {volume} {23}},\ \bibinfo {pages} {493} (\bibinfo {year} {1952})}\BibitemShut {NoStop}%
\bibitem [{\citenamefont {Sacchi}(2005)}]{PhysRevA.71.062340}%
  \BibitemOpen
  \bibfield  {author} {\bibinfo {author} {\bibfnamefont {M.~F.}\ \bibnamefont {Sacchi}},\ }\bibfield  {title} {\bibinfo {title} {Optimal discrimination of quantum operations},\ }\href {https://doi.org/10.1103/PhysRevA.71.062340} {\bibfield  {journal} {\bibinfo  {journal} {Physical Review A}\ }\textbf {\bibinfo {volume} {71}},\ \bibinfo {pages} {062340} (\bibinfo {year} {2005})}\BibitemShut {NoStop}%
\bibitem [{\citenamefont {Lloyd}(2008)}]{lloyd2008enhanced}%
  \BibitemOpen
  \bibfield  {author} {\bibinfo {author} {\bibfnamefont {S.}~\bibnamefont {Lloyd}},\ }\bibfield  {title} {\bibinfo {title} {Enhanced sensitivity of photodetection via quantum illumination},\ }\href {https://doi.org/10.1126/science.1160627} {\bibfield  {journal} {\bibinfo  {journal} {Science}\ }\textbf {\bibinfo {volume} {321}},\ \bibinfo {pages} {1463} (\bibinfo {year} {2008})}\BibitemShut {NoStop}%
\bibitem [{\citenamefont {Tan}\ \emph {et~al.}(2008)\citenamefont {Tan}, \citenamefont {Erkmen}, \citenamefont {Giovannetti}, \citenamefont {Guha}, \citenamefont {Lloyd}, \citenamefont {Maccone}, \citenamefont {Pirandola},\ and\ \citenamefont {Shapiro}}]{PhysRevLett.101.253601}%
  \BibitemOpen
  \bibfield  {author} {\bibinfo {author} {\bibfnamefont {S.-H.}\ \bibnamefont {Tan}}, \bibinfo {author} {\bibfnamefont {B.~I.}\ \bibnamefont {Erkmen}}, \bibinfo {author} {\bibfnamefont {V.}~\bibnamefont {Giovannetti}}, \bibinfo {author} {\bibfnamefont {S.}~\bibnamefont {Guha}}, \bibinfo {author} {\bibfnamefont {S.}~\bibnamefont {Lloyd}}, \bibinfo {author} {\bibfnamefont {L.}~\bibnamefont {Maccone}}, \bibinfo {author} {\bibfnamefont {S.}~\bibnamefont {Pirandola}},\ and\ \bibinfo {author} {\bibfnamefont {J.~H.}\ \bibnamefont {Shapiro}},\ }\bibfield  {title} {\bibinfo {title} {Quantum illumination with {G}aussian states},\ }\href {https://doi.org/10.1103/PhysRevLett.101.253601} {\bibfield  {journal} {\bibinfo  {journal} {Physical Review Letters}\ }\textbf {\bibinfo {volume} {101}},\ \bibinfo {pages} {253601} (\bibinfo {year} {2008})}\BibitemShut {NoStop}%
\bibitem [{\citenamefont {Wilde}\ \emph {et~al.}(2017)\citenamefont {Wilde}, \citenamefont {Tomamichel}, \citenamefont {Lloyd},\ and\ \citenamefont {Berta}}]{WTLB17}%
  \BibitemOpen
  \bibfield  {author} {\bibinfo {author} {\bibfnamefont {M.~M.}\ \bibnamefont {Wilde}}, \bibinfo {author} {\bibfnamefont {M.}~\bibnamefont {Tomamichel}}, \bibinfo {author} {\bibfnamefont {S.}~\bibnamefont {Lloyd}},\ and\ \bibinfo {author} {\bibfnamefont {M.}~\bibnamefont {Berta}},\ }\bibfield  {title} {\bibinfo {title} {Gaussian hypothesis testing and quantum illumination},\ }\href {https://doi.org/10.1103/PhysRevLett.119.120501} {\bibfield  {journal} {\bibinfo  {journal} {Physical Review Letters}\ }\textbf {\bibinfo {volume} {119}},\ \bibinfo {pages} {120501} (\bibinfo {year} {2017})}\BibitemShut {NoStop}%
\bibitem [{\citenamefont {Strassen}(1962)}]{strassen1962asymptotische}%
  \BibitemOpen
  \bibfield  {author} {\bibinfo {author} {\bibfnamefont {V.}~\bibnamefont {Strassen}},\ }\bibfield  {title} {\bibinfo {title} {Asymptotische abschatzugen in shannon's informationstheorie},\ }in\ \href@noop {} {\emph {\bibinfo {booktitle} {Transactions of the Third Prague Conference on Information Theory etc, 1962. Czechoslovak Academy of Sciences, Prague}}}\ (\bibinfo {year} {1962})\ pp.\ \bibinfo {pages} {689--723},\ \bibinfo {note} {{E}nglish translation available at \url{https://pi.math.cornell.edu/~pmlut/strassen.pdf}}\BibitemShut {NoStop}%
\bibitem [{\citenamefont {Stein}(shed)}]{stein_unpublished}%
  \BibitemOpen
  \bibfield  {author} {\bibinfo {author} {\bibfnamefont {C.}~\bibnamefont {Stein}},\ }\href@noop {} {\bibinfo {title} {Information and {{Comparison}} of {{Experiments}}}} (\bibinfo {year} {unpublished}),\ \bibinfo {note} {{C}harles Stein papers (SC1224), Box 12, Folder 7, Department of Special Collections and University Archives, Stanford University Libraries}\BibitemShut {NoStop}%
\bibitem [{\citenamefont {Chernoff}(1956)}]{chernoff_1956}%
  \BibitemOpen
  \bibfield  {author} {\bibinfo {author} {\bibfnamefont {H.}~\bibnamefont {Chernoff}},\ }\bibfield  {title} {\bibinfo {title} {Large-sample theory: Parametric case},\ }\href {https://doi.org/10.1214/aoms/1177728347} {\bibfield  {journal} {\bibinfo  {journal} {The Annals of Mathematical Statistics}\ }\textbf {\bibinfo {volume} {27}},\ \bibinfo {pages} {1} (\bibinfo {year} {1956})}\BibitemShut {NoStop}%
\bibitem [{\citenamefont {Matsumoto}(2010)}]{matsumoto2010metric}%
  \BibitemOpen
  \bibfield  {author} {\bibinfo {author} {\bibfnamefont {K.}~\bibnamefont {Matsumoto}},\ }\href@noop {} {\bibinfo {title} {On metric of quantum channel spaces}} (\bibinfo {year} {2010}),\ \Eprint {https://arxiv.org/abs/1006.0300} {arXiv:1006.0300 [quant-ph]} \BibitemShut {NoStop}%
\bibitem [{\citenamefont {Holland}\ and\ \citenamefont {Burnett}(1993)}]{HB93}%
  \BibitemOpen
  \bibfield  {author} {\bibinfo {author} {\bibfnamefont {M.~J.}\ \bibnamefont {Holland}}\ and\ \bibinfo {author} {\bibfnamefont {K.}~\bibnamefont {Burnett}},\ }\bibfield  {title} {\bibinfo {title} {Interferometric detection of optical phase shifts at the {H}eisenberg limit},\ }\href {https://doi.org/10.1103/PhysRevLett.71.1355} {\bibfield  {journal} {\bibinfo  {journal} {Physical Review Letters}\ }\textbf {\bibinfo {volume} {71}},\ \bibinfo {pages} {1355} (\bibinfo {year} {1993})}\BibitemShut {NoStop}%
\bibitem [{\citenamefont {Wiseman}\ and\ \citenamefont {Killip}(1998)}]{Wiseman_1998}%
  \BibitemOpen
  \bibfield  {author} {\bibinfo {author} {\bibfnamefont {H.~M.}\ \bibnamefont {Wiseman}}\ and\ \bibinfo {author} {\bibfnamefont {R.~B.}\ \bibnamefont {Killip}},\ }\bibfield  {title} {\bibinfo {title} {Adaptive single-shot phase measurements: The full quantum theory},\ }\href {https://doi.org/10.1103/PhysRevA.57.2169} {\bibfield  {journal} {\bibinfo  {journal} {Physical Review A}\ }\textbf {\bibinfo {volume} {57}},\ \bibinfo {pages} {2169} (\bibinfo {year} {1998})}\BibitemShut {NoStop}%
\bibitem [{\citenamefont {Demkowicz-Dobrza{\'n}ski}\ \emph {et~al.}(2015)\citenamefont {Demkowicz-Dobrza{\'n}ski}, \citenamefont {Jarzyna},\ and\ \citenamefont {Ko{\l}ody{\'n}ski}}]{demkowicz2015quantum}%
  \BibitemOpen
  \bibfield  {author} {\bibinfo {author} {\bibfnamefont {R.}~\bibnamefont {Demkowicz-Dobrza{\'n}ski}}, \bibinfo {author} {\bibfnamefont {M.}~\bibnamefont {Jarzyna}},\ and\ \bibinfo {author} {\bibfnamefont {J.}~\bibnamefont {Ko{\l}ody{\'n}ski}},\ }\bibfield  {title} {\bibinfo {title} {Quantum limits in optical interferometry},\ }\href {https://doi.org/https://doi.org/10.1016/bs.po.2015.02.003} {\bibfield  {journal} {\bibinfo  {journal} {Progress in Optics}\ }\textbf {\bibinfo {volume} {60}},\ \bibinfo {pages} {345} (\bibinfo {year} {2015})}\BibitemShut {NoStop}%
\bibitem [{\citenamefont {Dinani}\ and\ \citenamefont {Berry}(2014)}]{PhysRevA.90.023856}%
  \BibitemOpen
  \bibfield  {author} {\bibinfo {author} {\bibfnamefont {H.~T.}\ \bibnamefont {Dinani}}\ and\ \bibinfo {author} {\bibfnamefont {D.~W.}\ \bibnamefont {Berry}},\ }\bibfield  {title} {\bibinfo {title} {Loss-resistant unambiguous phase measurement},\ }\href {https://doi.org/10.1103/PhysRevA.90.023856} {\bibfield  {journal} {\bibinfo  {journal} {Physical Review A}\ }\textbf {\bibinfo {volume} {90}},\ \bibinfo {pages} {023856} (\bibinfo {year} {2014})}\BibitemShut {NoStop}%
\bibitem [{\citenamefont {Rodr{\'\i}guez-Garc{\'\i}a}\ \emph {et~al.}(2022)\citenamefont {Rodr{\'\i}guez-Garc{\'\i}a}, \citenamefont {DiMario}, \citenamefont {Barberis-Blostein},\ and\ \citenamefont {Becerra}}]{rodriguez2022determination}%
  \BibitemOpen
  \bibfield  {author} {\bibinfo {author} {\bibfnamefont {M.~A.}\ \bibnamefont {Rodr{\'\i}guez-Garc{\'\i}a}}, \bibinfo {author} {\bibfnamefont {M.~T.}\ \bibnamefont {DiMario}}, \bibinfo {author} {\bibfnamefont {P.}~\bibnamefont {Barberis-Blostein}},\ and\ \bibinfo {author} {\bibfnamefont {F.~E.}\ \bibnamefont {Becerra}},\ }\bibfield  {title} {\bibinfo {title} {Determination of the asymptotic limits of adaptive photon counting measurements for coherent-state optical phase estimation},\ }\href {https://doi.org/https://doi.org/10.1038/s41534-022-00601-8} {\bibfield  {journal} {\bibinfo  {journal} {npj Quantum Information}\ }\textbf {\bibinfo {volume} {8}},\ \bibinfo {pages} {94} (\bibinfo {year} {2022})}\BibitemShut {NoStop}%
\bibitem [{\citenamefont {Cooper}\ \emph {et~al.}(2013)\citenamefont {Cooper}, \citenamefont {Wright}, \citenamefont {S{\"o}ller},\ and\ \citenamefont {Smith}}]{cooper2013experimental}%
  \BibitemOpen
  \bibfield  {author} {\bibinfo {author} {\bibfnamefont {M.}~\bibnamefont {Cooper}}, \bibinfo {author} {\bibfnamefont {L.~J.}\ \bibnamefont {Wright}}, \bibinfo {author} {\bibfnamefont {C.}~\bibnamefont {S{\"o}ller}},\ and\ \bibinfo {author} {\bibfnamefont {B.~J.}\ \bibnamefont {Smith}},\ }\bibfield  {title} {\bibinfo {title} {Experimental generation of multi-photon {F}ock states},\ }\href {https://doi.org/https://doi.org/10.1364/OE.21.005309} {\bibfield  {journal} {\bibinfo  {journal} {Optics Express}\ }\textbf {\bibinfo {volume} {21}},\ \bibinfo {pages} {5309} (\bibinfo {year} {2013})}\BibitemShut {NoStop}%
\bibitem [{\citenamefont {Dowling}(2008)}]{dowling2008quantum}%
  \BibitemOpen
  \bibfield  {author} {\bibinfo {author} {\bibfnamefont {J.~P.}\ \bibnamefont {Dowling}},\ }\bibfield  {title} {\bibinfo {title} {Quantum optical metrology -- the lowdown on high-{N00N} states},\ }\href {https://doi.org/https://doi.org/10.1080/00107510802091298} {\bibfield  {journal} {\bibinfo  {journal} {Contemporary Physics}\ }\textbf {\bibinfo {volume} {49}},\ \bibinfo {pages} {125} (\bibinfo {year} {2008})}\BibitemShut {NoStop}%
\bibitem [{\citenamefont {Hoffman}(2007)}]{hoffman2007banach}%
  \BibitemOpen
  \bibfield  {author} {\bibinfo {author} {\bibfnamefont {K.}~\bibnamefont {Hoffman}},\ }\href@noop {} {\emph {\bibinfo {title} {Banach Spaces of Analytic Functions}}},\ Dover Books on Mathematics\ (\bibinfo  {publisher} {Dover Publications},\ \bibinfo {year} {2007})\BibitemShut {NoStop}%
\bibitem [{\citenamefont {Olivares}\ \emph {et~al.}(2013)\citenamefont {Olivares}, \citenamefont {Cialdi}, \citenamefont {Castelli},\ and\ \citenamefont {Paris}}]{PhysRevA.87.050303}%
  \BibitemOpen
  \bibfield  {author} {\bibinfo {author} {\bibfnamefont {S.}~\bibnamefont {Olivares}}, \bibinfo {author} {\bibfnamefont {S.}~\bibnamefont {Cialdi}}, \bibinfo {author} {\bibfnamefont {F.}~\bibnamefont {Castelli}},\ and\ \bibinfo {author} {\bibfnamefont {M.~G.~A.}\ \bibnamefont {Paris}},\ }\bibfield  {title} {\bibinfo {title} {Homodyne detection as a near-optimum receiver for phase-shift-keyed binary communication in the presence of phase diffusion},\ }\href {https://doi.org/10.1103/PhysRevA.87.050303} {\bibfield  {journal} {\bibinfo  {journal} {Physical Review A}\ }\textbf {\bibinfo {volume} {87}},\ \bibinfo {pages} {050303} (\bibinfo {year} {2013})}\BibitemShut {NoStop}%
\bibitem [{\citenamefont {Martin}\ \emph {et~al.}(2020)\citenamefont {Martin}, \citenamefont {Livingston}, \citenamefont {Hacohen-Gourgy}, \citenamefont {Wiseman},\ and\ \citenamefont {Siddiqi}}]{martin2020implementation}%
  \BibitemOpen
  \bibfield  {author} {\bibinfo {author} {\bibfnamefont {L.~S.}\ \bibnamefont {Martin}}, \bibinfo {author} {\bibfnamefont {W.~P.}\ \bibnamefont {Livingston}}, \bibinfo {author} {\bibfnamefont {S.}~\bibnamefont {Hacohen-Gourgy}}, \bibinfo {author} {\bibfnamefont {H.~M.}\ \bibnamefont {Wiseman}},\ and\ \bibinfo {author} {\bibfnamefont {I.}~\bibnamefont {Siddiqi}},\ }\bibfield  {title} {\bibinfo {title} {Implementation of a canonical phase measurement with quantum feedback},\ }\href {https://doi.org/https://doi.org/10.1038/s41567-020-0939-0} {\bibfield  {journal} {\bibinfo  {journal} {Nature Physics}\ }\textbf {\bibinfo {volume} {16}},\ \bibinfo {pages} {1046} (\bibinfo {year} {2020})}\BibitemShut {NoStop}%
\bibitem [{\citenamefont {Genoni}\ \emph {et~al.}(2011)\citenamefont {Genoni}, \citenamefont {Olivares},\ and\ \citenamefont {Paris}}]{PhysRevLett.106.153603}%
  \BibitemOpen
  \bibfield  {author} {\bibinfo {author} {\bibfnamefont {M.~G.}\ \bibnamefont {Genoni}}, \bibinfo {author} {\bibfnamefont {S.}~\bibnamefont {Olivares}},\ and\ \bibinfo {author} {\bibfnamefont {M.~G.~A.}\ \bibnamefont {Paris}},\ }\bibfield  {title} {\bibinfo {title} {Optical phase estimation in the presence of phase diffusion},\ }\href {https://doi.org/10.1103/PhysRevLett.106.153603} {\bibfield  {journal} {\bibinfo  {journal} {Physical Review Letters}\ }\textbf {\bibinfo {volume} {106}},\ \bibinfo {pages} {153603} (\bibinfo {year} {2011})}\BibitemShut {NoStop}%
\bibitem [{\citenamefont {Trapani}\ \emph {et~al.}(2015)\citenamefont {Trapani}, \citenamefont {Teklu}, \citenamefont {Olivares},\ and\ \citenamefont {Paris}}]{PhysRevA.92.012317}%
  \BibitemOpen
  \bibfield  {author} {\bibinfo {author} {\bibfnamefont {J.}~\bibnamefont {Trapani}}, \bibinfo {author} {\bibfnamefont {B.}~\bibnamefont {Teklu}}, \bibinfo {author} {\bibfnamefont {S.}~\bibnamefont {Olivares}},\ and\ \bibinfo {author} {\bibfnamefont {M.~G.~A.}\ \bibnamefont {Paris}},\ }\bibfield  {title} {\bibinfo {title} {Quantum phase communication channels in the presence of static and dynamical phase diffusion},\ }\href {https://doi.org/10.1103/PhysRevA.92.012317} {\bibfield  {journal} {\bibinfo  {journal} {Physical Review A}\ }\textbf {\bibinfo {volume} {92}},\ \bibinfo {pages} {012317} (\bibinfo {year} {2015})}\BibitemShut {NoStop}%
\bibitem [{\citenamefont {Luo}\ \emph {et~al.}(2020)\citenamefont {Luo}, \citenamefont {Zhan},\ and\ \citenamefont {Jonckheere}}]{LZJ20}%
  \BibitemOpen
  \bibfield  {author} {\bibinfo {author} {\bibfnamefont {Z.}~\bibnamefont {Luo}}, \bibinfo {author} {\bibfnamefont {Y.}~\bibnamefont {Zhan}},\ and\ \bibinfo {author} {\bibfnamefont {E.}~\bibnamefont {Jonckheere}},\ }\bibfield  {title} {\bibinfo {title} {Analysis on functions and characteristics of the {R}ician phase distribution},\ }in\ \href {https://doi.org/10.1109/ICCC49849.2020.9238805} {\emph {\bibinfo {booktitle} {2020 IEEE/CIC International Conference on Communications in China (ICCC)}}}\ (\bibinfo {year} {2020})\ pp.\ \bibinfo {pages} {306--311}\BibitemShut {NoStop}%
\bibitem [{\citenamefont {Leviant}\ \emph {et~al.}(2022)\citenamefont {Leviant}, \citenamefont {Xu}, \citenamefont {Jiang},\ and\ \citenamefont {Rosenblum}}]{Leviant2022}%
  \BibitemOpen
  \bibfield  {author} {\bibinfo {author} {\bibfnamefont {P.}~\bibnamefont {Leviant}}, \bibinfo {author} {\bibfnamefont {Q.}~\bibnamefont {Xu}}, \bibinfo {author} {\bibfnamefont {L.}~\bibnamefont {Jiang}},\ and\ \bibinfo {author} {\bibfnamefont {S.}~\bibnamefont {Rosenblum}},\ }\bibfield  {title} {\bibinfo {title} {Quantum capacity and codes for the bosonic loss-dephasing channel},\ }\href {https://doi.org/10.22331/q-2022-09-29-821} {\bibfield  {journal} {\bibinfo  {journal} {{Quantum}}\ }\textbf {\bibinfo {volume} {6}},\ \bibinfo {pages} {821} (\bibinfo {year} {2022})}\BibitemShut {NoStop}%
\bibitem [{\citenamefont {Mele}\ \emph {et~al.}(2024)\citenamefont {Mele}, \citenamefont {Salek}, \citenamefont {Giovannetti},\ and\ \citenamefont {Lami}}]{loss-dephasing}%
  \BibitemOpen
  \bibfield  {author} {\bibinfo {author} {\bibfnamefont {F.~A.}\ \bibnamefont {Mele}}, \bibinfo {author} {\bibfnamefont {F.}~\bibnamefont {Salek}}, \bibinfo {author} {\bibfnamefont {V.}~\bibnamefont {Giovannetti}},\ and\ \bibinfo {author} {\bibfnamefont {L.}~\bibnamefont {Lami}},\ }\href {https://doi.org/10.48550/arXiv.2401.15634} {\bibinfo {title} {Quantum communication on the bosonic loss-dephasing channel}} (\bibinfo {year} {2024}),\ \Eprint {https://arxiv.org/abs/2401.15634} {arXiv:2401.15634} \BibitemShut {NoStop}%
\bibitem [{\citenamefont {Winter}(2017)}]{VV-diamond}%
  \BibitemOpen
  \bibfield  {author} {\bibinfo {author} {\bibfnamefont {A.}~\bibnamefont {Winter}},\ }\href {https://doi.org/10.48550/arXiv.1712.10267} {\bibinfo {title} {Energy-constrained diamond norm with applications to the uniform continuity of continuous variable channel capacities}} (\bibinfo {year} {2017}),\ \Eprint {https://arxiv.org/abs/1712.10267} {arXiv:1712.10267} \BibitemShut {NoStop}%
\bibitem [{\citenamefont {Neyman}\ and\ \citenamefont {Pearson}(1933)}]{neyman1933ix}%
  \BibitemOpen
  \bibfield  {author} {\bibinfo {author} {\bibfnamefont {J.}~\bibnamefont {Neyman}}\ and\ \bibinfo {author} {\bibfnamefont {E.~S.}\ \bibnamefont {Pearson}},\ }\bibfield  {title} {\bibinfo {title} {{IX. O}n the problem of the most efficient tests of statistical hypotheses},\ }\href {https://doi.org/10.1098/rsta.1933.0009} {\bibfield  {journal} {\bibinfo  {journal} {Philosophical Transactions of the Royal Society of London A}\ }\textbf {\bibinfo {volume} {231}},\ \bibinfo {pages} {289} (\bibinfo {year} {1933})}\BibitemShut {NoStop}%
\bibitem [{\citenamefont {Renes}(2016)}]{Renes2016}%
  \BibitemOpen
  \bibfield  {author} {\bibinfo {author} {\bibfnamefont {J.~M.}\ \bibnamefont {Renes}},\ }\bibfield  {title} {\bibinfo {title} {{Relative submajorization and its use in quantum resource theories}},\ }\href {https://doi.org/10.1063/1.4972295} {\bibfield  {journal} {\bibinfo  {journal} {Journal of Mathematical Physics}\ }\textbf {\bibinfo {volume} {57}},\ \bibinfo {pages} {122202} (\bibinfo {year} {2016})}\BibitemShut {NoStop}%
\bibitem [{\citenamefont {Buscemi}\ and\ \citenamefont {Gour}(2017)}]{BG2017}%
  \BibitemOpen
  \bibfield  {author} {\bibinfo {author} {\bibfnamefont {F.}~\bibnamefont {Buscemi}}\ and\ \bibinfo {author} {\bibfnamefont {G.}~\bibnamefont {Gour}},\ }\bibfield  {title} {\bibinfo {title} {Quantum relative {L}orenz curves},\ }\href {https://doi.org/10.1103/PhysRevA.95.012110} {\bibfield  {journal} {\bibinfo  {journal} {Physical Review A}\ }\textbf {\bibinfo {volume} {95}},\ \bibinfo {pages} {012110} (\bibinfo {year} {2017})}\BibitemShut {NoStop}%
\bibitem [{\citenamefont {Zygmund}(1968)}]{ZYGMUND}%
  \BibitemOpen
  \bibfield  {author} {\bibinfo {author} {\bibfnamefont {A.}~\bibnamefont {Zygmund}},\ }\href@noop {} {\emph {\bibinfo {title} {Trigonometric series: {V}ols. {I}, {II}}}},\ \bibinfo {edition} {2nd}\ ed.\ (\bibinfo  {publisher} {Cambridge University Press, London-New York},\ \bibinfo {year} {1968})\BibitemShut {NoStop}%
\bibitem [{\citenamefont {Serafini}(2017)}]{serafini2017quantum}%
  \BibitemOpen
  \bibfield  {author} {\bibinfo {author} {\bibfnamefont {A.}~\bibnamefont {Serafini}},\ }\href {https://doi.org/10.1201/9781315118727} {\emph {\bibinfo {title} {Quantum Continuous Variables: A Primer of Theoretical Methods}}}\ (\bibinfo  {publisher} {CRC press},\ \bibinfo {year} {2017})\BibitemShut {NoStop}%
\bibitem [{\citenamefont {Han}\ and\ \citenamefont {Kobayashi}(1989)}]{HK89}%
  \BibitemOpen
  \bibfield  {author} {\bibinfo {author} {\bibfnamefont {T.~S.}\ \bibnamefont {Han}}\ and\ \bibinfo {author} {\bibfnamefont {K.}~\bibnamefont {Kobayashi}},\ }\bibfield  {title} {\bibinfo {title} {The strong converse theorem for hypothesis testing},\ }\href {https://doi.org/10.1109/18.42188} {\bibfield  {journal} {\bibinfo  {journal} {IEEE Transactions on Information Theory}\ }\textbf {\bibinfo {volume} {35}},\ \bibinfo {pages} {178} (\bibinfo {year} {1989})}\BibitemShut {NoStop}%
\bibitem [{\citenamefont {Hoeffding}(1965)}]{Hoeffding65}%
  \BibitemOpen
  \bibfield  {author} {\bibinfo {author} {\bibfnamefont {W.}~\bibnamefont {Hoeffding}},\ }\bibfield  {title} {\bibinfo {title} {Asymptotically optimal tests for multinomial distributions},\ }\href {https://doi.org/10.1214/aoms/1177700150} {\bibfield  {journal} {\bibinfo  {journal} {The Annals of Mathematical Statistics}\ }\textbf {\bibinfo {volume} {36}},\ \bibinfo {pages} {369} (\bibinfo {year} {1965})}\BibitemShut {NoStop}%
\bibitem [{\citenamefont {Salikhov}(1973)}]{salikhov1973asymptotic}%
  \BibitemOpen
  \bibfield  {author} {\bibinfo {author} {\bibfnamefont {N.~P.}\ \bibnamefont {Salikhov}},\ }\bibfield  {title} {\bibinfo {title} {Asymptotic properties of error probabilities of tests for distinguishing between several multinomial testing schemes},\ }\href@noop {} {\bibfield  {journal} {\bibinfo  {journal} {Doklady Akademii Nauk}\ }\textbf {\bibinfo {volume} {209}},\ \bibinfo {pages} {54} (\bibinfo {year} {1973})}\BibitemShut {NoStop}%
\bibitem [{\citenamefont {Torgersen}(1981)}]{torgersen1981measures}%
  \BibitemOpen
  \bibfield  {author} {\bibinfo {author} {\bibfnamefont {E.~N.}\ \bibnamefont {Torgersen}},\ }\bibfield  {title} {\bibinfo {title} {Measures of information based on comparison with total information and with total ignorance},\ }\href@noop {} {\bibfield  {journal} {\bibinfo  {journal} {The Annals of Statistics}\ }\textbf {\bibinfo {volume} {9}},\ \bibinfo {pages} {638} (\bibinfo {year} {1981})}\BibitemShut {NoStop}%
\bibitem [{\citenamefont {Leang}\ and\ \citenamefont {Johnson}(1997)}]{Leang1997}%
  \BibitemOpen
  \bibfield  {author} {\bibinfo {author} {\bibfnamefont {C.~C.}\ \bibnamefont {Leang}}\ and\ \bibinfo {author} {\bibfnamefont {D.~H.}\ \bibnamefont {Johnson}},\ }\bibfield  {title} {\bibinfo {title} {On the asymptotics of $m$-hypothesis {B}ayesian detection},\ }\href {https://doi.org/10.1109/18.567705} {\bibfield  {journal} {\bibinfo  {journal} {IEEE Transactions on Information Theory}\ }\textbf {\bibinfo {volume} {43}},\ \bibinfo {pages} {280} (\bibinfo {year} {1997})}\BibitemShut {NoStop}%
\bibitem [{\citenamefont {Salikhov}(1999)}]{salikhov1999one}%
  \BibitemOpen
  \bibfield  {author} {\bibinfo {author} {\bibfnamefont {N.~P.}\ \bibnamefont {Salikhov}},\ }\bibfield  {title} {\bibinfo {title} {On one generalization of {C}hernov's distance},\ }\href@noop {} {\bibfield  {journal} {\bibinfo  {journal} {Theory of Probability \& Its Applications}\ }\textbf {\bibinfo {volume} {43}},\ \bibinfo {pages} {239} (\bibinfo {year} {1999})}\BibitemShut {NoStop}%
\bibitem [{\citenamefont {Salikhov}(2003)}]{salikhov2003optimal}%
  \BibitemOpen
  \bibfield  {author} {\bibinfo {author} {\bibfnamefont {N.~P.}\ \bibnamefont {Salikhov}},\ }\bibfield  {title} {\bibinfo {title} {Optimal sequences of tests for several polynomial schemes of trials},\ }\href@noop {} {\bibfield  {journal} {\bibinfo  {journal} {Theory of Probability \& Its Applications}\ }\textbf {\bibinfo {volume} {47}},\ \bibinfo {pages} {286} (\bibinfo {year} {2003})}\BibitemShut {NoStop}%
\bibitem [{\citenamefont {Mishra}\ \emph {et~al.}(2023)\citenamefont {Mishra}, \citenamefont {Nussbaum},\ and\ \citenamefont {Wilde}}]{mishra2023optimal}%
  \BibitemOpen
  \bibfield  {author} {\bibinfo {author} {\bibfnamefont {H.~K.}\ \bibnamefont {Mishra}}, \bibinfo {author} {\bibfnamefont {M.}~\bibnamefont {Nussbaum}},\ and\ \bibinfo {author} {\bibfnamefont {M.~M.}\ \bibnamefont {Wilde}},\ }\href@noop {} {\bibinfo {title} {On the optimal error exponents for classical and quantum antidistinguishability}} (\bibinfo {year} {2023}),\ \Eprint {https://arxiv.org/abs/2309.03723} {arXiv:2309.03723 [quant-ph]} \BibitemShut {NoStop}%
\end{thebibliography}%

\appendix

\section{Equivalence of hypothesis testing regions}

\label{app:hypo-test-reg-equiv}

In this appendix, we prove that the following equality holds for every pair of bosonic dephasing channels $\mathcal{D}_p$ and $\mathcal{D}_q$:
\begin{equation}
    \left\{\left(\alpha_n (\mathcal{A}), \beta_n (\mathcal{A})\right)\right\}_{\mathcal{A}}=\left\{\left(\alpha_n (t), \beta_n (t)\right)\right\}_{t},
    \label{eq:hypo-test-reg-equiv}
\end{equation}
where $\alpha_n (\mathcal{A})$ and $\beta_n (\mathcal{A})$ are defined in~\eqref{eq:ch-disc-err-probs} and  $\alpha_n(t)$ and $\beta_n(t)$ are defined in~\eqref{eq:cl-disc-err-prob-a}--\eqref{eq:cl-disc-err-prob-b}. These sets are known as hypothesis testing regions and have been studied for a long time in statistics \cite{neyman1933ix} (see also \cite{Renes2016,BG2017} for more recent works in quantum information). Furthermore, the quantities in \eqref{eq:symm-err-cl} and \eqref{eq:asymm-err-prob-cl} can be understood as various boundary points of this hypothesis testing region. As such, the equalities in \eqref{eq:symm-err-cl} and \eqref{eq:asymm-err-prob-cl} follow as a consequence of \eqref{eq:hypo-test-reg-equiv}.

The containment
\begin{equation}
    \left\{\left(\alpha_n (\mathcal{A}), \beta_n (\mathcal{A})\right)\right\}_{\mathcal{A}}
    \subseteq
    \left\{\left(\alpha_n (t), \beta_n (t)\right)\right\}_{t},
\end{equation}
follows as a consequence of the same reasoning used to establish \eqref{eq:classical-sym-err-lower} and \eqref{eq:classical-asym-err-lower} and can again be understood by examining Figure~\ref{fig:env}. Indeed, every adaptive strategy $\mathcal{A}$ for distinguishing the BDCs $\mathcal{D}_p$ and $\mathcal{D}_q$ can be understood as a particular classical test $t$ for distinguishing the underlying densities $p$ and $q$. As such, the region of achievable pairs using quantum adaptive strategies is contained in the region of achievable error pairs for the underlying densities.

The other containment
\begin{equation}
    \left\{\left(\alpha_n (t), \beta_n (t)\right)\right\}_{t}
    \subseteq
    \left\{\left(\alpha_n (\mathcal{A}), \beta_n (\mathcal{A})\right)\right\}_{\mathcal{A}}
\end{equation}
follows by employing either one of the two strategies from Section~\ref{sec:phot-num-sup-scheme} or \ref{sec:coh-state-scheme}. Indeed, in the large energy limit, it is possible to employ either one of these two strategies and obtain $n$ samples of the underlying densities $p^{\otimes n}$ or~$q^{\otimes n}$. Once the samples are in hand, one can then perform an arbitrary classical test $t$ on them.

\section{Calculations for photon-number-superposition method}

\label{app:phot-num-fourier-dist}

Recall the definitions of $|+_{d}\rangle$ and $|u_k\rangle$ in~\eqref{eq:+_d} and~\eqref{eq:fourier-basis-states}, respectively. Measuring the state $e^{- i\hat{n}\phi}|+_{d}\rangle$ in the Fourier basis leads to the following
outcome probabilities:
\begin{align}
& \left\vert \langle u_{k}|e^{- i\hat{n}\phi}|+_{d}\rangle\right\vert ^{2}
 \notag \\
& =\left\vert \left(  \frac{1}{\sqrt{d}}\sum_{n^{\prime}=0}^{d-1}e^{2\pi
ikn^{\prime}/d}\langle n^{\prime}|\right)  \left(  \frac{1}{\sqrt{d}}%
\sum_{n=0}^{d-1}e^{- in\phi}|n\rangle\right)  \right\vert ^{2}\\
&  =\frac{1}{d^{2}}\left\vert \sum_{n^{\prime},n=0}^{d-1}e^{2\pi ikn^{\prime
}/d}e^{- i n\phi}\langle n^{\prime}|n\rangle\right\vert ^{2}\\
&  =\frac{1}{d^{2}}\left\vert \sum_{n=0}^{d-1}\exp\!\left(   in\left(
\frac{2\pi k}{d}-\phi\right)  \right)  \right\vert ^{2}\\
&  =\frac{1}{d^{2}}\left\vert \frac{1-\exp\!\left(   i\left(  2\pi k-d\phi
\right)  \right)  }{1-\exp\!\left(  i\left(  \frac{ 2\pi k}{d}-\phi\right)
\right)  }\right\vert ^{2}\\
&  =\frac{1}{d^{2}} \frac{2\left(1-\cos(     2\pi k-d\phi
  )  \right)}{2\left(1-\cos\!\left(    \frac{ 2\pi k}{d}-\phi\right) \right) }\\
&  = 
\frac{1}{d^{2}}\frac{\sin^{2}\!\left(  \pi k-\frac{d\phi}{2}  \right)}{\sin^{2}\!\left(  \frac{\pi k}{d}-\frac{\phi}{2} \right)  }\, ,
\end{align}
thus justifying~\eqref{eq:phot-num-fourier-dist}.

The lemma below rigorously justifies the convergence statement asserted in~\eqref{eq:convergence-photon-numb-fourier-meas}.

\begin{lemma} \label{pns_convergence_lemma}
Let $p:[-\pi,\pi]\to \mathbb{R}_+$ be a continuous non-negative function with $p(-\pi)=p(\pi)$ and $\int_{-\pi}^\pi d\phi\ p(\phi) = 1$. For every positive integer $d$, all $k\in \{0,1,\ldots, d-1\}$, and all $\hat{\phi}\in [-\pi,\pi]$, set
\begin{equation}
q_{d}(\hat{\phi}|k) \coloneqq \Pi_d\!\left( \hat{\phi} - \frac{2\pi k}{d} \,\,\mathrm{mod}\,\, 2\pi \right) ,
\end{equation}
where $\Pi_d$ is defined by~\eqref{eq:Pi_d}, and
\begin{equation}
x \,\,\mathrm{mod}\,\, 2\pi \coloneqq \min\left\{ x+2\pi k:\, x+2\pi k\geq 0,\, k\in \mathbb{Z} \right\}.    
\end{equation}
Then the function $p'_d:[-\pi,\pi]\to \mathbb{R}_+$ defined by
\begin{equation}
p'_d(\hat{\phi}) \coloneqq \sum_{k=0}^{d-1} q_d(\hat{\phi}|k)\, \langle u_{k}\vert \mathcal{D}_p\left( |+_{d}\rangle\!\langle +_{d}| \right) \vert u_k\rangle ,
\end{equation}
where $\mathcal{D}_p$ is the bosonic dephasing channel given by~\eqref{eq:bdc}, satisfies 
\begin{equation}
p'_d \underset{d\to\infty}{\longrightarrow} p
\label{eq:uniform_convergence_p_prime}
\end{equation}
uniformly on $[-\pi,\pi]$, and furthermore
\begin{equation}
\lim_{d\to\infty} \int_{-\pi}^\pi d\phi\, \left\vert p(\phi) - p'_d(\phi)\right\vert = 0 .
\label{eq:L1_convergence_p_prime}
\end{equation}
\end{lemma}

\begin{proof}
We start by observing that, due to the calculation in the first part of this appendix,
\begin{align}
p'_d(\hat{\phi}) &\overset{\text{(i)}}{=} \frac1d \sum_{k=0}^{d-1} \Pi_d\!\left( \hat{\phi} - \frac{2\pi k}{d} \,\,\mathrm{mod}\,\, 2\pi \right) \notag \\
&\hspace{6ex} \times \int_{-\pi}^{\pi} d\phi\ p(\phi) \,F_d\!\left(\frac{2\pi k}{d} - \phi\right) \\
&\overset{\text{(ii)}}{=} \int_{-\pi}^{\pi} \frac{d\phi}{2\pi}\ p(\phi) \,F_d\!\left(\frac{2\pi}{d}\! \floor{\frac{d\hat{\phi}'}{2\pi}} - \phi\right) \\
&\overset{\text{(iii)}}{=} \int_{-\pi}^{\pi} \frac{d\phi}{2\pi}\ p(\phi) \,F_d\!\left(\frac{2\pi}{d}\! \floor{\frac{d\hat{\phi}}{2\pi}} - \phi\right) \\
&\overset{\text{(iv)}}{=} (p\star F_d)\!\left(\frac{2\pi}{d}\! \floor{\frac{d\hat{\phi}}{2\pi}}\right)
\end{align}
Here, in~(i) we introduced the Fej\'er kernel $F_d(x)\coloneqq \frac{\sin^2(dx/2)}{d \sin^2(x/2)}$, in~(ii) we observed that the only nonzero term in the sum is for $k = \floor{\frac{d\hat{\phi}'}{2\pi}}$, where $\hat{\phi}' \coloneqq \hat{\phi} \,\,\mathrm{mod}\,\, 2\pi$; indeed, since changing $k\mapsto k+1$ displaces the point $\hat{\phi} - \frac{2\pi k}{d}$ by exactly $-2\pi/d$, and the function $\Pi_d$ is nonzero in an interval of length precisely equal to $2\pi/d$, there can be only one nonzero term in the sum; using that $\hat{\phi}' = \hat{\phi} + 2\pi p$, where $p\in \{0,1\}$, we can also verify that
\begin{equation}
\hat{\phi} - \frac{2\pi}{d} \floor{\frac{d\hat{\phi}'}{2\pi}} = \hat{\phi} - \frac{2\pi}{d} \floor{\frac{d\hat{\phi}}{2\pi}} - 2\pi p .
\end{equation}
Using $x-1<\floor{x}\leq x$, we now note that
\begin{equation}
0\leq \hat{\phi} - \frac{2\pi}{d} \floor{\frac{d\hat{\phi}}{2\pi}} < \hat{\phi} - \frac{2\pi}{d} \left( \frac{d\hat{\phi}}{2\pi} - 1 \right) = \frac{2\pi}{d} ,
\end{equation}
implying that indeed
\begin{align}
&\Pi_d\!\left(\hat{\phi} - \frac{2\pi}{d} \floor{\frac{d\bar{\hat{\phi}}}{2\pi}} \,\,\mathrm{mod}\,\,2\pi\right) \notag \\
&\qquad = \Pi_d\left(\hat{\phi} - \frac{2\pi}{d} \floor{\frac{d\hat{\phi}}{2\pi}}\right)\\
&\qquad = \frac{d}{2\pi} ;
\end{align}
continuing with the justification of the first chain of identities, in~(iii) we used the periodicity of $F_d$ to substitute~$\hat{\phi}'$ with $\hat{\phi}$, and finally in~(iv) we introduced the notation 
\begin{equation}
(p\star F_d)(\xi)\coloneqq \int_{-\pi}^{\pi} \frac{d\phi}{2\pi}\ p(\phi) \,F_d(\xi - \phi).
\end{equation}
Now, 
calling $\tilde{p}$ the periodic extension of $p$ to the whole real line, for all $\xi\in \mathbb{R}$ one sees that
\begin{align}
(p\star F_d)(\xi) &= \int_{-\pi+\xi}^{\pi+\xi} \frac{d\theta}{2\pi}\ p(\xi - \theta) \,F_d\!\left(\theta\right) \\
&= \int_{-\pi}^{\pi} \frac{d\theta}{2\pi}\ \tilde{p}(\xi - \theta) \,F_d\!\left(\theta\right) .
\end{align}
Note that since $p$ is continuous on the compact set $[-\pi,\pi]$, it is also uniformly continuous. Due to the fact that $p(-\pi) = p(\pi)$, its extension $\tilde{p}$ can also be shown to be uniformly continuous. Let $\omega$ be the modulus of continuity of $\tilde{p}$. This means that $\omega:[0,\infty) \to [0,\infty)$ is a non-decreasing continuous function, with $\omega(0)=0$, such that for all $\xi, \xi'\in \mathbb{R}$ it holds that
\begin{equation}
\left| \tilde{p}(\xi) - \tilde{p}(\xi')\right| \leq \omega(|\xi-\xi'|) .
\end{equation}
Now, for $\xi,\xi'\in \mathbb{R}$ we can write that
\begin{align}
&\left| (p\star F_d)(\xi) - (p\star F_d)(\xi') \right| \notag \\
&\qquad = \left| \int_{-\pi}^{\pi} \frac{d\theta}{2\pi}\ \big(\tilde{p}(\xi - \theta) - \tilde{p}(\xi' - \theta)\big) \,F_d(\theta) \right| \\
&\qquad \leq \int_{-\pi}^{\pi} \frac{d\theta}{2\pi}\ \big|\tilde{p}(\xi - \theta) - \tilde{p}(\xi' - \theta)\big| \,F_d(\theta) \\
&\qquad \leq \omega\big(|\xi-\xi'|\big) \int_{-\pi}^{\pi} \frac{d\theta}{2\pi}\ F_d(\theta) \\
&\qquad = \omega\big(|\xi-\xi'|\big) ,
\end{align}
where in the last line we leveraged the fact that $\int_{-\pi}^{\pi} \frac{d\theta}{2\pi}\ F_d(\theta) = 1$ for all $d$. In other words, also $p\star F_d$ is uniformly continuous, and it has the same modulus of continuity as $\tilde{p}$.

We are finally ready to put everything together and prove the first half of the claim. We write that
\begin{align}
&\left| p'_d(\hat{\phi}) - p(\hat{\phi})\right| \notag \\
&\quad = \left| (p\star F_d)\!\left(\frac{2\pi}{d}\! \floor{\frac{d\hat{\phi}}{2\pi}}\right) - p(\hat{\phi})\right| \\
&\quad \leq \left| (p\star F_d)\!\left(\frac{2\pi}{d}\! \floor{\frac{d\hat{\phi}}{2\pi}}\right) - (p\star F_d)\!\left(\hat{\phi}\right)\right| \notag \\
&\quad \qquad + \left| (p\star F_d)\!\left(\hat{\phi}\right) - p(\hat{\phi})\right| \\
&\quad \leq \omega\!\left( \left| \frac{2\pi}{d}\! \floor{\frac{d\hat{\phi}}{2\pi}} - \hat{\phi} \right| \right) + \left| (p\star F_d)\!\left(\hat{\phi}\right) - p(\hat{\phi})\right| \\
&\quad \leq \omega(2\pi/d) + \left| (p\star F_d)\!\left(\hat{\phi}\right) - p(\hat{\phi})\right| ,
\end{align}
where in the last line we noted that $\hat{\phi}\leq \frac{2\pi}{d}\! \floor{\frac{d\hat{\phi}}{2\pi}} \leq \hat{\phi} + \frac{2\pi}{d}$, because of the elementary properties of the floor function. Now, since $\lim_{d\to\infty} \omega(2\pi/d) = 0$, to establish~\eqref{eq:uniform_convergence_p_prime} we only need to check that $p\star F_d$ converges uniformly to $p$ as $d\to\infty$; and this is well known to follow from the continuity of $p$, due to Fej\'er's theorem~\cite[Theorem~3.4]{ZYGMUND}.

To deduce~\eqref{eq:L1_convergence_p_prime} from~\eqref{eq:uniform_convergence_p_prime} it suffices to note that 
\begin{equation}
\int_{-\pi}^\pi d\phi\, \left\vert p(\phi) - p'_d(\phi)\right\vert \leq 2\pi \sup_\phi \left\vert p(\phi) - p'_d(\phi)\right\vert ,
\end{equation}
and the right-hand side tends to $0$ as $d\to\infty$ due to~\eqref{eq:uniform_convergence_p_prime}.
\end{proof}

\section{Calculations for coherent-state  method}

\label{app:coh-state-method}

Let us first justify the equality in~\eqref{eq:coh_rot}. 
After the phase rotation $e^{- i\hat{n}\phi}$ acts, the state becomes
\begin{align}
e^{- i\hat{n}\phi}|\alpha\rangle &  =e^{-\frac{1}{2}\left\vert
\alpha\right\vert ^{2}}\sum_{n=0}^{\infty}\frac{\alpha^{n}}{\sqrt{n!}}e^{-
i\hat{n}\phi}|n\rangle\\
&  =e^{-\frac{1}{2}\left\vert \alpha\right\vert ^{2}}\sum_{n=0}^{\infty}%
\frac{\alpha^{n}}{\sqrt{n!}}e^{- in\phi}|n\rangle\\
&  =e^{-\frac{1}{2}\left\vert \alpha\right\vert ^{2}}\sum_{n=0}^{\infty}%
\frac{\left(  \alpha e^{- i\phi}\right)  ^{n}}{\sqrt{n!}}|n\rangle\\
&  =|\alpha e^{- i\phi}\rangle.
\end{align}

After performing heterodyne detection, and as discussed in the main text, we compute the argument of~$\beta$ as the
estimate of $\phi$, i.e.,
$
\hat{\phi}\coloneqq \operatorname{arg}(\beta)$. 
The induced probability density function for $\hat{\phi}$ is known as
the Rician phase distribution (see~\cite[Eqs.~(10) \& (20)]{LZJ20}). In particular,
we can model the random process by which $\hat{\phi}$ is generated as
being like that in~\cite[Eq.~(3)]{LZJ20}, given by
\begin{align}
\beta=\alpha\exp(- i\phi)+n,
\end{align}
where $n$ is a complex Gaussian random variable $\mathcal{CN}(0,1)$ (such that
the variance for each of the real and imaginary parts is 1/2, i.e.,
$\sigma^{2}=1/2$, using the notation of~\cite[Eq.~(3)]{LZJ20}). We can restrict
$\alpha$ to be a positive real number, and in this case, we have that
$A=\alpha$ and $B=-1$, using the notation of~\cite[Eq.~(3)]{LZJ20}. Following
\cite[Eqs.~(10) \& (20)]{LZJ20}, we find that the probability density $p_{\alpha}(\hat{\phi}|\phi)$ for $\hat{\phi} \in [-\pi,\pi]$ is given by
\begin{multline}
p_{\alpha}(\hat{\phi}|\phi)\coloneqq \frac{e^{-\left\vert \alpha\right\vert ^{2}%
}}{2\pi}
+\frac{1}{2}\frac{\left\vert \alpha\right\vert }{\sqrt{\pi}}\cos
(\hat{\phi}-\phi )e^{-\left\vert \alpha\right\vert ^{2}\sin
^{2}(\hat{\phi}-\phi)}\times \\
\left[  1+\operatorname{erf}(\left\vert
\alpha\right\vert \cos(\hat{\phi}-\phi ))\right]  .
\end{multline}
We now show that this probability density converges to a Dirac delta at $\phi$ in the limit as $\alpha\rightarrow\infty$, in the sense stated in~\eqref{eq:coh-state-sch-converge}.

\begin{lemma} \label{coherent_state_method_convergence_lemma}
Let $p:[-\pi,\pi]\to \mathbb{R}_+$ be a continuous non-negative function with $p(-\pi)=p(\pi)$ and $\int_{-\pi}^\pi d\phi\ p(\phi) = 1$. For all $\alpha\geq 0$, let $p_\alpha(\hat{\phi}|\phi)$ be the Rician probability distribution defined by~\eqref{eq:Rician}, and set
\begin{equation}
p'_\alpha(\hat{\phi}) \coloneqq \int_{-\pi}^\pi d\phi\ p(\phi)\, p_\alpha(\hat{\phi}|\phi) .
\label{eq:conditional-rician}
\end{equation}
Then
\begin{equation}
p'_\alpha \underset{\alpha\to\infty}{\longrightarrow} p
\label{eq:pointwise_convergence_coherent_state_method}
\end{equation}
pointwise on $[-\pi,\pi]$, and furthermore
\begin{equation}
\lim_{\alpha\to\infty} \int_{-\pi}^\pi d\phi\, \left\vert p(\phi) - p'_\alpha(\phi)\right\vert = 0 .
\label{eq:L1_convergence_coherent_state_method}
\end{equation}
\end{lemma}

\begin{proof}
For this proof it is  ideal to work with the integral representation of the Rician probability density given in~\eqref{eq:integral-rep-rician}; namely,
\begin{equation}
p_\alpha(\hat{\phi}|\phi) = \int_0^\infty db\ \frac{b}{\pi}\, e^{-|\alpha - b e^{-i(\hat{\phi}-\phi)}|^2} .
\label{eq:integral-rep-rician-other}
\end{equation}
Substituting \eqref{eq:integral-rep-rician-other} into \eqref{eq:conditional-rician}, we now have that
\begin{align}
p'_\alpha(\hat{\phi}) &\overset{\text{(i)}}{=} \int_{-\pi}^{\pi} d\phi\ \tilde{p}(\phi)\int_0^\infty db\ \frac{b}{\pi}\, e^{-|\alpha - b e^{-i(\hat{\phi}-\phi)}|^2} \\
&\overset{\text{(ii)}}{=} \int_{\mathbb{C}} \frac{d^2\gamma}{\pi}\ \tilde{p}\big(\hat{\phi} + \mathrm{arg}(\gamma)\big)\, e^{-|\alpha - \gamma|^2} \\
&\overset{\text{(iii)}}{=} \int_{\mathbb{C}} \frac{d^2 z}{\pi}\ \tilde{p}\big(\hat{\phi} + \mathrm{arg}(z+\alpha)\big)\, e^{-|z|^2}. 
\end{align}
The justification of the above chain of identities is as follows: in~(i) we used the above integral representation of the Rician probability density and introduced the periodic extension $\tilde{p}$ of $p$ to the whole real line; in~(ii) we used Fubini's theorem, changing variables to $\gamma \coloneqq b e^{-i(\hat{\phi}-\phi)}$; in~(iii) we changed variables again, setting $z\coloneqq \gamma - \alpha$.

Now, since $\alpha\geq 0$ is real, we have that
\begin{equation}
    \mathrm{arg}(z+\alpha) = \arctan\!\left(\frac{z_I}{\alpha + z_R}\right)
\end{equation}
up to multiples of $2\pi$, and hence 
\begin{equation}
    \lim_{\alpha\to\infty} \tilde{p}\big(\hat{\phi} + \mathrm{arg}(z+\alpha)\big) = \tilde{p}(\hat{\phi})
\end{equation}
due to the continuity of $\tilde{p}$ and the $\arctan$ function. We can then write 
\begin{align}
&\lim_{\alpha\to\infty} \int_{\mathbb{C}} \frac{d^2 z}{\pi}\ \tilde{p}\big(\hat{\phi} + \mathrm{arg}(z+\alpha)\big)\, e^{-|z|^2} \notag \\
&\qquad \hspace{-2.5pt}\overset{\text{(iv)}}{=} \int_{\mathbb{C}} \frac{d^2 z}{\pi}\ \lim_{\alpha\to\infty} \tilde{p}\big(\hat{\phi} + \mathrm{arg}(z+\alpha)\big)\, e^{-|z|^2} \\
&\qquad = \int_{\mathbb{C}} \frac{d^2 z}{\pi}\ \tilde{p}(\hat{\phi})\, e^{-|z|^2} \\
&\qquad = \tilde{p}(\hat{\phi}) = p(\hat{\phi}) ,
\end{align}
where in~(iv) we employed Lebesgue's dominated convergence theorem, which is applicable because, due to its periodicity and continuity, $\tilde{p}$ is a bounded function, which means that $\tilde{p}\big(\hat{\phi} + \mathrm{arg}(z+\alpha)\big)\, e^{-|z|^2}\leq M\, e^{-|z|^2}$ for some constant $M > 0$, and the right-hand side is an absolutely integrable function of $z$. This completes the proof of~\eqref{eq:pointwise_convergence_coherent_state_method}.

To deduce~\eqref{eq:L1_convergence_coherent_state_method}, we first note that if $\tilde{p}(\xi)\leq M$ for all $\xi\in \mathbb{R}$, then also
\begin{align}
p'_\alpha(\phi) &= \int_{\mathbb{C}} \frac{d^2 z}{\pi}\ \tilde{p}\big(\phi + \mathrm{arg}(z+\alpha)\big)\, e^{-|z|^2} \\
&\leq M \int_{\mathbb{C}} \frac{d^2 z}{\pi}\ e^{-|z|^2} = M ,
\end{align}
implying that $|p'_\alpha(\phi) - p(\phi)|\leq M$ for all $\phi \in [-\pi,\pi]$ and again by Lebesgue's dominated convergence (since we integrate on a finite-measure space)
\begin{align}
&\lim_{\alpha\to\infty} \int_{-\pi}^\pi d\phi\, \left\vert p(\phi) - p'_\alpha(\phi)\right\vert \notag \\
&\qquad = \int_{-\pi}^\pi d\phi\, \lim_{\alpha\to\infty} \left\vert p(\phi) - p'_\alpha(\phi)\right\vert \\
&\qquad = 0 .
\end{align}
This establishes~\eqref{eq:L1_convergence_coherent_state_method} and thereby concludes the proof.
\end{proof}

\section{Derivation of Rician phase probability density function}

\label{app:rician-phase-deriv}

Here, for completeness, we provide a derivation of the Rician phase probability density function.
Consider that the probability density for obtaining the outcome $\beta
\in\mathbb{C}$ when performing heterodyne detection on a coherent state $|\alpha\rangle$, where 
$\alpha\in\mathbb{C}$, is as follows \cite[Eqs.~(4.7) \& (5.122)]{serafini2017quantum}:%
\begin{equation}
p(\beta|\alpha)=\frac{1}{\pi}e^{-\left\vert \alpha-\beta\right\vert ^{2}}.
\end{equation}
Letting $\alpha=re^{-i\phi}$ and $\beta=be^{-i\hat{\phi}}$, with $r,b\geq0$
and $\phi,\hat{\phi}\in\left[  -\pi,\pi\right]  $, we find that
\begin{align}
& p(\beta|\alpha)\ d^{2}\beta \notag \\
& =\frac{1}{\pi}e^{-\left\vert \alpha
-\beta\right\vert ^{2}}\ d^{2}\beta\\
& =\frac{1}{\pi}\exp\!\left(  -\left\vert re^{-i\phi}-be^{-i\hat{\phi}%
}\right\vert ^{2}\right)  \ b\ db\ d\hat{\phi}.
\end{align}
Then we obtain the marginal probability density for the phase $\hat{\phi}$ by
integrating over the magnitude $b$:%
\begin{align}
p(\hat{\phi}|r,\phi)  & =\int_{0}^{\infty}db\ \frac{b}{\pi}\exp\!\left(
-\left\vert re^{-i\phi}-be^{-i\hat{\phi}}\right\vert ^{2}\right)  \\
& =\int_{0}^{\infty}db\ \frac{b}{\pi}\exp\!\left(  -\left\vert r-be^{-i\left(
\hat{\phi}-\phi\right)  }\right\vert ^{2}\right)  .
\label{eq:integral-rep-rician}
\end{align}
Considering that%
\begin{align}
& \left\vert r-be^{-i\left(  \hat{\phi}-\phi\right)  }\right\vert ^{2}  \notag \\ 
&
=r^{2}-2rb\cos(\hat{\phi}-\phi)+b^{2}\\
& =\left(  r\sin(\hat{\phi}-\phi)\right)  ^{2}+\left(  b-r\cos(\hat{\phi}%
-\phi)\right)  ^{2},
\end{align}
we find that%
\begin{align}
& p(\hat{\phi}|r,\phi)\nonumber\\
& =\frac{e^{-\left(  r\sin(\hat{\phi}-\phi)\right)  ^{2}}}{\pi}\int
_{0}^{\infty}db\ b\ e^{-\left(  b-r\cos(\hat{\phi}-\phi)\right)  ^{2}}\\
& =\frac{e^{-\left(  r\sin(\hat{\phi}-\phi)\right)  ^{2}}}{\pi}\int
_{-r\cos(\hat{\phi}-\phi)}^{\infty}d\bar{b}\ \left(  \bar{b}+r\cos(\hat{\phi
}-\phi)\right)  \ e^{-\bar{b}^{2}}.
\end{align}
Now consider that%
\begin{align}
& \int_{-r\cos(\hat{\phi}-\phi)}^{\infty}d\bar{b}\ \left(  \bar{b}+r\cos
(\hat{\phi}-\phi)\right)  \ e^{-\bar{b}^{2}}\nonumber\\
& =\int_{-r\cos(\hat{\phi}-\phi)}^{\infty}d\bar{b}\ \bar{b}\ e^{-\bar{b}^{2}%
}\nonumber\\
& \qquad+r\cos(\hat{\phi}-\phi)\int_{-r\cos(\hat{\phi}-\phi)}^{\infty}d\bar
{b}\ e^{-\bar{b}^{2}}\\
& =-\frac{1}{2}\int_{-r\cos(\hat{\phi}-\phi)}^{\infty}\frac{d}{d\bar{b}%
}\ e^{-\bar{b}^{2}}\nonumber\\
& \qquad+r\cos(\hat{\phi}-\phi)\frac{\sqrt{\pi}}{2}\left(
1+\operatorname{erf}(r\cos(\hat{\phi}-\phi))\right)  \\
& =\frac{1}{2}e^{-\left(  r\cos(\hat{\phi}-\phi)\right)  ^{2}}\nonumber\\
& \qquad+r\cos(\hat{\phi}-\phi)\frac{\sqrt{\pi}}{2}\left(
1+\operatorname{erf}(r\cos(\hat{\phi}-\phi))\right)  .
\end{align}
In the above, we made use of the error function
\begin{equation}
    \operatorname{erf}(x) \coloneqq \frac{2}{\sqrt{\pi}} \int_0^x dt \, e^{-t^2}  ,
\end{equation}
and some of its properties: $\operatorname{erf}(+\infty)=1$ and $\operatorname{erf}(x) = -\operatorname{erf}(-x)$.
Thus, we finally conclude that
\begin{align}
& p(\hat{\phi}|r,\phi)
 =\frac{e^{-\left(  r\sin(\hat{\phi}-\phi)\right)  ^{2}}}{2\pi} \Bigg(
e^{-\left(  r\cos(\hat{\phi}-\phi)\right)  ^{2}} \notag 
\\
& \qquad +r\cos(\hat{\phi}-\phi)\sqrt{\pi}\left(  1+\operatorname{erf}(r\cos(\hat{\phi
}-\phi))\right)
\Bigg)  \\
& =\frac{e^{-r^{2}}}{2\pi}+ \notag \\
& \quad \frac{e^{-\left(  r\sin(\hat{\phi}-\phi)\right)
^{2}}}{2\sqrt{\pi}}r\cos(\hat{\phi}-\phi)\left(  1+\operatorname{erf}%
(r\cos(\hat{\phi}-\phi))\right)  .
\end{align}

\section{Other scenarios: Strong converse exponent, error exponent, multiple channel discrimination, and antidistinguishability}

\label{app:other-scens}

In this appendix, we discuss various other scenarios to which our results
apply. The first two are known as the strong converse exponent and error
exponent, which also go by the names Han--Kobayashi~\cite{HK89} and Hoeffding~\cite{Hoeffding65},
respectively. These are settings related to binary hypothesis testing. The
other two scenarios are multiple channel discrimination and antidistinguishability.

\subsection{Strong converse exponent}

The non-asymptotic strong converse exponent for channel discrimination is
defined for $r>0$ as follows:%
\begin{multline}
H_{n}(r,\mathcal{N}_{0},\mathcal{N}_{1})\coloneqq \\
\inf_{\mathcal{A}}\left\{  -\frac
{1}{n}\ln\left(  1-\alpha_{n}(\mathcal{A})\right)  :\beta_{n}(\mathcal{A})\leq
e^{-rn}\right\}  ,
\end{multline}
where $\alpha_{n}(\mathcal{A})$ and $\beta_{n}(\mathcal{A})$ are defined in
\eqref{eq:ch-disc-err-probs}. By applying the same reasoning as given in Sections~\ref{sec:optimality} and~\ref{sec:attainability}, we conclude
for BDCs $\mathcal{D}_{p}$ and $\mathcal{D}_{q}$ that%
\begin{equation}
H_{n}(r,\mathcal{D}_{p},\mathcal{D}_{q})=\inf_{t}\left\{  -\frac{1}{n}%
\ln\left(  1-\alpha_{n}(t)\right)  :\beta_{n}(t)\leq e^{-rn}\right\}  ,
\end{equation}
where $\alpha_{n}(t)$ and $\beta_{n}(t)$ are defined in~\eqref{eq:cl-disc-err-prob-a}--\eqref{eq:cl-disc-err-prob-b} and taken with
respect to the probability densities $p$ and $q$ defining $\mathcal{D}_{p}$
and $\mathcal{D}_{q}$, respectively. By taking the $n\rightarrow\infty$ limit
and applying the classical result of~\cite{HK89}, we conclude that%
\begin{equation}
\lim_{n\rightarrow\infty}H_{n}(r,\mathcal{D}_{p},\mathcal{D}_{q})=\sup
_{\alpha>1}\frac{\alpha-1}{\alpha}\left(  r-D_{\alpha}(p\Vert q)\right)  ,
\end{equation}
where the R\'enyi relative entropy $D_{\alpha}(p\Vert q)$ is defined for
$\alpha\in\left(  0,1\right)  \cup\left(  1,\infty\right)  $ as%
\begin{equation}
D_{\alpha}(p\Vert q)\coloneqq \frac{1}{\alpha-1}\ln\int_{-\pi}^{\pi}d\phi\ p^{\alpha
}(\phi)\ q^{1-\alpha}(\phi).\label{eq:renyi-rel-ent}%
\end{equation}

\subsection{Error exponent}

The non-asymptotic error exponent for channel discrimination is defined for
$r>0$ as follows:%
\begin{equation}
B_{n}(r,\mathcal{N}_{0},\mathcal{N}_{1})\coloneqq \sup_{\mathcal{A}}\left\{  -\frac
{1}{n}\ln\alpha_{n}(\mathcal{A}):\beta_{n}(\mathcal{A})\leq e^{-rn}\right\}  ,
\end{equation}
where $\alpha_{n}(\mathcal{A})$ and $\beta_{n}(\mathcal{A})$ are defined in
\eqref{eq:ch-disc-err-probs}. By applying the same reasoning as given in Sections~\ref{sec:optimality} and~\ref{sec:attainability}, we conclude
for BDCs $\mathcal{D}_{p}$ and $\mathcal{D}_{q}$ that
\begin{equation}
B_{n}(r,\mathcal{D}_{p},\mathcal{D}_{q})=\sup_{t}\left\{  -\frac{1}{n}%
\ln\alpha_{n}(t):\beta_{n}(t)\leq e^{-rn}\right\}  ,
\end{equation}
where $\alpha_{n}(t)$ and $\beta_{n}(t)$ are defined in~\eqref{eq:cl-disc-err-prob-a}--\eqref{eq:cl-disc-err-prob-b} and taken with
respect to the probability densities $p$ and $q$ defining $\mathcal{D}_{p}$
and $\mathcal{D}_{q}$, respectively. By taking the $n\rightarrow\infty$ limit
and applying the classical result of~\cite{Hoeffding65}, we conclude that
\begin{equation}
\lim_{n\rightarrow\infty}B_{n}(r,\mathcal{D}_{p},\mathcal{D}_{q})=\sup
_{\alpha\in(0,1)}\frac{\alpha-1}{\alpha}\left(  r-D_{\alpha}(p\Vert q)\right)
,
\end{equation}
where the R\'enyi relative entropy $D_{\alpha}(p\Vert q)$ is defined in~\eqref{eq:renyi-rel-ent}.

\subsection{Multiple channel discrimination}

The goal of multiple channel discrimination is to decide which channel has
been chosen from a tuple of channels. More formally, let $\left(
\mathcal{N}_{i}\right)  _{i=1}^{\ell}$ be a tuple of channels. Then an
adaptive protocol for channel discrimination consists of an adaptive strategy
of the form discussed previously in Section~\ref{sec:QCD}, with the only difference
being that the final measurement is $\mathcal{Q}\coloneqq \left(  Q_{i}\right)
_{i=1}^{\ell}$. Letting $\rho_{i}^{(n)}$ be the final state of such a protocol
when the $i$th channel has been selected, the success probability of multiple
channel discrimination is%
\begin{equation}
p_{n}^{s}(\left(  \mathcal{N}_{i}\right)  _{i=1}^{\ell})\coloneqq \sup_{\mathcal{A}%
}\sum_{i=1}^{\ell}\lambda_{i}\operatorname{Tr}\!\left[Q_{i}\rho_{i}^{(n)}\right],
\end{equation}
where $\lambda_{i}$ is the prior probability that channel $\mathcal{N}_{i}$ is
selected. (Thus, the following constraints apply:$\ \lambda_{i}\geq0$ for all
$i\in\left\{  1,\ldots,\ell\right\}  $ and $\sum_{i=1}^{\ell}\lambda_{i}=1$).

Now let us consider classical multiple hypothesis testing. Let $(p_{i}%
)_{i=1}^{\ell}$ be a tuple of probability densities. Here the goal is to
observe a sample $\phi^{n}\equiv\left(  \phi_{1},\ldots,\phi_{n}\right)  $
from one of the product densities (i.e., of the form $p_{i}^{\otimes n}$) and
decide the value of $i$ (i.e., which density generated the sample sequence).
The success probability is given by%
\begin{equation}
p_{n}^{s}((p_{i})_{i=1}^{\ell})\coloneqq \sup_{t}\sum_{i=1}^{\ell}\lambda_{i}\int
d\phi^{n}\ t(i|\phi^{n})\ p_{i}^{\otimes n}(\phi^{n}),
\end{equation}
where $\lambda_{i}$ is a prior probability and $(t(i|\phi^{n}))_{i=1}^{\ell}$
is a conditional probability distribution (i.e., satisfying $t(i|\phi^{n}%
)\geq0$ for all $i\in\left\{  1,\ldots,\ell\right\}  $ and $\sum_{i=1}^{\ell
}t(i|\phi^{n})=1$).

By the same reasoning from Sections~\ref{sec:optimality} and~\ref{sec:attainability}, our main result here is that%
\begin{equation}
p^{s}_n(\left(  \mathcal{D}_{p_{i}}\right)  _{i=1}^{\ell})=p^{s}_n((p_{i}%
)_{i=1}^{\ell}),\label{eq:mult-ch-disc-result}%
\end{equation}
where $\left(  \mathcal{D}_{p_{i}}\right)  _{i=1}^{\ell}$ is a tuple of
bosonic dephasing channels defined by the corresponding tuple $(p_{i}%
)_{i=1}^{\ell}$ of probability densities.
By employing the known result \cite{salikhov1973asymptotic} (see also \cite[Theorem~4.2]{torgersen1981measures} and   \cite{Leang1997,salikhov1999one,salikhov2003optimal}) that the asymptotic error exponent for multiple hypothesis testing is equal to the minimum pairwise Chernoff divergence, we conclude the following:
\begin{equation}
    \lim_{n\to \infty} -\frac{1}{n}\ln \!\left(1-p^{s}_n(\left(  \mathcal{D}_{p_{i}}\right)  _{i=1}^{\ell})\right) = \min_{i \neq j} C(p_i \Vert p_j),
\end{equation}
where the Chernoff divergence $C(p_i \Vert p_j)$ is defined from~\eqref{eq:Chernoff-div}.

\subsection{Antidistinguishability}

The problem of antidistinguishability has the same structure as multiple
channel discrimination, but the goal is the opposite. That is, the goal is to
decide which channel was not selected. That is, if the $i$th channel is
selected, the goal is to report back \textquotedblleft not $i$\textquotedblright.\ We can thus adopt all of the notation from the previous
section, but the error probability for the antidistinguishability problem is
given by%
\begin{equation}
p_{n}^{e}(\left(  \mathcal{N}_{i}\right)  _{i=1}^{\ell})\coloneqq \inf_{\mathcal{A}%
}\sum_{i=1}^{\ell}\lambda_{i}\operatorname{Tr}\!\left[Q_{i}\rho_{i}^{(n)}\right].
\end{equation}
Similarly, for the classical antidistinguishability problem, the error
probability is given by%
\begin{equation}
p_{n}^{e}((p_{i})_{i=1}^{\ell})\coloneqq \inf_{t}\sum_{i=1}^{\ell}\lambda_{i}\int
d\phi^{n}\ t(i|\phi^{n})\ p_{i}^{\otimes n}(\phi^{n}).
\end{equation}
Thus, the main difference with multiple hypothesis testing mathematically is to minimize the objective
functions rather than maximize them. By the same reasoning from Sections~\ref{sec:optimality} and~\ref{sec:attainability},
we conclude that%
\begin{equation}
p_{n}^{e}(\left(  \mathcal{D}_{p_{i}}\right)  _{i=1}^{\ell})=p_{n}^{e}%
((p_{i})_{i=1}^{\ell}),\label{eq:antidist-1-shot}%
\end{equation}
where $\left(  \mathcal{D}_{p_{i}}\right)  _{i=1}^{\ell}$ is a tuple of
bosonic dephasing channels defined by the corresponding tuple $(p_{i}%
)_{i=1}^{\ell}$ of probability densities.

As shown recently in
\cite{mishra2023optimal}, there is a solution for
the asymptotic error exponent of antidistinguishability. Namely, the following limit holds:%
\begin{equation}
\lim_{n\rightarrow\infty}-\frac{1}{n}\ln p_{n}^{e}((p_{i})_{i=1}^{\ell
})=C((p_{i})_{i=1}^{\ell}),\label{eq:asymp-anti}%
\end{equation}
where the multivariate Chernoff divergence $C((p_{i})_{i=1}^{\ell})$ is
defined as%
\begin{equation}
C((p_{i})_{i=1}^{\ell})\coloneqq -\ln\inf_{s}\int d\phi\ \prod\limits_{i=1}^{\ell
}p_{i}^{s_{i}}(\phi),
\end{equation}
with the optimization over $s\coloneqq \left(  s_{i}\right)  _{i=1}^{\ell}$, a
probability vector (satisfying $s_{i}\geq0$ for all $i\in\left\{
1,\ldots,\ell\right\}  $ and $\sum_{i=1}^{\ell}s_{i}=1$). Combining
\eqref{eq:antidist-1-shot} and~\eqref{eq:asymp-anti}, we conclude that%
\begin{equation}
\lim_{n\rightarrow\infty}-\frac{1}{n}\ln p_{n}^{e}(\left(  \mathcal{D}_{p_{i}%
}\right)  _{i=1}^{\ell})=C((p_{i})_{i=1}^{\ell}),
\end{equation}
where $\left(  \mathcal{D}_{p_{i}}\right)  _{i=1}^{\ell}$ is a tuple of
bosonic dephasing channels defined by the corresponding tuple $(p_{i}%
)_{i=1}^{\ell}$ of probability densities. 

\section{Multimode bosonic dephasing channels}

\label{app:multimode}

In this appendix, we briefly argue how all of our results apply to multimode BDCs as well. Recall from~\cite{lami2023exact} that a multimode BDC is defined as
\begin{equation}
    \mathcal{D}_p^{(m)}(\rho) \coloneqq \int_{[-\pi,\pi]^m}  d^m \phi\ p(\phi)\, e^{-i \sum_j \hat{n}_{j}  \phi_j} \rho\, e^{i\sum_j \hat{n}_{j}  \phi_j} ,
\label{eq:multimode-BDC}
\end{equation}
where $\phi \coloneqq (\phi_1, \ldots, \phi_m)$ is a vector of $m$ phases and $p(\phi)$ is a joint probability density function. 

This claim holds because the same arguments used in Section~\ref{sec:optimality} for optimality and in Section~\ref{sec:attainability} for attainability go through. Indeed, the channel $\mathcal{D}_p^{(m)}$ can be decomposed similarly to~\eqref{eq:deph-ch-decompose}, as
\begin{equation}
    \mathcal{D}_p^{(m)} = \mathcal{G}^{(m)} \circ \mathcal{F}^{(m)}_p,
\end{equation}
where
\begin{align}
    \mathcal{F}^{(m)}_p( \rho) & \coloneqq \rho \otimes \sigma_p^{(m)}, \\
    \sigma_p^{(m)} & \coloneqq  \int_{[-\pi,\pi]^m}  d^m \phi\ p(\phi) |\phi\rangle\!\langle \phi |, \\
    |\phi\rangle & \coloneqq |\phi_1 \rangle \otimes \cdots \otimes |\phi_m \rangle,
\end{align}
\begin{multline}
    \mathcal{G}^{(m)}(\rho' \otimes \rho'')  \coloneqq \\
    \int_{[-\pi,\pi]^m}  d^m \phi\ 
    e^{-i \sum_j \hat{n}_{j}  \phi_j} \rho'\, e^{i\sum_j \hat{n}_{j}  \phi_j} \, \operatorname{Tr}[|\phi\rangle\!\langle \phi |\rho''].
\end{multline}
As such, the same simulation argument as before applies, with all distinguishability or estimation tasks limited by the distinguishability or estimability of the classical environmental states of the form $\sigma_p^{(m)}$. The attainability part follows because one can simply employ a tensor product of the strategies considered in Section~\ref{sec:attainability}. That is, for all $i \in \{1,\ldots, m\}$, suppose that $(\rho^i_\nu)_{\nu \in \mathbb{N}}$ is a sequence of states and $(M_{\phi_i})_{\phi_i}$ is a corresponding measurement such that~\eqref{eq:att-limit-bdc} and~\eqref{eq:p_nu-dens} hold for a single-mode BDC. Then it follows immediately that
\begin{equation}
    p(\phi) = \lim_{\nu \to \infty} p_\nu (\phi),
\end{equation}
for all $\phi= (\phi_1, \ldots, \phi_m)$, where
\begin{equation}
    p_\nu (\phi) \coloneqq \operatorname{Tr}\!\left[ \bigotimes_{i=1}^m M_{\phi_i} \mathcal{D}_p^{(m)}\left(\bigotimes_{i=1}^m \rho^i_{\nu}\right)\right  ].
\end{equation}

Let us finally note that, generalizing the statements in Section~\ref{sec:deph-plus-loss}, all of these results for multimode BDCs hold even when a set of multimode BDCs are affected by a common multimode pure-loss channel.

\section{Generalization to multiparameter channel estimation}

\label{app:multiparam}

Our channel estimation results generalize to the setting of multiparameter channel estimation. This follows by considering a cost function of multiple parameters, which results in a risk function. Then optimizing over all adaptive strategies leads to the same risk function evaluated on the underlying probability density of the BDC. The reasoning here is essentially the same as that used for all other conclusions in our paper: all adaptive strategies for estimation are particular estimation strategies on the underlying probability densities, and one can obtain samples from these underlying densities by employing either of the schemes in Sections~\ref{sec:phot-num-sup-scheme} or \ref{sec:coh-state-scheme}.

\section{Photon-number-superposition method in the presence of loss}

\label{app:lossy-fejer}

In this appendix, we calculate the probability distribution that results when
using the photon-number-superposition method if there is photon loss in
addition to the action of a phase rotation. By modeling the pure-loss channel
as a beamsplitter interaction between the input and an environment mode in the
vacuum state \cite{serafini2017quantum}, we find that the input state transforms as follows:%
\begin{equation}
|+_{d}\rangle\rightarrow\frac{1}{\sqrt{d}}\sum_{n=0}^{d-1}\sum_{\ell=0}%
^{n}\sqrt{\eta^{n-\ell}\left(  1-\eta\right)  ^{\ell}\binom{n}{\ell}}%
|n-\ell\rangle|\ell\rangle .
\end{equation}
After the action of the phase rotation $e^{-i\hat{n}\phi} \otimes I$, the state becomes
\begin{multline}
|\psi(\phi,\eta)\rangle \coloneqq \\
\frac{1}{\sqrt{d}}\sum_{n=0}^{d-1}\sum_{\ell=0}
^{n}e^{-i\left(  n-\ell\right)  \phi}\sqrt{\eta^{n-\ell}\left(  1-\eta\right)
^{\ell}\binom{n}{\ell}}|n-\ell\rangle|\ell\rangle.
\end{multline}
The probability for obtaining outcome $k$ after measuring in the Fourier basis
$\left\{  |u_{k}\rangle\right\}  _{k=0}^{d-1}$ (defined in \eqref{eq:fourier-basis-states}) is then calculated according to
the Born rule as follows:
\begin{widetext}
\begin{align}
 \left\Vert \left(  \langle u_{k}|\otimes I\right)  |\psi(\phi,\eta
)\rangle\right\Vert _{2}^{2} 
& =\left\Vert
\begin{array}
[c]{c}%
\left(  \frac{1}{\sqrt{d}}\sum_{m=0}^{d-1}e^{2\pi imk/d}\langle m|\otimes
I\right)  \times\\
\left(  \frac{1}{\sqrt{d}}\sum_{n=0}^{d-1}\sum_{\ell=0}^{n}e^{-i\left(
n-\ell\right)  \phi}\sqrt{\eta^{n-\ell}\left(  1-\eta\right)  ^{\ell}\binom
{n}{\ell}}|n-\ell\rangle|\ell\rangle\right)
\end{array}
\right\Vert _{2}^{2}\\
& =\frac{1}{d^{2}}\left\Vert \sum_{n=0}^{d-1}\sum_{\ell=0}^{n}e^{2\pi i\left(
n-\ell\right)  k/d}e^{-i\left(  n-\ell\right)  \phi}\sqrt{\eta^{n-\ell}\left(
1-\eta\right)  ^{\ell}\binom{n}{\ell}}|\ell\rangle\right\Vert _{2}^{2}\\
& =\frac{1}{d^{2}}\left\Vert \sum_{n=0}^{d-1}\sum_{\ell=0}^{n}e^{i\left(
n-\ell\right)  \left(  2\pi k/d-\phi\right)  }\sqrt{\eta^{n-\ell}\left(
1-\eta\right)  ^{\ell}\binom{n}{\ell}}|\ell\rangle\right\Vert _{2}^{2}\\
& =\frac{1}{d^{2}}\left(  \sum_{m=0}^{d-1}\sum_{\ell^{\prime}=0}%
^{m}e^{-i\left(  m-\ell^{\prime}\right)  \left(  2\pi k/d-\phi\right)  }%
\sqrt{\eta^{m-\ell^{\prime}}\left(  1-\eta\right)  ^{\ell^{\prime}}\binom
{m}{\ell^{\prime}}}\langle\ell^{\prime}|\right)  \nonumber\\
& \qquad\times\left(  \sum_{n=0}^{d-1}\sum_{\ell=0}^{n}e^{i\left(
n-\ell\right)  \left(  2\pi k/d-\phi\right)  }\sqrt{\eta^{n-\ell}\left(
1-\eta\right)  ^{\ell}\binom{n}{\ell}}|\ell\rangle\right)  \\
& =\frac{1}{d^{2}}\sum_{m,n=0}^{d-1}\sum_{\ell=0}^{\min\left\{  n,m\right\}
}e^{-i\left(  m-n\right)  \left(  2\pi k/d-\phi\right)  }\eta^{\left(
m+n\right)  /2-\ell}\left(  1-\eta\right)  ^{\ell}\sqrt{\binom{m}{\ell}%
\binom{n}{\ell}}.
\label{eq:lossy-fejer}
\end{align}
\end{widetext}
We recover the probability distribution in \eqref{eq:phot-num-fourier-dist} in the limit as
$\eta\rightarrow1$. Numerical experiments indicate that the distribution in \eqref{eq:lossy-fejer} is highly peaked around $\phi$ for fixed transmissivity $\eta \in (0,1)$ as $d$ becomes larger.

\end{document}